\documentclass[11pt,a4paper]{article}

\usepackage{amsmath}
\usepackage{amsthm}
\usepackage{amssymb}
\usepackage{amsfonts}
\usepackage{graphicx}
\usepackage{makeidx}
\usepackage{authblk}

\makeindex

\usepackage{latexsym}
\usepackage{graphicx}
\usepackage{marvosym}
\usepackage{hyperref}
\usepackage{float}
\usepackage{color}
\def\lie{{\cal L}}

\def\tr{\mbox{tr}}
\def\hi{\mathcal{H}}
\def\o{\mathcal{O}}
\def\pu{\partial_{u}}
\def\pv{\partial_{v}}
\def\s{\mathcal{S}}

\newcommand{\li}{\mbox{$\lie \mkern-9.5mu /$}}

\def\f12{\frac 1 2}

\def\a{\alpha}
\def\b{\beta}

\def\hh{\mathcal{H}}
\def\vh{\mathcal{V}_{\hh}}
\def\m{\mathcal{M}}

\def\h{\mathcal{H}}

\newcommand{\nabb}{\mbox{$\nabla \mkern-13mu /$\,}}
\newcommand{\lapp}{\mbox{$\triangle \mkern-13mu /$\,}}
\newcommand{\epsi}{\mbox{$\epsilon \mkern-7.4mu /$\,}}
\newcommand{\gi}{\mbox{$g \mkern-8.8mu /$\,}}
\newcommand{\di}{\mbox{$d \mkern-9.2mu /$\,}}
\newcommand{\divv}{\mbox{$\div \mkern-16mu /$\,\,}}

\def\f12{\frac 1 2}

\def\div{\text{div}}

\newtheorem{definition}{Definition}[section]

\newtheorem{remark}{Remark}[section]
\newtheorem{lemma}{Lemma}[section]
\newtheorem{theorem}{Theorem}[section]
\newtheorem{proposition}{Proposition}[section]

\newtheorem{corollary}{Corollary}[section]

\newtheorem{mytheo}{Theorem}

\begin{document}
\title{The characteristic gluing problem and conservation laws  for the wave equation on null hypersurfaces}

\author[1,2]{Stefanos Aretakis}
\affil[1]{ Princeton University, Department of Mathematics, Fine Hall, NJ 08544, USA}
\affil[2]{Institute for Advanced Study, Einstein Drive, Princeton, NJ 08540, USA}
\normalsize

\date{October 4, 2013}

\maketitle

\begin{abstract}

We obtain necessary and sufficient conditions for the existence of ``conservation laws''  on null hypersurfaces  for the wave equation on general four-dimensional Lorentzian manifolds. 
Examples of null hypersurfaces exhibiting such conservation laws include the standard null cones of Minkowski spacetime and the degenerate horizons of extremal black holes. Another (limiting) example of such a conservation law is that which gives rise to the well-known Newman--Penrose constants along the null infinity of asymptotically flat spacetimes. The existence of such conservation laws can be viewed as an obstruction to a certain gluing construction for characteristic initial data for the wave equation. We initiate the general study of the latter gluing problem and show that the existence of conservation laws is in fact the only obstruction. Our method relies on a novel elliptic structure  associated to a foliation with 2-spheres of a null hypersurface.

\end{abstract}

\tableofcontents

\section{Introduction}
\label{sec:Introduction}

This paper will address the question of existence of \textit{conserved charges  on null hypersurfaces} (and their associated \textit{conservation laws}) for the wave equation
\begin{equation}
\Box_{g}\psi=0
\label{wave}
\end{equation}
on a general four-dimensional Lorentzian manifold $(\mathcal{M},g)$. 

The simplest example of such a conserved charge arises in Minkowski space and is given by \[char\left[\psi;S_{v}\right]=\int_{S_{v}}\frac{1}{r^{2}}\partial_{u}(r\psi) \]
 which for all solutions of \eqref{wave} satisfies the conservation law 
\[\partial_{v}\Big(char\left[\psi;S_{v}\right]\Big)=0.\]
 Here $u,v$ are standard null coordinates and $S_{v}$ are the spherical sections of the standard null cones $\left\{u=c\right\}$.   Another example is the recently discovered conservation laws on the degenerate event horizons of extremal black hole spacetimes (see \cite{aretakis4, hj2012, murata2012}). A third (limiting) example are the celebrated \textit{Newman--Penrose constants} which are conserved along  null infinity in any asymptotically flat spacetime (see \cite{np1,np2}).

In the present paper, we define a general notion of conserved charges (see Section \ref{sec:ConservationLawsForTheWaveEquations}) encompassing all the above examples and give a characterization of  null hypersurfaces admitting such charges in terms of the kernel of an elliptic operator (defined for the first time here and in our companion paper \cite{aretakiselliptic}).\footnote{It will follow in particular from this characterization that  generic Lorentzian manifolds do not admit such charges.}    In fact,
we show that the only information that can be propagated by all solutions to the wave equation along null hypersurfaces is given precisely by these conserved charges.  For this,  we initiate the general study of \textit{gluing constructions for the characteristic initial value problem} (see Section \ref{sec:TheCharacteristicInitialValueProblem}) and we show that the only obstruction to gluing along a null hypersurface $\hh$ is the existence of conserved charges (in our sense) on $\hh$. 

Part of the importance of the conservation laws on degenerate horizons referred to  above lies in their role in the instability properties of the wave equation on extremal black holes  (see  Section \ref{sec:Remarks}). This result led to the so-called ``horizon instability of extremal black holes''. The present general study may thus shed light on new aspects of the global evolution of the wave equation on more general backgrounds.

The statement of the main result can be found in Section \ref{sec:TheMainResultxx}. Our proof introduces a new method which we hope will be relevant for applications to other linear and non-linear equations such as the Einstein equations.

\subsection{Conservation laws for the wave equation}
\label{sec:ConservationLawsForTheWaveEquations}

We first present some basic geometric definitions that will be useful for defining the notion of conservation laws on null hypersurfaces. For more details about the geometric setting see Section \ref{sec:TheDoubleNullFoliation}; our notation follows  \cite{DC09,christab}. 

\paragraph{Null foliations\medskip \\}
\label{sec:nullfoliations}

Let $\hh$ be a regular null hypersurface of a four-dimensional Lorentzian manifold $(\m,g)$. A foliation  $\mathcal{S}=\big(S_{v}\big)_{v\in\mathbb{R}}$ of $\hh$, that is a collection of sections $S_{v}$ which vary smoothly in $v$ such that $\cup_{v}S_{v}=\hh$, can be uniquely determined by the choice of one section $S_{0}$, the choice of a smooth function $\Omega$ on $\hh$ and the choice of a null geodesic vector field $L_{geod}$ tangential to the null generators of $\hh$ and such that 
\[\nabla_{L_{geod}}L_{geod}=0.\]
Indeed, if we define the vector field 
\[L=\Omega^{2}\cdot L_{geod}\]
on $\hh$ and consider the affine parameter $v$ of $L$ such that 
\[Lv=1, \text{ with }v=0 \text{ on } S_{0},\]
then the level sets $S_{v}$ of $v$ on $\hh$ are precisely the leaves of the foliation $\mathcal{S}$. We use the notation
\begin{equation}
\mathcal{S}=\Big\langle S_{0},L_{geod}, \Omega\Big\rangle.
\label{foliation}
\end{equation}
 \begin{figure}[H]
   \centering
		\includegraphics[scale=0.06]{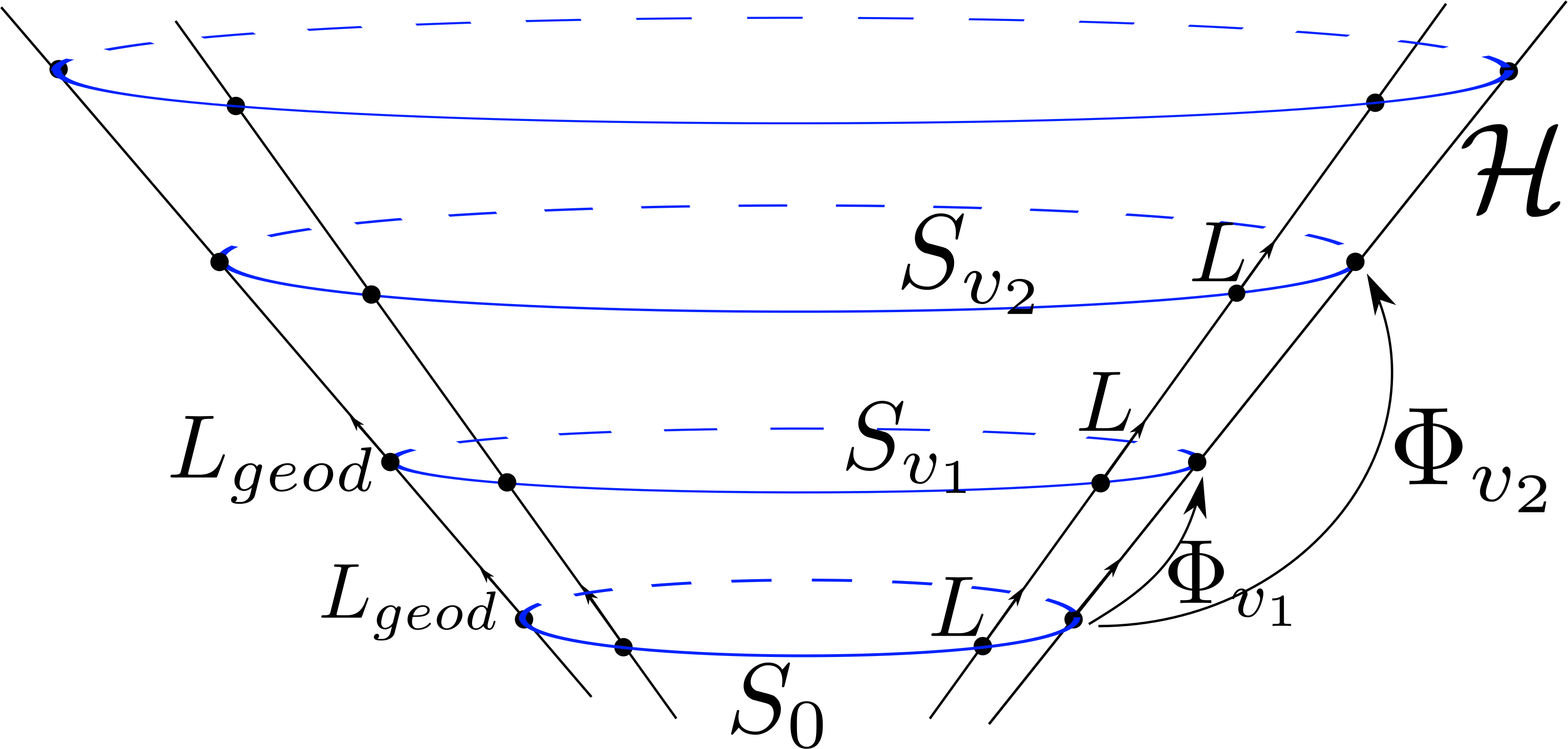}
	\label{fig:glue1}
\end{figure}
We will henceforth restrict to the case that all sections $S_{v}$ above are diffeomorphic (via a diffeomorphism $\Phi$) to the 2-sphere.\footnote{ Topologies with higher genus can be treated analogously. Our argument heavily relies  on the compactness of the sections and hence the non-compact case remains an open problem.} The flow of $L$ on $\hh$ provides a diffeomorphism $\Phi_{v}$ between the sections $S_{v}$ and $S_{0}$. In addition to the induced metric on $S_{v}$, which we will denote by $\gi$, we can also equip all sections with the standard metric on the unit sphere $\gi_{\mathbb{S}^{2}}$ (via $\Phi$) such that it is invariant under the flow of $L$. The volume form on $S_{v}$ with respect to $\gi_{\mathbb{S}^{2}}$ will be denoted by $d\mu_{_{\mathbb{S}^{2}}}$.

Given any section $S_{v}$, there is a unique metric $\hat{g}$ which is conformal to the induced metric $\gi$ such that the volume form $d\mu_{_{\hat{g}}}$ with respect to $\hat{g}$ and the volume form $d\mu_{_{\mathbb{S}^{2}}}$ with respect to $\gi_{\mathbb{S}^{2}}$ are equal:
\[d\mu_{_{\hat{g}}}=d\mu_{_{\mathbb{S}^{2}}}. \]We denote by $\phi$ the conformal factor:
\begin{equation}
\gi=\phi^{2}\cdot \hat{g}. 
\label{eq:theconformalfactorintroduction}
\end{equation}

Furthermore, given a foliation $\mathcal{S}$ we denote by $Y^{\s}$ the unique null vector field   which is normal to the sections $S_{v}$, conjugate to $\hh$ and normalized such that 
\begin{equation}
g\big(L_{geod},Y^{\s}\big)=-1.
\label{Y}
\end{equation}
 \begin{figure}[H]
   \centering
		\includegraphics[scale=0.06]{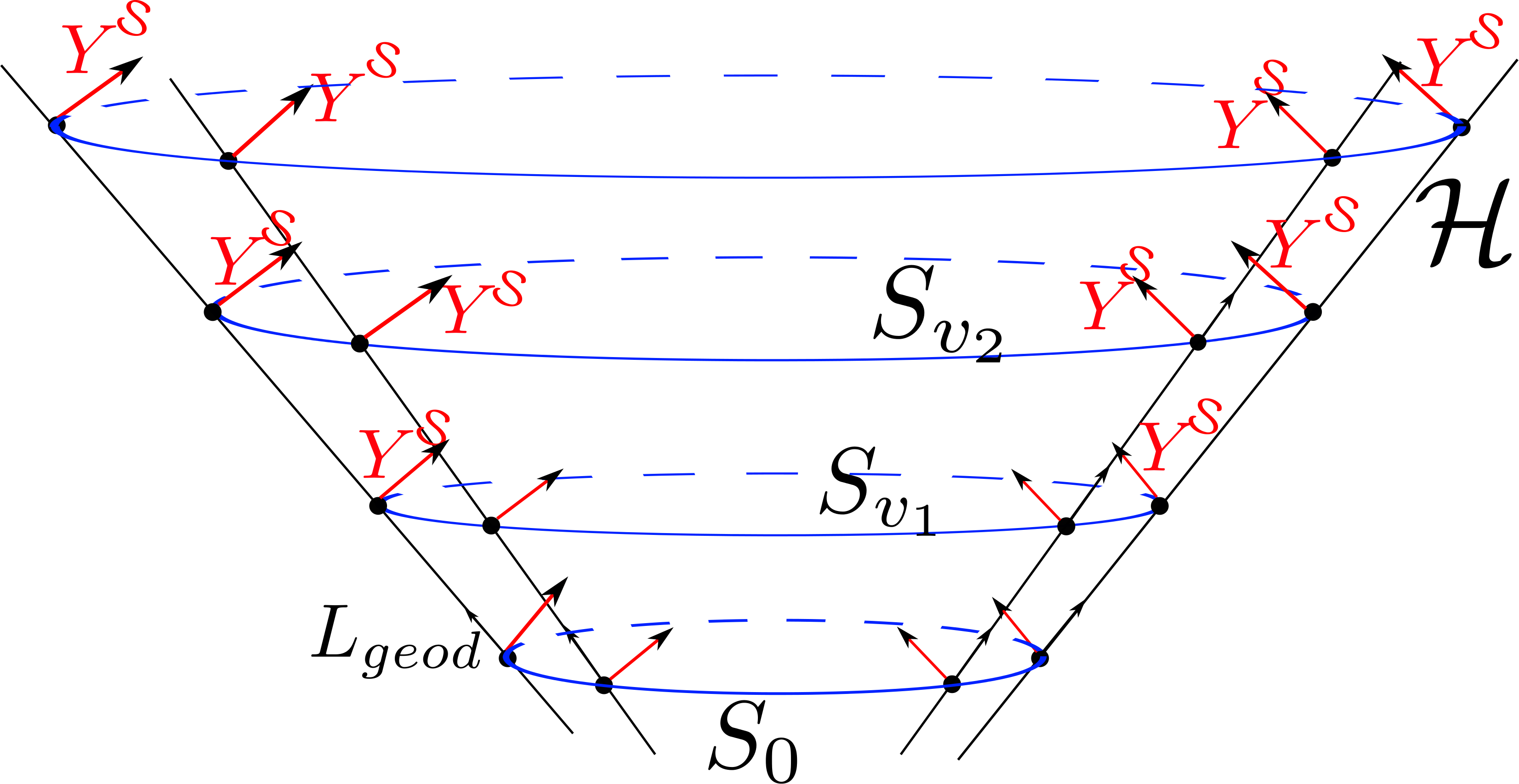}
	\label{fig:glue2}
\end{figure}
The vector field $Y^{\s}$ is transversal to $\hh$ and can be seen as the generator of an appropriately normalized ``retarded time'' $u$ such that $u=0$ on $\hh$. Specifically, one can construct an optical function $u$ such that $Y^{\s}u=1$ on $\hh$ and the level sets of $u$ are ``outgoing'' null hypersurfaces $\hh_{u}$ (hence we assume here that $\hh$ is an ``outgoing'' null hypersurface). Note that in a similar  fashion as above we can define the conformal factor $\phi$ of the section $S_{u,v}$ which are the intersections of $\hh_{u}$ and the ``incoming'' null hypersurfaces $\underline{\hh}_{v}$ generated by the null geodesics normal to $S_{v}$ and conjugate to $\hh$.
 \begin{figure}[H]
   \centering
		\includegraphics[scale=0.07]{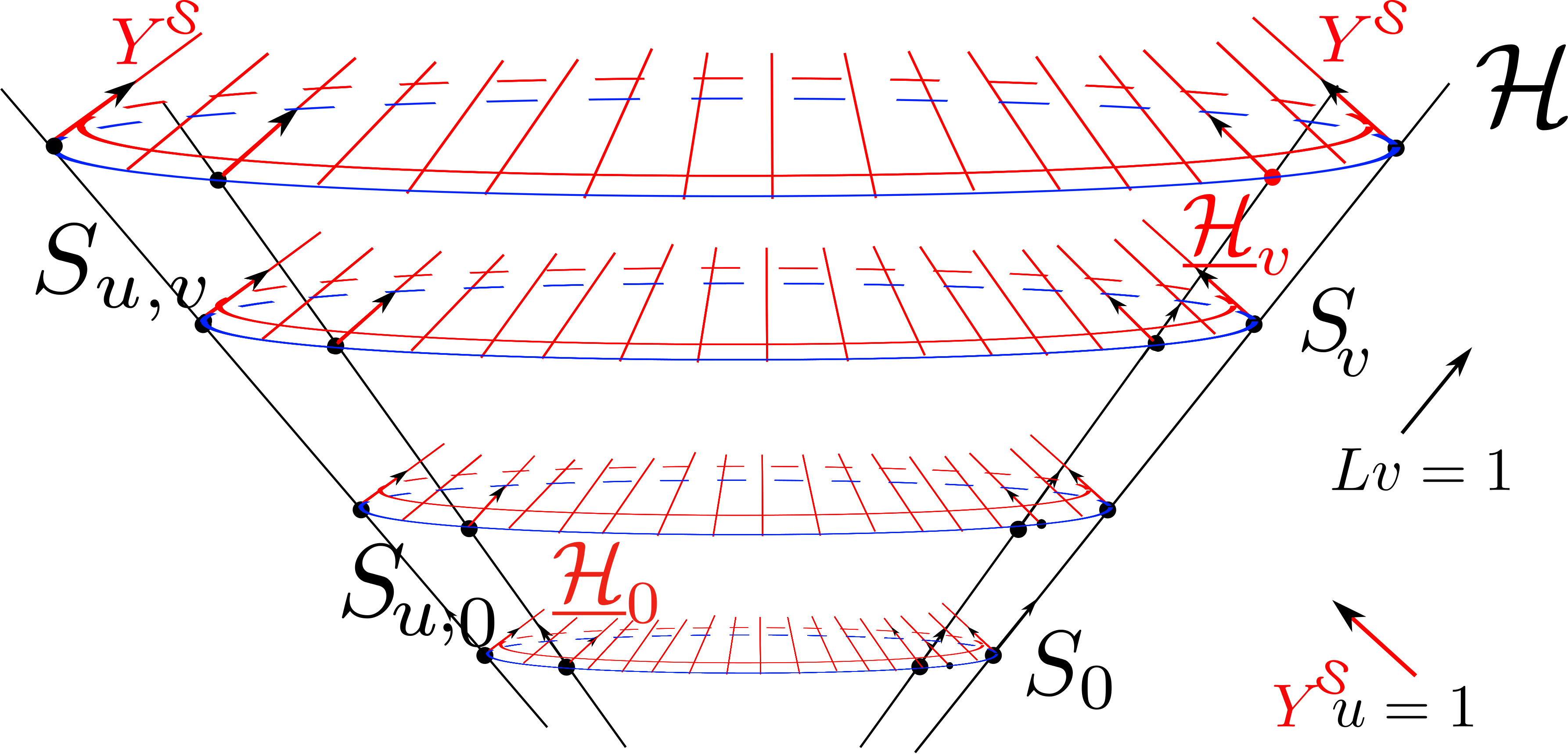}
	\label{fig:glue3}
\end{figure}

\newpage 

\paragraph{Conservation laws on $\hh$\medskip \\}
\label{sec:conservation laws}

Consider the linear space $\mathcal{V}_{\hh}$ consisting of all smooth functions on $\hh$ which are constant along the null generators of $\hh$, i.e.
\begin{equation}
\mathcal{V}_{\hh}=\Big\{f\in C^{\infty}(\hh)\, :\, Lf=0\Big\}.
\label{linearspace}
\end{equation}
Let $\mathcal{S}=\big(S_{v}\big)_{v\in\mathbb{R}}$ be a foliation of $\hh$ and let $Y^{\s}$ be the vector field and $\phi$ the conformal factor defined above.  
We define the linear space $\mathcal{W}^{\mathcal{S}}$ to be the subspace of $\mathcal{V}_{\hh}$ such that for all $\Theta^{\s}\in \mathcal{W}^{\mathcal{S}}$ and for all solutions $\psi$ to the wave equation \eqref{wave} the integrals
\begin{equation}
\int_{S_{v}}Y^{\s}\big(\phi\cdot\psi\big)\cdot\Theta^{\s} \, d\mu_{_{\mathbb{S}^{2}}}
\label{eq:integrals}
\end{equation}
are conserved, i.e.~independent of $v$. That is,
\begin{equation}
\mathcal{W}^{\s}=\left\{\Theta^{\s}\in C^{\infty}(\hh)\, :\, L\Theta^{\s}=0, \ \partial_{v}\left(\int_{S_{v}}Y^{\s}\big(\phi\cdot\psi\big)\cdot\Theta^{\s} \, d\mu_{_{\mathbb{S}^{2}}}\right)=0\right\}\subset \mathcal{V}_{\hh}. 
\label{eq:}
\end{equation}

 We make the following definition:
\begin{definition} \textbf{(Conservation laws on $\hh$)}:
We say that a null hypersurface $\hh$ admits (first order) conservation laws  with respect to a foliation $\mathcal{S}$ of $\hh$ if
\begin{equation}
\dim\mathcal{W}^{\mathcal{S}}\geq 1.
\label{dimensionofw}
\end{equation} 
If \eqref{dimensionofw} holds then we will refer to the space $\mathcal{W}^{\mathcal{S}}$ and the number $\dim \mathcal{W}^{\s}$ as the kernel and the dimension of the conservation laws, respectively. The integrals of the form \eqref{eq:integrals} will be called conserved charges and will be denoted by $char\big(S_{v}\big)[\psi; \Theta^{\s}]$.
\label{definitionconservationlaw}
\end{definition}
A priori Definition \ref{definitionconservationlaw} appears to be a very restrictive notion of conservation laws. However, as we shall show (see Theorem \ref{theoremmainintro}), the conservation laws in the sense of Definition \ref{definitionconservationlaw} are in fact the only type of ``first order'' conservation laws that a null hypersurface $\hh$ might admit.

 One could also define higher order conservation laws by considering higher derivatives of $\phi\cdot\psi$ in \eqref{eq:integrals}. In order to make our method clear, in the bulk of this paper we only consider first order conservation laws and for this reason we will simply refer to them as conservation laws (see however Section \ref{sec:Genericity} where we consider higher order conservation laws for spherically symmetric geometries).

\paragraph{Examples \\}
\label{sec:examples}

We next consider three main examples of spacetimes admitting conserved charges for the wave equation.

\paragraph{\small 1. Minkowski spacetime\medskip \\}
\label{sec:MinkowskiSpacetimea}

\normalsize

The wave equation in double null coordinates $(u,v)$ on the Minkowski spacetime reads
\[\partial_{v}\partial_{u}(r\psi)=\frac{1}{2r^{2}}\lapp_{\mathbb{S}^{2}}(r\psi),   \]
where $\lapp_{\mathbb{S}^{2}}$ is the standard Laplacian on the unit sphere. Hence, 
\[\int_{S_{v}}\pv\pu(r\psi)\, d\mu_{_{\mathbb{S}^{2}}}=\frac{1}{2r^{2}}\int_{S_{v}}\lapp_{\mathbb{S}^{2}}(r\psi)\, d\mu_{_{\mathbb{S}^{2}}}=0\]
and thus the integral 
\[\int_{S_{v}}\pu(r\psi)\, d\mu_{_{\mathbb{S}^{2}}}\, \]
is conserved along the null hypersurfaces $\left\{u=c\right\}$. This conserved charge can be written in  terms of Definition \ref{definitionconservationlaw}; indeed, if we consider the foliation $\s=\left\langle S_{0}, L_{geod}=\partial_{v},\Omega=1\right\rangle$ of $\hh$ then $\phi=r$, $Y^{\s}=\partial_{u}$ and  $\left\langle 1\right\rangle\subset \mathcal{W}^{\s}$ and hence $\dim \mathcal{W}^{\s}\geq 1$.

\paragraph{\small 2. Extremal black holes\medskip \\}
\label{sec:MinkowskiSpacetime4}

\normalsize

Consider the coordinate vector fields
\[T=\partial_{v}, \ \ \ \ \ R=\partial_{r}, \ \ \ \ \ \Phi=\partial_{\phi^{*}} \]
with respect to the ingoing Eddington--Finkelstein coordinates  $(v,r,\theta,\phi^{*})$ on an extremal Kerr black hole with mass parameter equal to $M$. Then, the quantity
\begin{equation}
\int_{S_{v}} \left[R\psi+\frac{\sin^{2}\theta}{4}\cdot T\psi+ \frac{1}{2M}\cdot\psi\right]\,d\mu_{_{\mathbb{S}^{2}}}
\label{eextremalcharge}
\end{equation}
is conserved along the horizon $\hh=\left\{r=M\right\}$, i.e.~it is independent of $v$. This conservation law was first found in \cite{aretakis4} and was then generalized to all extremal black holes by Lucietti and Reall \cite{hj2012} and Murata \cite{murata2012}. 

The above conserved charge can be written in  terms of  Definition \ref{definitionconservationlaw} as follows: Consider the foliation $\s=\left\langle S_{0}, L_{geod}=\partial_{v},\Omega=1\right\rangle$ of $\hh$. Then  \[Y^{\s}=\frac{1}{\frac{\sin^{2}\theta}{2}-1}\cdot\left[R+\frac{\sin^{2}\theta}{4}\cdot T+\frac{3+\cos^{2}\theta}{8M}\cdot\Phi\right]\] and 
\[\left.\phi\right|_{\hh}=\sqrt{2}\cdot M,\ \ \ \ \ \left.Y^{\s}\phi\right|_{\hh}=\frac{\sqrt{2}}{2}\cdot\frac{1}{\frac{\sin^{2}\theta}{2}-1}, \ \ \ \ \  \left\langle 
\frac{\sin^{2}\theta}{2}-1\right\rangle\subset \mathcal{W}^{\s} \]
and hence $\dim \mathcal{W}^{\s}\geq 1$. Note that since the integral curves of $\Phi$ are closed, the $\Phi$-derivative drops out from \eqref{eextremalcharge}. 

For more results based on this conservation law see \ref{sec:Remarks}.

\paragraph{\small 3. Null infinity of asymptotically flat spacetimes \medskip \\}
\label{sec:MinkowskiSpacetime}

\normalsize

 Given sufficient smoothness for $\psi$ at the null infinity $\mathcal{I}^{+}$ of an asymptotically flat spacetime we can write
\[\psi\big(u,r,\theta^{1},\theta^{2}\big)=\frac{\a_{1}\big(u,\theta^{1},\theta^{2}\big)}{r}+\frac{\a_{2}\big(u,\theta^{1},\theta^{2}\big)}{r^{2}}+O\left(\frac{1}{r^{3}}\right)\]
with respect to outgoing Eddington--Finkelstein coordinates $\big(u,r,\theta^{1},\theta^{2}\big)$. Here we identity $\mathcal{I}^{+}=\left\{r=\infty\right\}$. If $\psi$ is a solution to the wave equation, then the quantity
\begin{equation}
\lim_{r\rightarrow+\infty}\int_{S_{u}}r^{2}\cdot\partial_{r}(r\psi)\, d\mu_{_{\mathbb{S}^{2}}}
\label{eq:npc}
\end{equation}
does not depend on $u$. This charge as well as other charges involving higher order derivatives were found by Newman and Penrose \cite{np1,np2} (see also \cite{npexton}) and are known as Newman--Penrose constants.

The origin of these peculiar constants has been the object of intense study. In particular, we mention the work of Goldberg \cite{goldberg1,goldberg2} who showed that these constants do not arise from non-trivial transformation laws. The same author was able to rederive these constants in the flat case by using Green's theorem in appropriate regions in conjunction with the fundamental solution to the wave equation. See also the related work by Robinson \cite{robinson}. Further work on the Newman--Penrose constants can be found in \cite{chrugrav,goldberg3,pressnp,valientenp1, valientenp2,valientenp3} and references there-in. A nice geometric relation of the the Newman--Penrose constants and the  charges at the event horizon of extremal Kerr  was given by Bizon and Friedrich \cite{bizon2012}  and independently by Lucietti et al \cite{hm2012}.

The Newman--Penrose constants can be seen as a limiting example of the conserved charges given by Definition \ref{definitionconservationlaw} as follows: Let $\mathcal{I}_{S_{0}}$ be an incoming null hypersurface  and  $\s=\left\langle S_{0}, L_{geod}, \Omega=1\right\rangle$ be a foliation  of it such that $\frac{1}{2}\big(L_{geod}+Y^{\s}\big)\left.\!\right|_{S_{0}}$ is the (unit timelike) binormal of $S_{0}$ (see Section \ref{sec:TheNullInfinityMathcalI}). Let also $r=\sqrt{A/4\pi}$  be the area-radius function of the sections of $\s$ on $\mathcal{I}_{S_{0}}$, where $A$ is the area of the sections. Then $\mathcal{I}_{S_{0}}\rightarrow \mathcal{I}$ as $r\big(S_{0}\big)\rightarrow+\infty$ and the Newman--Penrose constants can be retrieved in the limit as $r\rightarrow +\infty$ by the conservation law of Definition \ref{definitionconservationlaw} if we take $Y^{\s}$ to be the limit of $Y^{\s}\left.\right|_{\mathcal{I}_{S_{0}}}$ and $\Theta^{S}$ to be the limit of $r^{2}\left.\right|_{\mathcal{I}_{S_{0}}}$ as $\mathcal{I}_{S_{0}}\rightarrow\mathcal{I}$ .

Our general theory (see Theorem \ref{theo3}) will in particular show that  the conserved charge \eqref{eq:npc} is the \textbf{only} non-trivial conserved  charge along the null infinity $\mathcal{I}$ which involves the 1-jet of $\psi$.

\subsection{The characteristic gluing problem}
\label{sec:TheCharacteristicInitialValueProblem}

The characteristic gluing problem for the wave equation  provides a means to formally show that Definition \ref{definitionconservationlaw} is the right notion of conservation laws on null hypersurfaces. We introduce this problem below.

Let $(\m,g)$ be a four-dimensional Lorentzian manifold and $\hh, \underline{\hh}$  be two regular null hypersurfaces intersecting at a two dimensional sphere $S_{0}$. Characteristic initial data for the wave equation \eqref{wave} correspond to prescribing the restriction of $\psi$ on  the union $\hh\cup \underline{\hh}$. 
In fact, given smooth data at $\mathcal{A},\underline{\mathcal{A}}$, as depicted below, there is a unique smooth solution to the wave equation in the domain of dependence $\mathcal{R}$, depicted schematically below:

 \begin{figure}[H]
   \centering
		\includegraphics[scale=0.09]{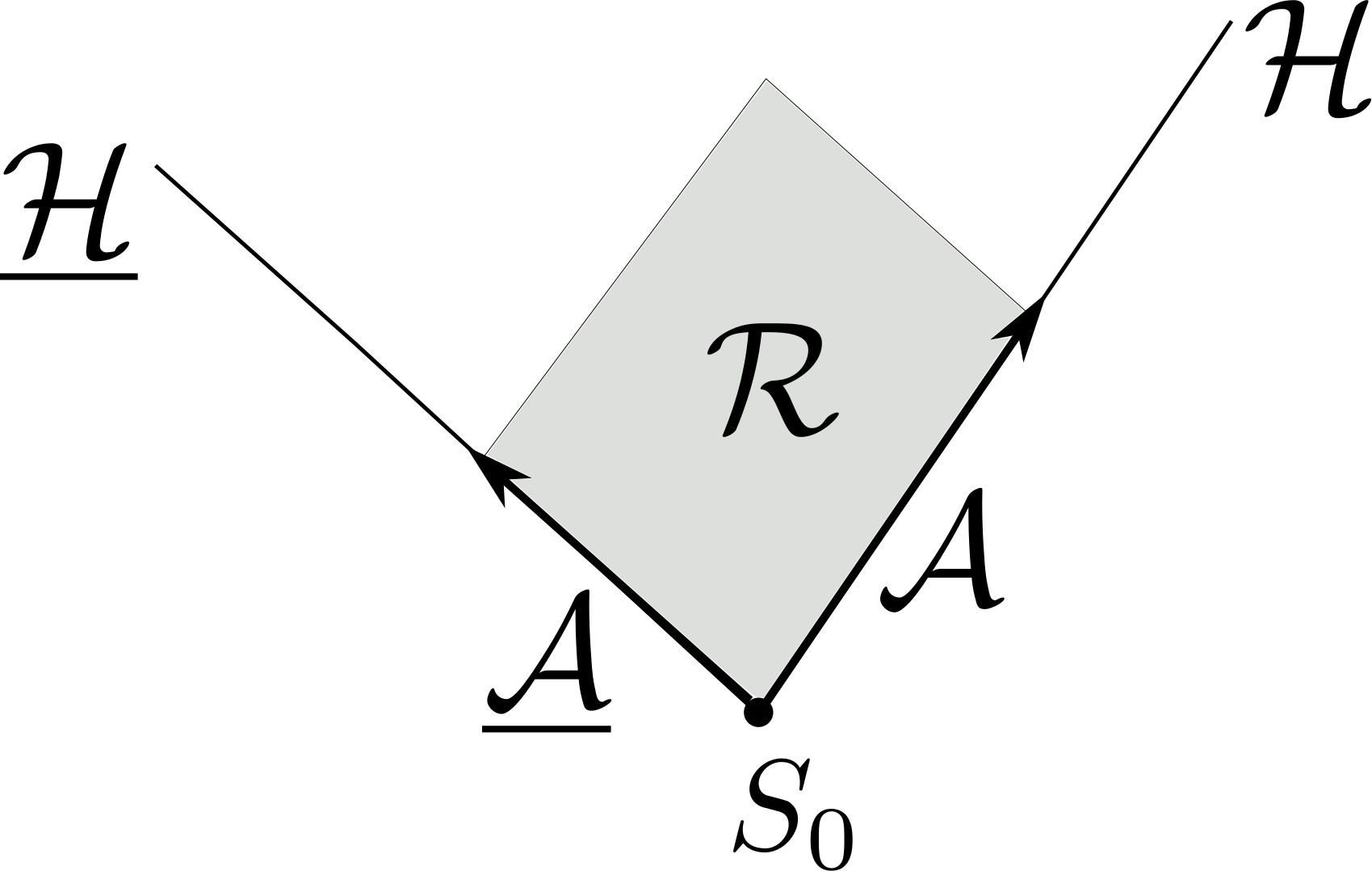}
	\label{fig:pfsdkfjapwoeijwe45}
\end{figure}

The problem of gluing constructions which we wish to formulate is the following: Consider a null hypersurface $\hh$ and two conjugate null hypersurfaces $\underline{\hh}_{0}$ and $\underline{\hh}_{1}$ intersecting $\hh$ at the two-dimensional spheres $S_{0}$ and $S_{1}$. We  prescribe initial data for the wave equation \eqref{wave} on the hypersurfaces $\mathcal{A}_{0},\underline{\mathcal{A}}_{0}$ and $\mathcal{A}_{1},\underline{\mathcal{A}}_{1}$ depicted in the figure below

 \begin{figure}[H]
   \centering
		\includegraphics[scale=0.129]{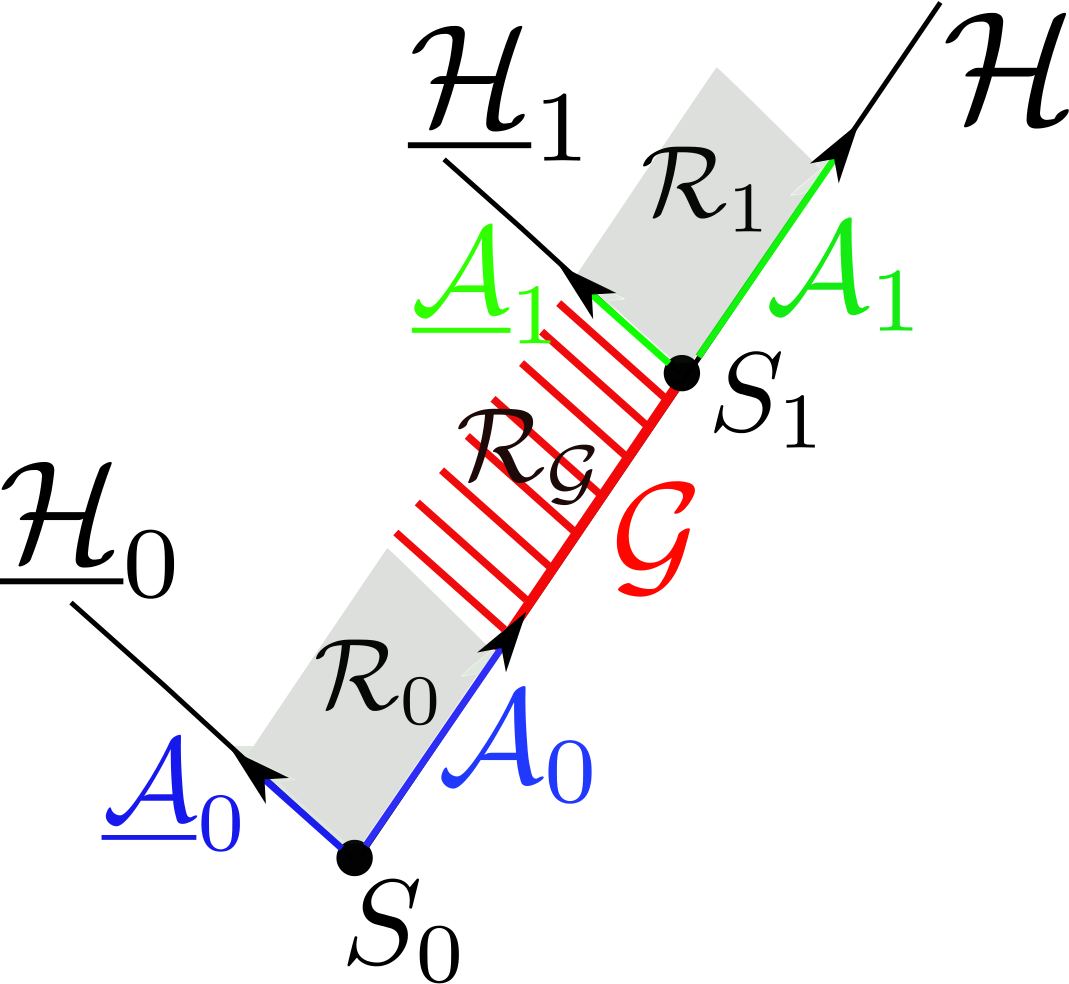}
	\label{figfdsp45picutre02}
\end{figure}

\noindent and we want to extend the data on the  truncated hypersurface $\mathcal{G}$ such that there is a smooth solution $\psi$ to the wave equation in the region $\mathcal{R}_{0}\cup\mathcal{R}_{\mathcal{G}}\cup\mathcal{R}_{1}$ such that $\left.\psi\right|_{\underline{A}_{0}\cup\underline{A}_{1}}$ and $\left.\psi\right|_{\mathcal{A}_{0}\cup\mathcal{A}_{1}}$ coincide with the prescribed data. 

In this paper we address in fact a weaker version of the above gluing problem, which is however sufficient for the complete classification of all null hypersurfaces admitting conservation laws. We make the following definition
\begin{definition} We shall say that ``we can perform first order gluing along $\hh$'' of the characteristic data $(\underline{\mathcal{A}}_{0},\mathcal{A}_{0})$, $(\underline{\mathcal{A}}_{1},\mathcal{A}_{1})$ as defined above if we can  smoothly extend the data in $\mathcal{G}$ such that the arising solutions

\begin{itemize}
	\item $\psi_{0}$ with data  given on $\underline{\mathcal{A}}_{0},\mathcal{A}_{0}\cup\mathcal{G}$, and 
	\item $\psi_{1}$ with data given on $\underline{\mathcal{A}}_{1},\mathcal{A}_{1}$
\end{itemize} agree at $S_{1}$ to all orders tangential to  $\hh$ and up to first order in directions transversal to $\hh$; that is, $\psi_{0}=\psi_{1}$ at $S_{1}$ to all orders tangential to  $\hh$ and $Y\psi_{0}=Y\psi_{1}$ at $S_{1}$, where $Y$ is a smooth vector field transversal to $\hh$. 
\label{firstordergluingdefinition}
\end{definition}

 \begin{figure}[H]
   \centering
		\includegraphics[scale=0.1]{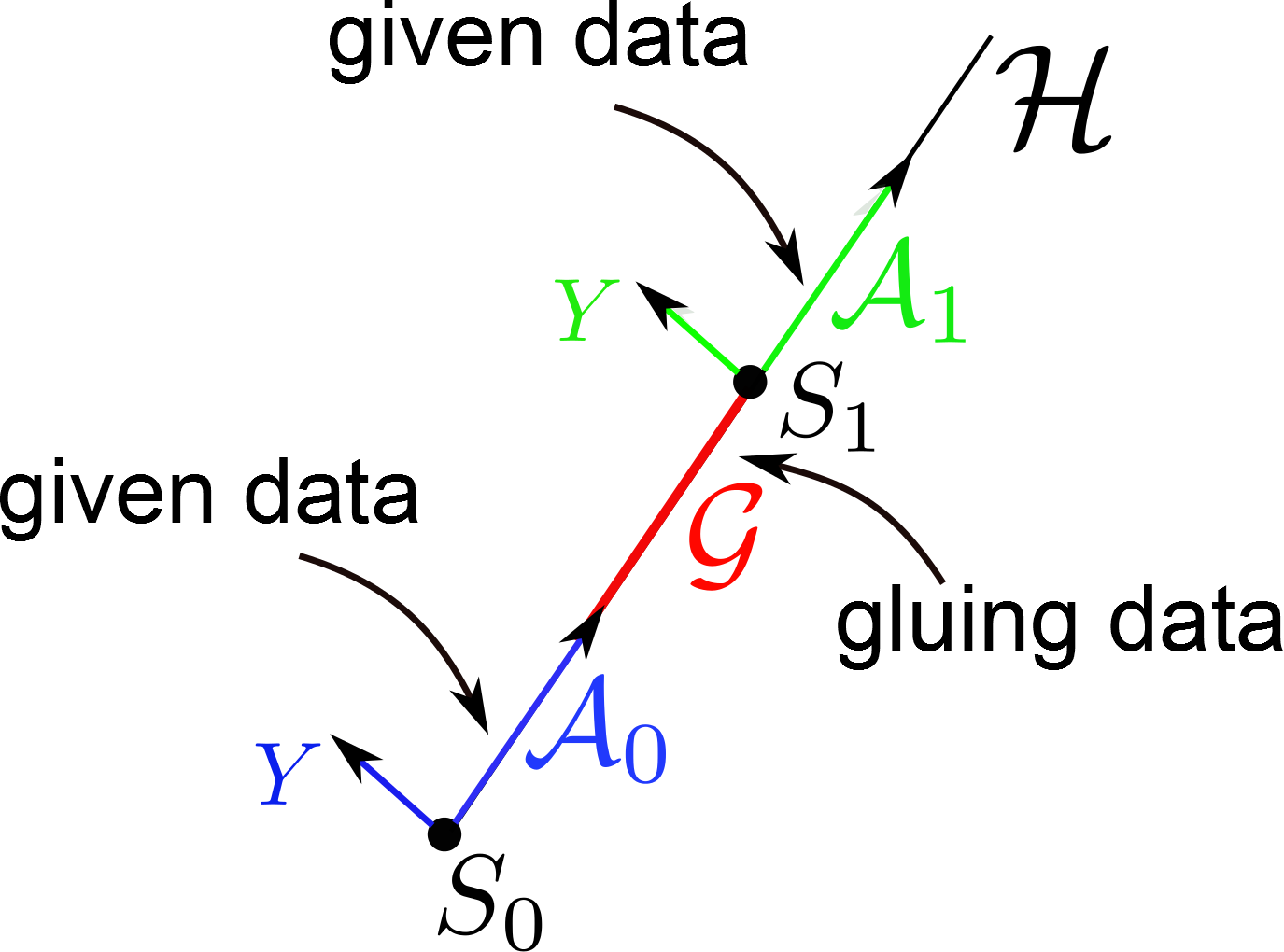}
	\label{figfdsp45picutre03}
\end{figure}

If $\psi$ solves the wave equation then the transversal derivative $Y\psi$ on $\hh$ is completely determined by the data $\left.\psi\right|_{\hh}$ on  $\hh$ and the transversal derivative $Y\psi$ at a section $S$ of $\hh$. For this reason, it is convenient to  ``forget'' about the incoming null hypersurfaces $\underline{\hh}_{1},\underline{\hh}_{2}$ and hence just ``keep'' the following data
\[\psi\left.\right|_{\mathcal{A}_{0}}, \ \ \ Y\psi\left.\right|_{S_{0}}\]
and 
\[\psi\left.\right|_{\mathcal{A}_{1}}, \ \ \ Y\psi\left.\right|_{S_{1}}.\]
In fact,\textit{ we can simply think of the data as given at the two spheres $S_{0},S_{1}$ as follows
\[ \text{Data}(S_{0})=\left\{  Y\psi\left.\right|_{S_{0}}, \ L^{n}\psi\left.\right|_{S_{0}},\, n\geq 0\right\}\]
and 
\[ \text{Data}(S_{1})=\left\{  Y\psi\left.\right|_{S_{1}}, \ L^{n}\psi\left.\right|_{S_{1}},\, n\geq 0\right\},\]
where $L$ is tangential to the null generator of $\hh$. 
Then, our problem is to smoothly extend $\psi$ on $\hh$ between $S_{0}$ and $S_{1}$ such that the transversal derivative $Y\psi$ is continuous on $\hh\cap\left\{0\leq v\leq 1\right\}$ (and hence such that $\psi$ is $C^{1}$ on $\hh\cap\left\{0\leq v\leq 1\right\}$). 
}

\medskip

Note that the $C^{k}$ case with $k>1$ (where one needs to ``glue'' transversal derivatives up to the $k$'th order) is addressed here, for simplicity, only in the spherically symmetric case; however the general case can be addressed by an amalgamation of the methods presented in Sections \ref{sec:TheGeneralCase} and \ref{sec:Genericity}.

Let  $S_{0}$ and $S_{1}$ be two leaves of a foliation $\s$ of $\hh$. Clearly, if $\hh$ admits non-trivial conservation laws with respect to the foliation $\s$ in the sense of  Definition \ref{definitionconservationlaw}, then gluing constructions of data at $S_{0}$ and $S_{1}$ are not always possible since the prescribed charges at $S_{0}$ and $S_{1}$ may not coincide. However, it is not a priori obvious if such charges are the only obstruction to gluing of characteristic data. On the other hand, if we can show that we can always glue to first order  characteristic data on $\hh$ then it immediately follows that $\hh$ does not admit any non-trivial charges.

\subsection{The main theorems}
\label{sec:TheMainResultxx}

The next theorem 1) characterizes the conservation laws on null hypersurfaces,  2) derives their role in the characteristic gluing problem, 3) provides necessary and sufficient conditions for their existence in terms of the kernel of the elliptic operator $\mathcal{O}^{\s}$ defined by \eqref{adjoint}, 4) uncovers their behavior under change of foliations and 5) establishes their non-genericity. 

\begin{mytheo}
Let $\hh$ be a regular, free from conjugate or focal points null hypersurface of a four-dimensional Lorentzian manifold $(\m,g)$. Let also  $\mathcal{S}=\big(S_{v}\big)_{v\in\mathbb{R}}$ be a foliation of $\hh$, such that $S_{v}$ are diffeomorphic to $\mathbb{S}^{2}$, and $\o^{\mathcal{S}}$ be the associated elliptic operator given by \eqref{adjoint}. Then we have the following
\begin{itemize}

	\item \textbf{Characteristic gluing constructions and conservations laws:}	
	One can perform first order gluing constructions on $\hh$ in the sense of  Definition \ref{firstordergluingdefinition} for general characteristic data if and only if there are no first order conservation laws on $\hh$ in the sense of  Definition \ref{definitionconservationlaw}, i.e. $\mathcal{W}^{\s}=\left\{0\right\}$. If $\hh$ admits conservation laws, then we can glue characteristic data if and only if their associated charges are equal. 
		
	\item  \textbf{Classification of null hypersurfaces admitting conservation laws:} Consider the following linear space
	\begin{equation}
	\mathcal{U}^{\s}=\left\{\Theta^{\s}\in C^{\infty}(\hh)\, :\ L\Theta^{\s}=0, \ \mathcal{O}^{\s}\left(\frac{1}{\phi}\cdot\Theta^{\s}\right)=0 \text{ on }\hh  \right\}\subset \mathcal{V}_{\hh},
	\label{newdef}
	\end{equation}
	where $\phi$ denotes the conformal factor of the sections of $\mathcal{S}$. Then, the null hypersurface $\hh$ admits first order conservation laws with respect to $\s$ for the wave equation \eqref{wave} in the sense of Definition \ref{definitionconservationlaw} if and only if   $\mathcal{U}^{\mathcal{S}}\neq \left\{ 0\right\}$.  In fact, the kernel of the conservation laws satisfies $\mathcal{W}^{\mathcal{S}}=\mathcal{U}^{\mathcal{S}}$.

	\item \textbf{Classification of conservation laws on null hypersurfaces:} No conserved (linear or non-linear) quantities (or, more generally, monotonic in $v$ quantities) involving the 1-jet of solutions to the wave equation exist on $\hh$ apart from the conservation laws given precisely by  Definition \ref{definitionconservationlaw}. Moreover, $\hh$ can only admit finitely many linearly independent conservation laws, i.e.~$\dim \mathcal{W}^{\mathcal{S}}<\infty$. 
	
		\item \textbf{Conservation laws and change of foliation:} If $\hh$ admits a conservation law with respect to the foliation $\mathcal{S}=\left\langle S_{0},L_{geod}, \Omega\right\rangle$, then it also admits a conservation law with respect to any other foliation $\mathcal{S}'=\big\langle S'_{0},L_{geod}',\Omega'\big\rangle$.  Specifically, the kernels $\mathcal{W}^{\s},\mathcal{W}^{\s'}$ satisfy $\mathcal{W}^{\s'}=f^{2}\cdot\mathcal{W}^{\s}=\left\{f^{2}\cdot \Theta^{\s},\,\Theta^{\s}\in \mathcal{W}^{\s}\right\}$, where $f\in\vh$ such that $L_{geod}'=f^{2}\cdot L_{geod}$, and so $\dim \mathcal{W}^{\s'}=\dim \mathcal{W}^{\s}$. Moreover, the value of the conserved charges is independent of the choice of foliation.

\item \textbf{Non-genericity of conservation laws:} A null hypersurface does not admit conservation laws for generic ambient metrics. The same result holds even if we restrict to spacetimes $(\m,g)$ satisfying the Einstein-vacuum equations. 
\end{itemize}
\label{theoremmainintro}
\end{mytheo}

Theorem \ref{theoremmainintro} can be used to show that the event horizon of a subextremal black hole does not admit conservation laws and that the event horizon of an extemal black hole admits a unique conservation law:

\begin{mytheo}
 \textbf{Conservation laws on extremal black holes:}	The event horizon ${\hh}^{+}$ of 
an extremal black hole satisfies $\dim \mathcal{U}^{\mathcal{S}}_{{\hh}^{+}} =1$ and the unique corresponding conservation law coincides with the conservation law on  extremal black holes found in \cite{aretakis4, hj2012, murata2012}. On the other hand, any Killing horizon with positive surface gravity and negative transversal null mean curvature (such as the event horizon of any subextremal Kerr black hole) does not admit conservation laws.

\label{the2}
\end{mytheo}
	
	Theorem \ref{theoremmainintro} applied in a limiting sense on null infinity of asymptotically flat spacetime recovers the Newman--Penrose constant and in fact shows that it is the \textbf{only} conserved charge along null infinity:
	
	\begin{mytheo}
	\textbf{Conservation laws and the Newman--Penrose constant:}
	The null infinity $\mathcal{I}$ of an asympotically flat spacetime satisfies $\dim\mathcal{U}^{\mathcal{S}}_{\mathcal{I}}=1$ with respect to an appropriately rescaled operator $\o_{\mathcal{I}}^{\ \mathcal{S}}$, and the \textbf{unique} corresponding conservation law  coincides with the  (first-order) Newman--Penrose constant on $\mathcal{I}$.
	\label{theo3}
	\end{mytheo}
	
	Finally, the following theorem derives necessary and sufficient conditions for the existence of higher order conservation laws in the context of spherical symmetry.

	 \begin{mytheo}
	\textbf{Higher order gluing constructions and conservation laws in spherical symmetry:} Let $(\m,g)$ be a spherically symmetric spacetime and $\hh$ be a spherically symmetric null hypersurface. Then, for all $k,l\in\mathbb{N}$ there is a unique expression $R_{k,l}$ which depends only on the geometry of $\hh$ such that if  $R_{i,l}\neq 0$ for  $i=1,...,k-1$ and $R_{k,l}=0$ on $\hh$ then there is a conservation law involving the $k$-jet of solutions to the wave equation. The kernel of this conservation law consists of all the eigenfunctions of the standard spherical Laplacian $\lapp_{\mathbb{S}^{2}}$ which correspond to the eigenvalue $-l(l+1)$. If, on the other hand, we have that for all $l\in \mathbb{N}$  $R_{i,l}\neq 0,\,  i=1,...,k$ almost everywhere on $\hi$ then we can glue characteristic data up to the $k$'th order. Moreover, the higher-order Newman--Penrose constants are limiting examples of the above charges. 

\label{theo4}
\end{mytheo}

\subsection{PDE aspects of the conservation laws}
\label{sec:Remarks}

The existence of the charge \eqref{eextremalcharge} implies that the solutions to the wave equation on extremal black holes do not disperse along the event horizon. This is in stark contrast to the result of Dafermos and Rodnianski \cite{tria} who derived decay results for $\psi$ and all its derivatives up to and including the event horizon for the general subextremal Kerr family $|a|<M$ (for decay results on slowly rotating Kerr black holes see \cite{blukerr, lecturesMD, tataru2}). In fact, a closer analysis \cite{aretakis1,aretakis2,aretakis3} showed that for generic initial data we have
\[\left|\psi\right|\rightarrow 0, \ \ \left|Y\psi\right|\nrightarrow 0, \ \ \left|Y^{k}\psi\right|\rightarrow+\infty,\text{ for }k\geq 2, \]
along the (degenerate) event horizon. Similar blow-up results were obtained for the higher order energies.  The asymptotic behavior of tails along the event horizon was studied by Lucietti et al \cite{hm2012} and Ori \cite{ori2013}.  Further applications were given by  Dain and Dotti  \cite{dd2012}. A very recent numerical analysis of this instability in the context of the Einstein-Maxwell-scalar field system under spherical symmetry can be found in \cite{harvey2013}.  A non-linear instability arrising from the conserved charges was given in \cite{aretakis2013}.

In a broader content and at a slightly philosophical level, transversal derivatives to globally distinguished null hypersurfaces usually pose significant difficulties in proving dispersive estimates for wave type equations in view of the fact that one does not have an a priori equation for these derivatives along their direction (i.e.~we do not have an equation of the form $Y^{\s}Y^{\s}\psi=...$).  For this reason, understanding the structure of the transversal derivatives in PDE is of fundamental importance; see for example the null condition introduced by Klainerman \cite{SK86} which allowed him to obtain global existence of quasilinear equations with small initial data. It would be interesting to see further global repercussions of the conservation laws defined here to other linear and non-linear settings.

\subsection{Outline of the paper}
\label{sec:OutlineOfThePaper}

An outline of the paper is as follows: In Section \ref{sec:TheGeometricSetting} we introduce the basic geometric concepts relevant to a double null foliation and in Section \ref{sec:TheGeneralCase} we derive the relation between the  conservation laws and  characteristic gluing constructions on a general null hypersurface. In Section \ref{sec:TheNullInfinityMathcalI} we restrict to the null infinity of asympotically flat spacetimes and recover the Newman--Penrose constants. In Section \ref{sec:KillingHorizons} we apply the main theorem to Killing horizons with particular emphasis on the event horizon of black holes. Finally, in the Section \ref{sec:Genericity} we derive necessary and sufficient conditions for higher order conservation laws and gluing constructions in spherical symmetry.

\section{The geometric setting}
\label{sec:TheGeometricSetting}

Let $\hh$ be a regular null hypersurface of a four-dimensional Lorentzian manifold $(\m,g)$. We will express the wave equation in an appropriate frame which will allow us to explicitly make use of the elliptic structure on $\hh$. For we work with canonical coordinates suitably adapted to a double null foliation of $\m$ which we introduce in the next subsection.

\subsection{The double null foliation}
\label{sec:TheDoubleNullFoliation}

\paragraph{Null Foliations and Optical Functions\medskip \\}
\label{sec:NullFoliationsandOpticalFunctions}

A \textit{foliation} $\mathcal{S}$ of a null hypersurface $\hh$ in a four-dimensional Lorentzian manifold $(\m,g)$ is a collection of sections $S_{v}$, smoothly varying in $v$, such that $\cup_{v}S_{v}=\hh$. We assume that $S_{v}$ are diffeomorphic to the 2-sphere. We will show that any foliation is uniquely determined by the choice of one section, say $S_{0}$, the choice of a null tangential to $\hh$ vector field $\left.L_{geod}\right|_{S_{0}}$ restricted on $S_{0}$ and a function $\Omega$ on $\hh$. We extend $\left.L_{geod}\right|_{S_{0}}$  to a  null vector field tangential to the null generators of $\hh$ such that
\[\nabla_{{L}_{geod}}{{L}_{geod}}=0.\]We then define the vector field
\begin{equation}
L=\Omega^{2}\cdot L_{geod}
\label{eq:l}
\end{equation}
on $\hh$ and consider the affine parameter $v$ of $L$ such that 
\[Lv=1, \text{ with } v=0, \text{ on } S_{0}.\]
Let $S_{v}$ denote the level sets of $v$ on $\hh$ which are the leaves of the foliaton $\mathcal{S}$. We use the notation
\begin{equation}
\mathcal{S}=\left\langle S_{0}, \left.L_{geod}\right|_{S_{0}}, \Omega \right\rangle.
\label{foliationdef}
\end{equation}
We can also extend the above construction to obtain a foliation of (an appropriate region of) $\m$ by null hypersurfaces $C_{u},\underline{C}_{v}$ intersecting at embedded 2-spheres $S_{u,v}$.  We denote $C_{0}=\hh$ and define $\underline{C}_{0}$ to be the null hypersurface generated by null hypersurfaces normal to $S_{0}$ and conjugate to $C_{0}$. 
 \begin{figure}[H]
   \centering
		\includegraphics[scale=0.1]{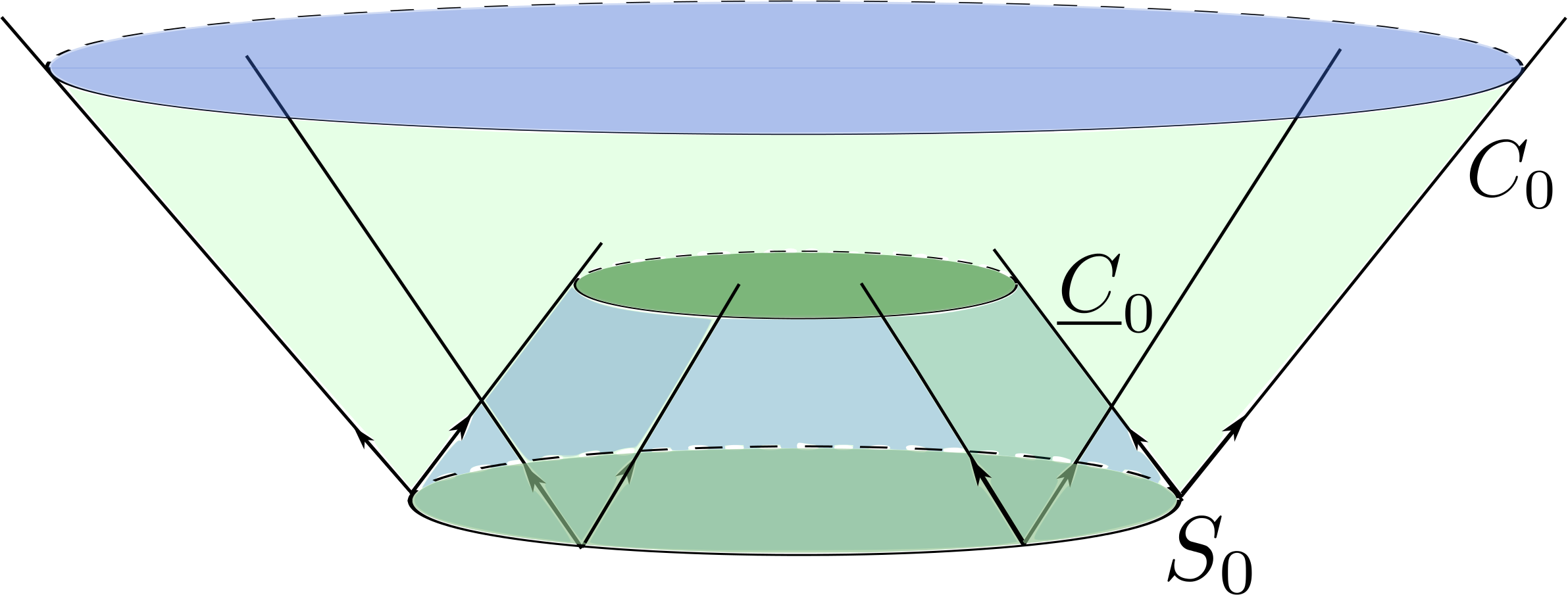}
	\label{fig:p43}
\end{figure}
We then consider the vector field $\underline{L}_{geod}\!\left.\right|_{S_{0}}$ at $S_{0}$ tangential to the null generators of $\underline{C}_{0}$ determined by the relation
\[g({L}_{geod},\underline{L}_{geod})=-\Omega^{-2}\]
and extend it on $\underline{C}_{0}$ via the geodesic equation
\[ \nabla_{\underline{L}_{geod}}\underline{L}_{geod}=0. \]
We now extend $\Omega$ to be a function on $\underline{C}_{0}$ and, as before, consider the vector field
\begin{equation}
\begin{split}
 \underline{L}=\Omega^{2}\cdot \underline{L}_{geod}
\end{split}
\label{normalizedvf}
\end{equation}
and similarly define the function $u$ on $\underline{C}_{0}$ by 
\begin{equation*}
\underline{L}u=1, \ \ \text{with } u=0 \text{ on }S_{0}.
\end{equation*}
Let $\underline{S}_{\tau}$ be the embedded 2-surface on $\underline{C}_{0}$ such that $u=\tau$.
 \begin{figure}[H]
   \centering
		\includegraphics[scale=0.101]{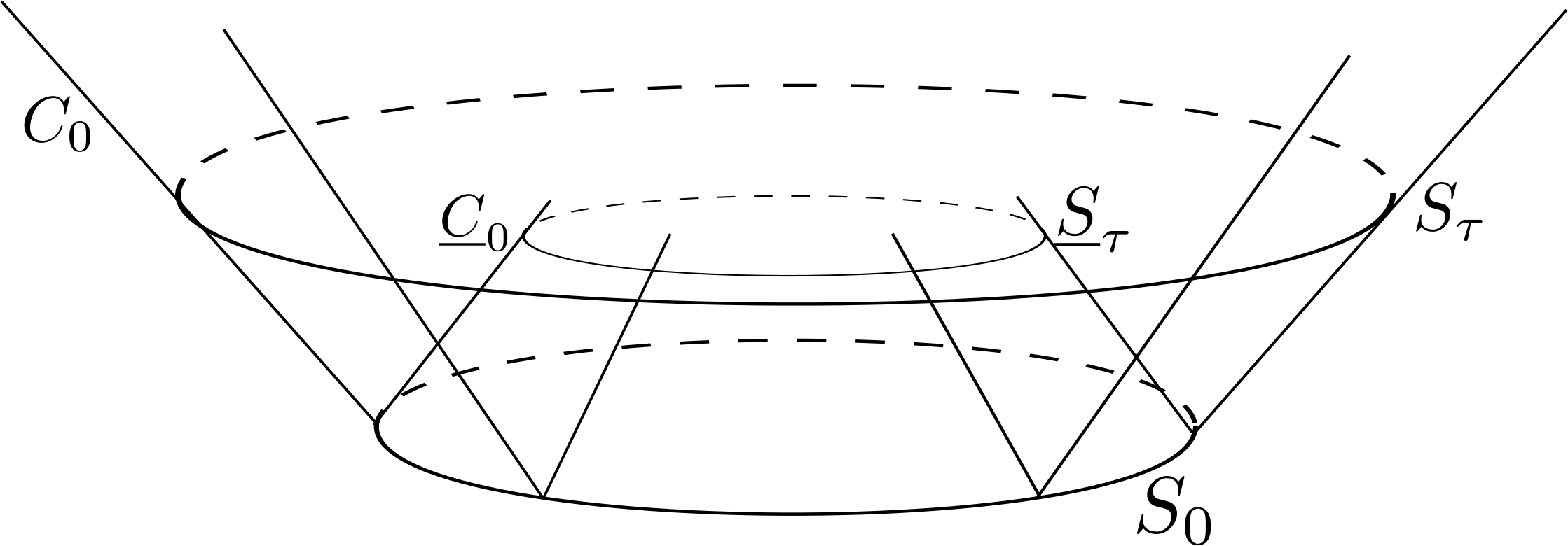}
	\label{fig:p44}
\end{figure}
\noindent We also define $\underline{L}_{geod}$ on $C_{0}$ such that $\underline{L}_{geod}$ is null and normal to $S_{\tau}$ and $g(L_{geod},\underline{L}_{geod})=-\Omega^{-2}$. We similarly extend $L_{geod}$ on $\underline{C}_{0}$. 
 \begin{figure}[H]
   \centering
		\includegraphics[scale=0.11]{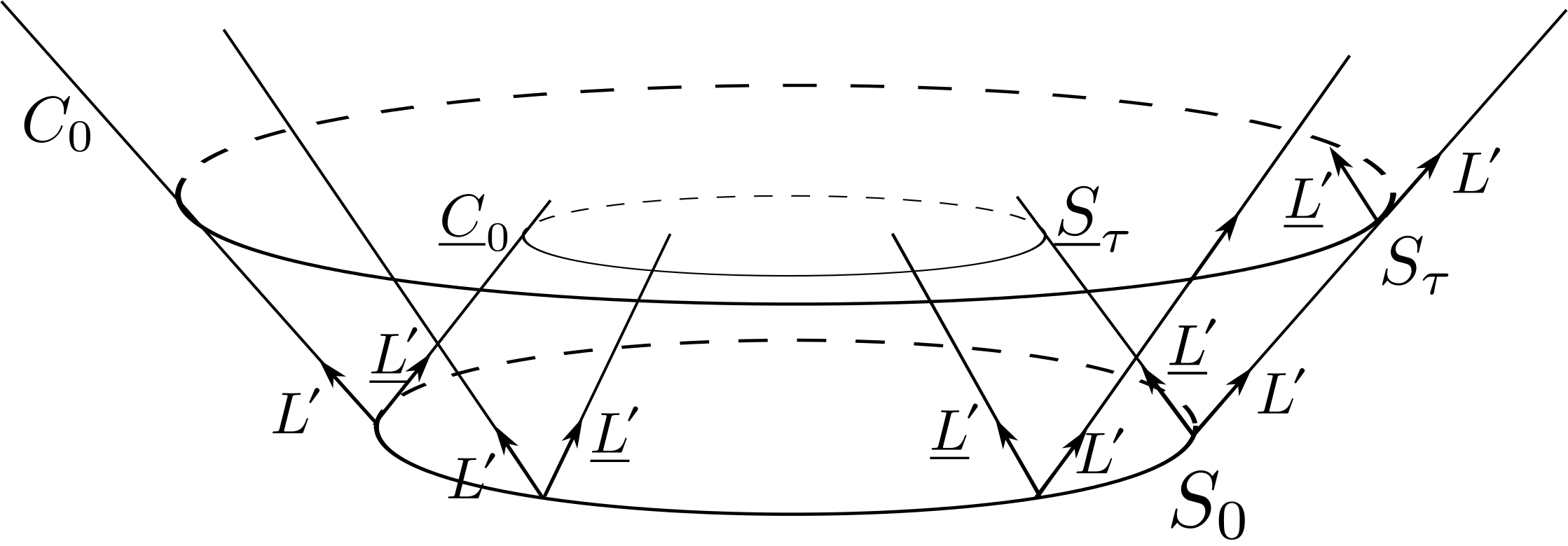}
	\label{fig:p44}
\end{figure}
Consider the affinely parametrized null geodesics which emanate from the points on $S_{\tau}$ with initial tangent vector $\underline{L}_{geod}$. These geodesics, whose tangent we denote by $\underline{L}_{geod}$, span  null hypersurfaces  which we will denote by $\underline{C}_{\tau}$. Hence, $C_{0}\cap\underline{C}_{\tau}=S_{\tau}$. Similarly, we define $L_{geod}$ globally (and the hypersurfaces $C_{\tau}$ such that their normal is $L_{geod}$). Extend the vector field $L,\underline{L}$ to global vector fields such that 
 \begin{equation*}
\begin{split}
L=\Omega^{2}\cdot L_{geod}, \ \ \ \underline{L}=\Omega^{2}\cdot\underline{L}_{geod}.
\end{split}
\end{equation*}
We will refer to $\Omega$ as the null lapse function. We also extend the functions $u,v$ to global functions such that 
 \begin{equation*}
\begin{split}
Lu=0, \ \ \ \ \ \underline{L}\,v=0.
\end{split}
\end{equation*}
Therefore, 
 \begin{equation*}
\begin{split}
C_{\tau}=\left\{u=\tau\right\}, \ \ \ \underline{C}_{\tau}=\left\{v=\tau\right\},
\end{split}
\end{equation*}
and hence $u,v$ are optical functions.  
\begin{figure}[H]
   \centering
		\includegraphics[scale=0.087]{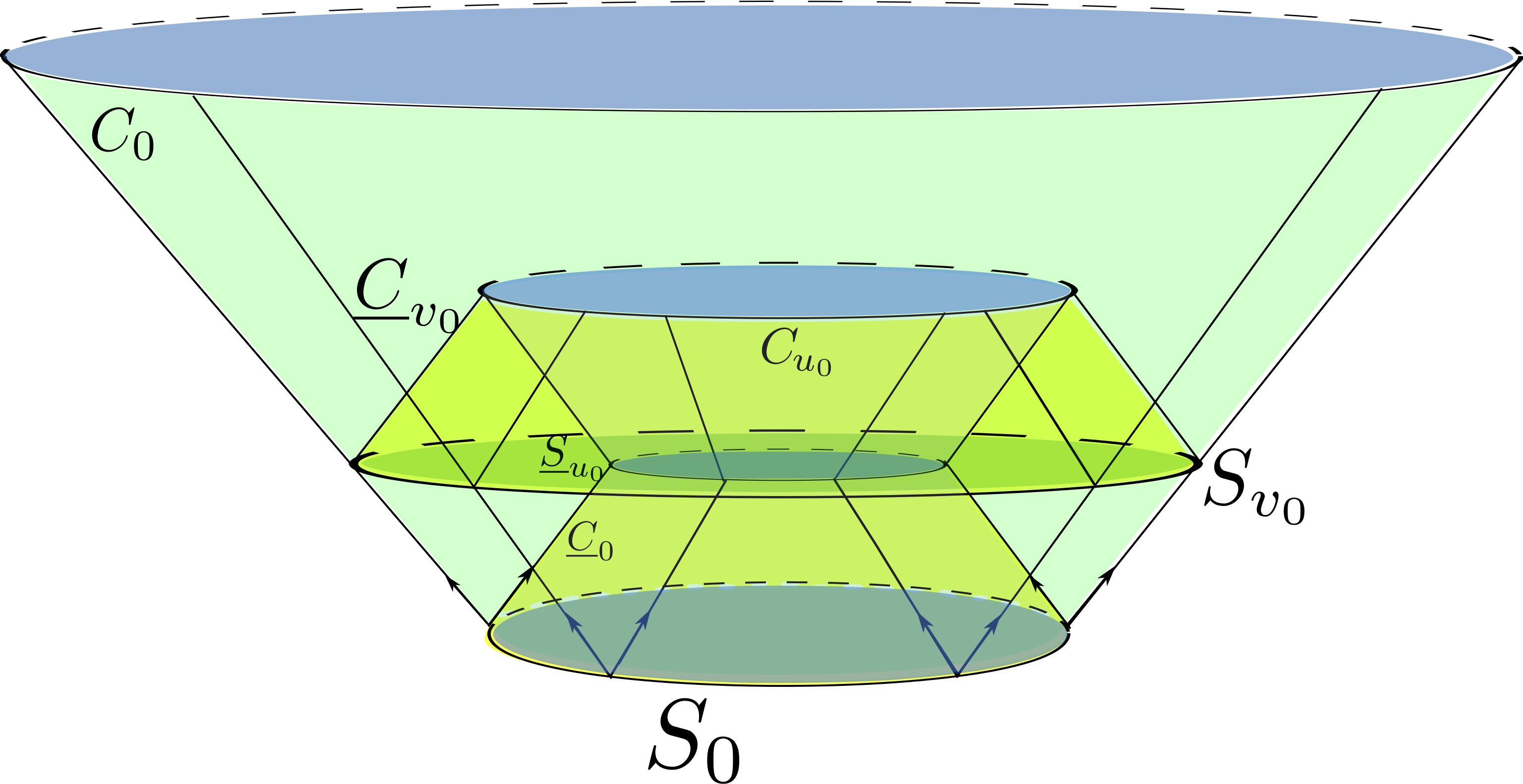}
	\label{fig:p42}
\end{figure}

 The importance of the renormalized vector fields $L,\underline{L}$ is manifest from the following 
\begin{proposition}
The optical functions $u,v$ satisfy the following relations
\[\nabla v=-\underline{L}_{geod}, \ \ \ \ \nabla u=-L_{geod}\]
and 
\[Lv=1, \ \ \ \ \underline{L} u=1. \]
\label{impo}
\end{proposition}
\begin{proof}

Since $u,v$ are optical functions we have that $\nabla u,\nabla v$ satisfy the geodesic equation. Since $L_{geod},\underline{L}_{geod}$ satisfy the geodesic equation as well, it suffices to show, for instance, that $\nabla v=-\underline{L}_{geod}$ on $C_{0}$. Expressing $\nabla v$ in terms of the null frame $(L_{geod},\underline{L}_{geod}, e_{1}, e_{2})$, where $e_{1},e_{2}$ is a local frame on the spheres $S_{\tau}$, we obtain
 \begin{equation*}
(\nabla v)^{\underline{L}_{geod}}=g^{\underline{L}_{geod}L_{geod}}\cdot(L_{geod}v)+g^{\underline{L}_{geod}\underline{L}_{geod}}\cdot(\underline{L}_{geod}v)=-\Omega^{2}\cdot\Omega^{-2}=-1,
\end{equation*}
on $C_{0}$, and similarly, we obtain that the remaining components with respect to the above frame are zero. This proves the first equation. For the second, it suffices to notice that 
\[Lv=g(L,\nabla v)=g(L,-\underline{L}_{geod})=-\Omega^{2}\cdot g(L_{geod},\underline{L}_{geod})=1.\]

\end{proof}
Note  that
\begin{equation}
g\big(\underline{L},L_{geod}\big)=-1,
\label{eq:unormalizationgeod}
\end{equation}
which shows that $L_{geod}$ determines to first order the optical function $u$ on $\hh=C_{0}$.

\paragraph{Gauge Freedom\medskip \\}
\label{sec:FixingAGauge}

The above analysis shows that a double null foliation $\mathcal{D}$ can be completely determined by the following 
\[\mathcal{D}=\left\langle S_{0}, \left. L_{geod}\right|_{S_{0}}, \left.\Omega\right|_{C_{0}},\left.\Omega\right|_{\underline{C}_{0}}\right\rangle.\]
Note that the freedom for the vector field $\left. L_{geod}\right|_{S_{0}}$ and the functions $\left.\Omega\right|_{C_{0}},\left.\Omega\right|_{\underline{C}_{0}}$ reflects the freedom for the functions $u'=u'(u), v'=v'(v)$. Indeed, fixing $\Omega$ and $L_{geod}$ determines up to additive constants the (optical) functions $u,v$.

\paragraph{The Diffeomorphisms $\Phi_{u,v}$\medskip \\}
\label{sec:coordinatescdiffeo}

We can construct a  diffeomorphism $\Phi_{u,v}$ from any sphere $S_{u,v}$ to  $S_{0}$ as follows: If $p\in S_{u,v}$ then we can consider the point $q\in S_{0,v}$ which is the intersection of $C_{0}$ and the null generator of $\underline{C}_{u}$ passing through $p$. We then let $\Phi_{u,v}(p)$ to be the unique point of intersection of $S_{0}$ and the null generator of $C_{0}$ passing through $q$. We can also consider a diffeomorphism 
\begin{equation}
\Phi: S_{0}\rightarrow\mathbb{S}^{2}
\label{phidiffeo}
\end{equation}
and compose $\Phi_{u,v}$ with $\Phi$ to obtain a diffeomorphism from $S_{u,v}$ to $\mathbb{S}^{2}$. This constructions allows us, for example, to equip all surfaces $S_{u,v}$ with the standard metric in a very precise way. See also the figure below.


\paragraph{The Canonical Coordinate System\medskip \\}
\label{sec:coordinatesc}
One can consider at each point a \textit{null frame} $(L,\underline{L},e_{1},e_{2})$ adapted to the double null foliation, where $e_{1},e_{2}$ is a local frame for the spheres $S_{u,v}$. However, these vector fields do not correspond to a coordinate system. 

Using the optical functions $u,v$ we will introduce a coordinate system suitably adapted to the corresponding double null foliation of the spacetime.

If $p\in\m$ then $p\in C_{u_{0}}\cap\underline{C}_{v_{0}}$ and hence $u(p)=u_{0}, v(p)=v_{0}$. We next prescribe angular coordinates for the point $p$ on the 2-surface $C_{u_{0}}\cap\underline{C}_{v_{0}}$.

Let $(\theta^{1},\theta^{2})$ denote a coordinate system on a domain of $S_{0}$.  Then  we assign to $p$ the angular coordinates of the point $\Phi_{u_{0},v_{0}}(p)$ as depicted below
 \begin{figure}[H]
   \centering
		\includegraphics[scale=0.07]{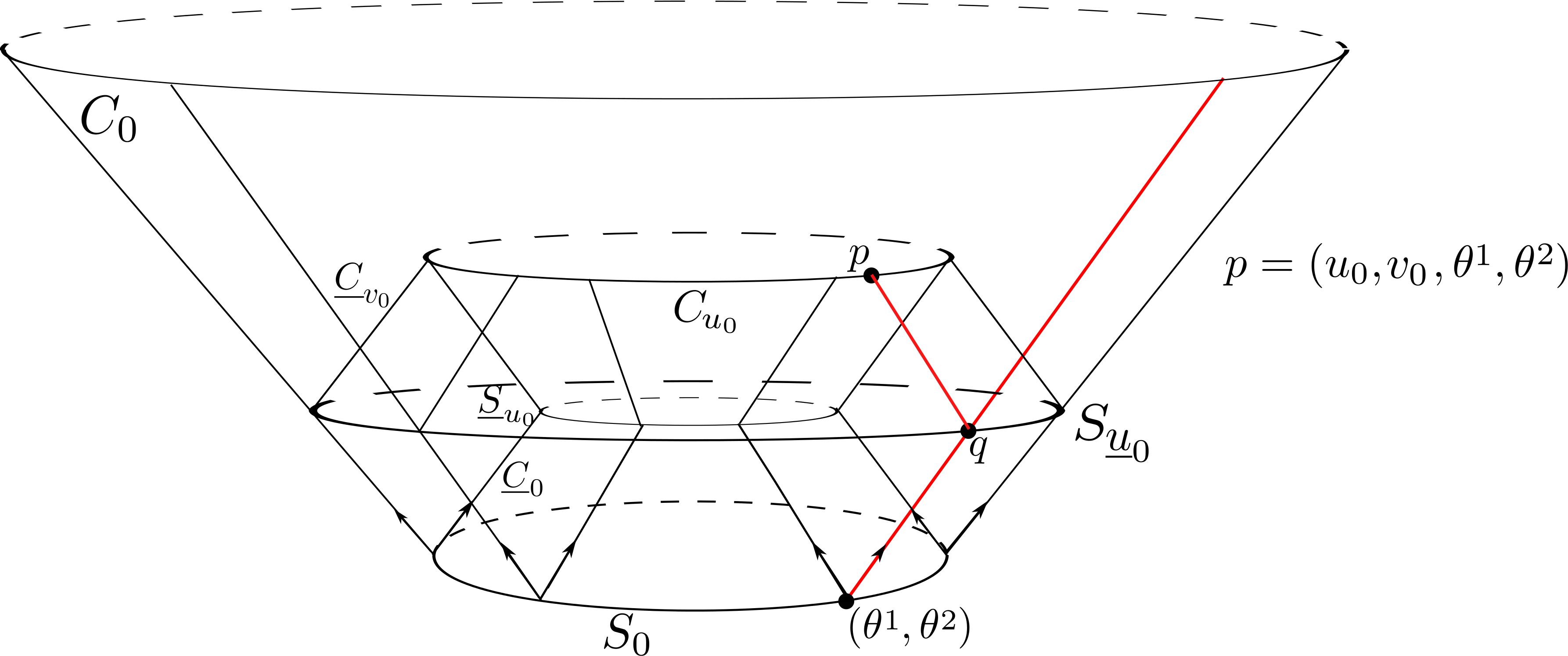}
	\label{fig:p45}
\end{figure}
By construction we have everywhere: \[\frac{\partial}{\partial{u}}=\underline{L}\] and 
\[\frac{\partial}{\partial{\theta^{1}}},\frac{\partial}{\partial{\theta^{2}}}\in TS_{u,v},\]
whereas
\[\frac{\partial}{\partial{v}}=L \text{ : on }C_{0}.\]
Note that the latter equation will not in general hold everywhere, as it is easily seen from the picture below:
 \begin{figure}[H]
   \centering
		\includegraphics[scale=0.08]{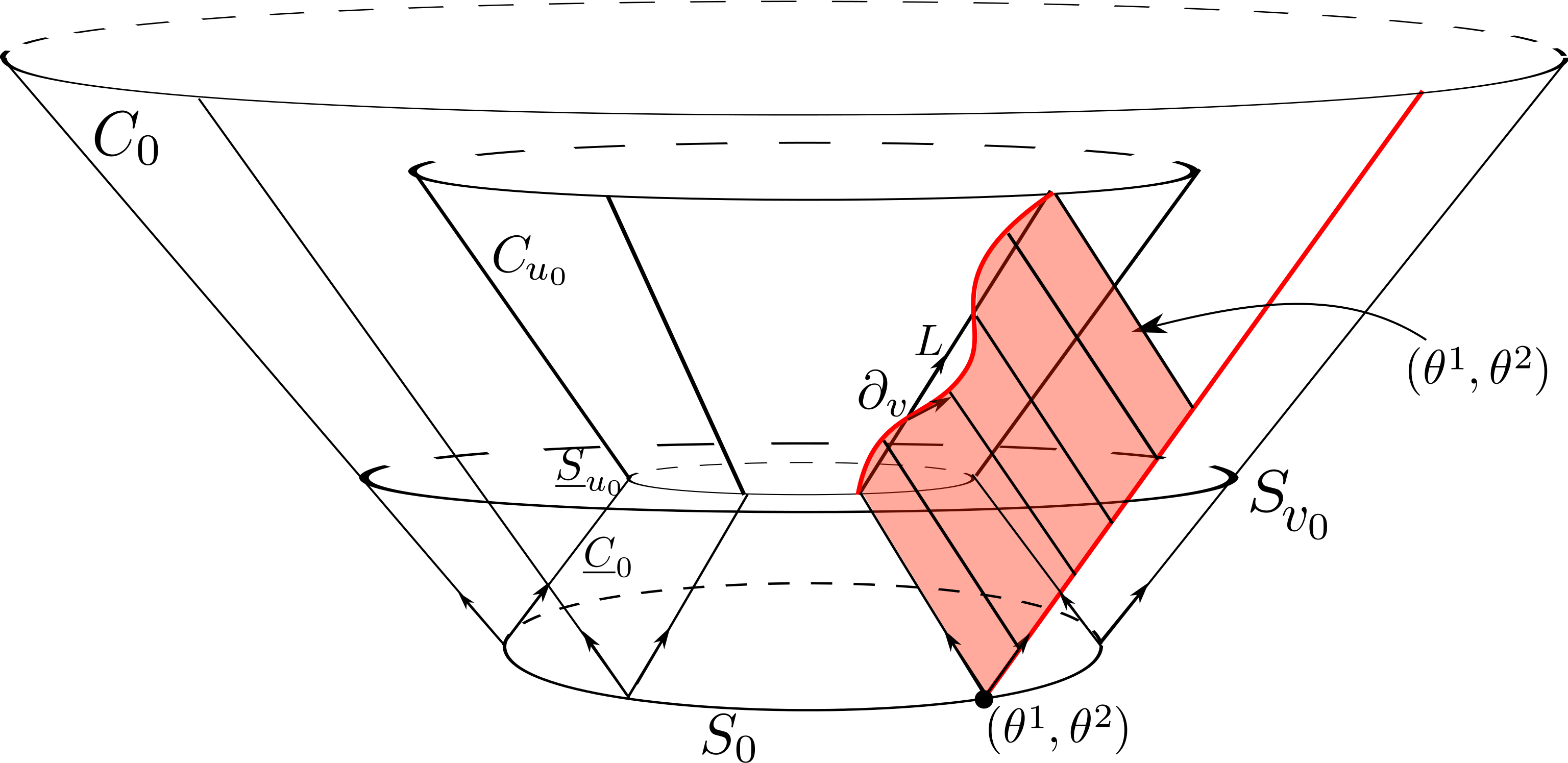}
	\label{fig:p46}
\end{figure}
From now on, for simplicity we denote $\partial_{v}=\frac{\partial}{\partial{v}}$, and so on.
 In general we have:
\[\partial_{v}=L+b^{i}\partial_{\theta^{i}}.\]
By virtue of the equations $\underline{L}=\partial_{u}$ and $[\partial_{u},\partial_{v}]=[\partial_{u},\partial_{\theta^{i}}]=0$ we obtain:
\begin{equation*}
[L,\underline{L}]=-\frac{\partial{b^{i}}}{\partial{u}}\partial_{\theta^{i}}\in TS_{u,v},
\end{equation*}
and therefore,
\begin{equation}
\frac{\partial{b^{i}}}{\partial{u}}=-d\theta^{i}([L,\underline{L}]), \text{ and }b^{i}=0 \text{ on }C_{0}=\left\{u=0\right\}. 
\label{cootor}
\end{equation}
Hence, the S-tangent vector field $b=b^{i}\partial_{\theta^{i}}$ is the obstruction to the integrability of  $\left\langle L,\underline{L}\right\rangle=\big(TS_{u,v}\big)^{\perp}$. In order to compute  $b$ it suffices to compute $[L,\underline{L}]$. Since $[L,\underline{L}]=\nabla_{L}\underline{L}-\nabla_{\underline{L}}L$, it suffices to compute the connection coefficients (see below).

The metric $g$ with respect to the canonical coordinates is given by
\begin{equation}
g=-2\Omega^{2}dudv+(b^{i}\, b^{j}\,\gi_{ij})dvdv-2(b^{i}\, \gi_{ij})d\theta^{j}dv+\gi_{ij}\,d\theta^{i}d\theta^{j},
\label{metriccan}
\end{equation}
where $\gi$ denotes the induced metric on the 2-surfaces $S_{u,v}=C_{u}\cap \underline{C}_{v}$. We immediately obtain
\begin{equation}
\text{det}(g)=-\Omega^{4}\cdot \text{det}(\gi).
\label{detg}
\end{equation}
\bigskip

\noindent\textbf{Null Frames}
\bigskip
 
\noindent From now on we denote $S_{u,v}=C_{u}\cap \underline{C}_{v}$. If $\left\{e_{1},e_{2}\right\}=\big(e_{A}\big)_{A=1,2}$ is an arbitrary frame on the spheres $S_{(u,v)}$ then we, in fact, have the following null frames:
\begin{itemize}
	\item 
\textbf{Geodesic frame:} $(e_{1},e_{2},L_{geod}, \underline{L}_{geod})$,
\item
\textbf{Equivariant frame:} $(e_{1},e_{2},L,\underline{L})$,
\item
\textbf{Normalized frame:} $(e_{1},e_{2}, e_{3},e_{4}),$
\item
\textbf{Coordinate frame:} $(\partial_{\theta^{1}},\partial_{\theta^{2}},\partial_{v},\partial_{u})$.

\end{itemize}
where \[e_{3}=\Omega \underline{L}_{geod}=\frac{1}{\Omega}\underline{L}, \ \ \ \ e_{4}=\Omega L_{geod}=\frac{1}{\Omega}{L}.\]
\begin{figure}[H]
   \centering
		\includegraphics[scale=0.135]{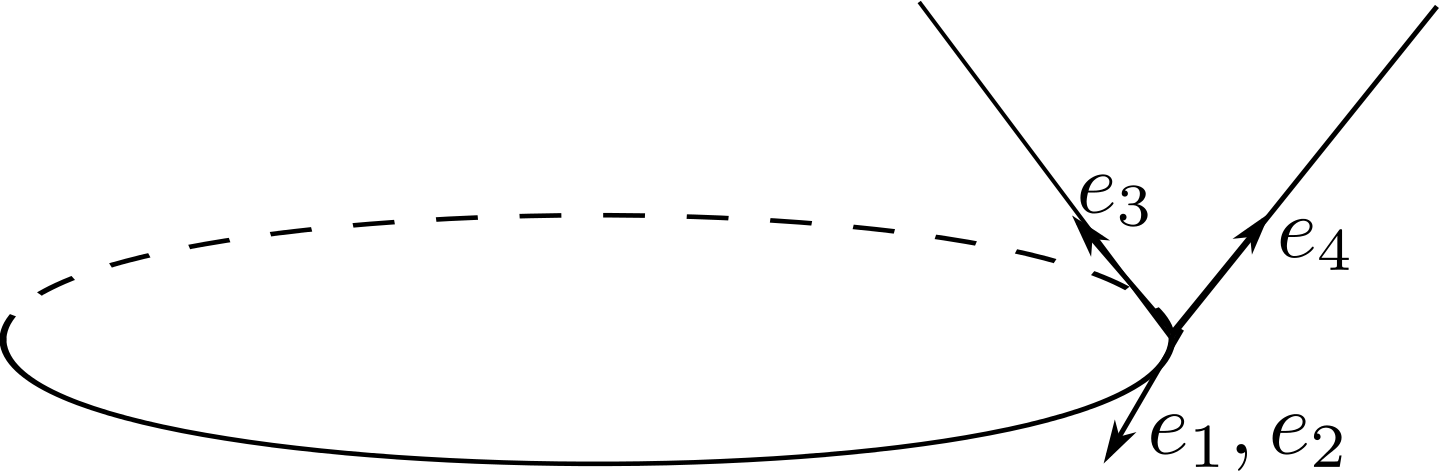}
	\label{fig:p3344}
\end{figure} 
Note that $e_{3},e_{4}$ satisfy the normalization property
\begin{equation}
g(e_{3},e_{4})=-1.
\label{eq:}
\end{equation}

\smallskip

\noindent\textbf{The conformal geometry and conformal factor}
\bigskip

The conformal class of $\gi$ contains a unique representative (metric) $\hat{g}$ such that $\sqrt{\hat{g}}=\sqrt{\gi_{\mathbb{S}^{2}}}$. Here we consider the diffeomorphism $\Phi\circ\Phi_{v}$ to identify the section $S_{v}$ with $\mathbb{S}^{2}$. Equivalently, $\gi$ and $h$ are such that the induced volume forms on $S_{v}$ are equal. Since $\gi$ and $\hat{g}$ are conformal there is a conformal factor such that $\gi=\phi^{2}\cdot \hat{g}$. Then, $\sqrt{\gi}=\phi^{2}\sqrt{\hat{g}}=\phi^{2}\sqrt{\gi_{\mathbb{S}^{2}}}$ and hence 
\[\phi^{2}= \frac{\sqrt{\gi}}{\sqrt{\gi_{\mathbb{S}^{2}}}}. \]
 Clearly the above imply that 
\begin{equation}
\phi=\frac{\sqrt[4]{\gi}}{\sqrt[4]{\gi_{\mathbb{S}^{2}}}}.
\label{phi}
\end{equation}
Note that $\phi$ is a smooth function on the sphere $S_{v}$ and does not depend on the choice of the coordinate system. Note also that for a spherically symmetric metric we have $\phi=r$, where $r$ is the radius.

\bigskip

\noindent\textbf{Connection Coefficients}

\bigskip

We consider the \textbf{normalized frame} $(e_{1},e_{2},e_{3},e_{4})$ defined above.  We define the connection coefficients with respect to this frame to be the smooth functions $\Gamma^{\lambda}_{\mu\nu}$
such that 
\[\nabla_{e_{\mu}}e_{\nu}=\Gamma^{\lambda}_{\mu\nu}e_{\lambda}, \ \ \lambda, \mu, \nu\in\left\{1,2,3,4\right\}\]
Here $\nabla$ denotes the connection of the spacetime metric $g$. We are mainly interested in the case where at least one of the indices $\lambda,\mu,\nu$ is either 3 or 4 (otherwise, we obtain the Christoffel symbols with respect to the induced metric $\gi$).  Following \cite{DC09,christab}, these coefficients  are completely determined by the following components:
\medskip

\noindent\textbf{The components} $\chi,\underline{\chi},\eta,\underline{\eta}, \omega,\underline{\omega},\zeta$:
\smallskip
\begin{equation}
\begin{split}
&\chi_{AB}=g(\nabla_{A}e_{4},e_{B}), \ \ \ \ \ \underline{\chi}_{AB}=g(\nabla_{A}e_{3},e_{B}),\\
&\ \eta_{A}=g(\nabla_{3}e_{4},e_{A}),\ \ \ \ \  \ \,  \underline{\eta}_{A}=g(\nabla_{4}e_{3},e_{A}),\\
&\ \ \ \omega=-g(\nabla_{4}e_{4},e_{3}), \ \ \ \  \underline{\omega}=-g(\nabla_{3}e_{3},e_{4}),\\
& \ \ \ \ \ \ \ \ \ \ \ \ \ \ \ \ \ \ \ \ \  \zeta_{A}=g(\nabla_{A}e_{4},e_{3})
\end{split}
\label{concoef}
\end{equation}
where $\big(e_{A}\big)_{A=1,2}$ is an arbitrary frame on the spheres $S_{v}$ and $\nabla_{\mu}=\nabla_{e_{\mu}}$. Note that $\underline{\zeta}=-\zeta$. The covariant tensor fields $\chi,\underline{\chi},\eta,\underline{\eta},\zeta$ are only defined on $T_{x}S_{v}$. \textbf{We can naturally extend these to tensor fields to be defined on $T_{x}\m$ by simply letting their value to be zero if they act on $e_{3}$ or $e_{4}$. Such tensor fields will in general be called \underline{$S$-tensor fields}. Note that a vector field is an $S$-vector field if it is tangent to the spheres $S_{v}$.} The importance the $S$-tensors originates from the fact that one is interested in understanding the embedding of $S_{v}$ in $\hh$.

The connection coefficients $\Gamma$ can be recovered by the following relations:
\begin{equation}
\begin{split}
\nabla_{A}e_{B}=\nabb_{A}e_{B}+&\chi_{AB}e_{3}+\underline{\chi}_{AB}e_{4},\\
\nabla_{3}e_{A}=\nabb_{3}e_{A}+\eta_{A}e_{3},& \ \ \ \nabla_{4}e_{A}=\nabb_{4}e_{A}+\underline{\eta}_{A}e_{4},\\
\nabla_{A}e_{3}=\underline{\chi}_{A}^{\ \ \sharp B}e_{B}+\zeta_{A}e_{3},& \ \ \ \nabla_{A}e_{4}={\chi}_{A}^{\ \ \sharp B}e_{B}-\zeta_{A}e_{4},\\
\nabla_{3}e_{4}=\eta^{\sharp A}e_{A}-\underline{\omega}e_{4},& \ \ \ \nabla_{4}e_{3}=\underline{\eta}^{\sharp A}e_{A}-\omega e_{3},\\
\nabla_{3}e_{3}=\underline{\omega}e_{3},& \ \ \ \nabla_{4}e_{4}=\omega e_{4},
\end{split}
\label{connectioncoef}
\end{equation}

\bigskip

\noindent\textbf{Curvature Components}

\bigskip

We next decompose the Riemann curvature $R$ in terms of the normalized null frame. First, we define the following components, which contain at most two S-tangential components (and hence at least 2 null components):
\begin{equation}
\begin{split}
\a_{AB}=R_{A4B4},&\ \ \ \ \underline{\a}_{AB}=R_{A3B3},\\
\beta_{A}=R_{A434},&\ \ \ \ \underline{\beta}_{A}=R_{A334},\\
\rho= R_{3434}, & \ \ \ \sigma=\frac{1}{2}\epsi^{AB}R_{AB34}.
\end{split}
\label{curvcompdeflist}
\end{equation}
Note that $R(\cdot,\cdot, e_{3},e_{4})$, when restricted on $T_{x}S_{v}$, is an antisymmetric form and hence collinear to the volume form $\epsi$ on $S_{v}$. Furthermore, if 
\[(*R)_{3434}=*\rho,\]
then $*\rho= 2\sigma$. Here,  the dual $*R$ of the Riemann curvature is defined to be the (0,4) tensor:
\[(*R)_{\a\b\gamma\delta}=\epsilon_{\mu\nu\a\beta}\,R^{\mu\nu}_{\ \ \  \gamma\delta}.\]
Clearly, the (0,2) $S$-tensor fields $\a,\underline{\a}$ are symmetric. Note that if the Einstein equations $Ric(g)=0$ are satisfied, then all the remaining curvature components can be expressed in terms of the above components. 

\bigskip

\noindent\textbf{Remarks:}
\medskip

\textbf{1.} Recall, that the second fundamental form of a manifold $S$ embedded in a manifold $\m$ is defined to be the symmetric $(0,2)$ tensor field $II$ such that for each $x\in S$ we have
\[II_{x}:T_{x}S\times T_{x}S\rightarrow (T_{x}S)^{\perp}, \]
where the $\perp$ is defined via the decomposition $T_{x}\m=T_{x}S\oplus(T_{x}S)^{\perp}$. Specifically, if $X,Y\in T_{x}S$ then \[II_{x}(X,Y)=\big(\nabla_{X}Y\big)^{\perp}\] and hence\[\nabla_{X}Y=\nabb_{X}Y+II(X,Y).\]
Here $\nabb$ denotes the induced connection on $S$ (which is taken by projecting the spacetime connection $\nabla$ on $T_{x}S$).

 The $S$-tensor fields $\chi,\underline{\chi}$ give us the projections of $II_{AB}$ on $e_{3}$ and $e_{4}$, respectively. Indeed
 \[II(X,Y)=\chi(X,Y)e_{3}+\underline{\chi}(X,Y)e_{4}.\]
 
 For this reason we will refer to 
 $\chi,\underline{\chi}$ as the \textit{null second fundamental forms} of $S_{v}$.   One can easily verify that $\chi$ and $\underline{\chi}$ are symmetric (0,2) $S$-tensor fields. Indeed, a simple calculation shows that if $X,Y$ are $S$-tangent vector fields then
 \begin{equation*}
 \begin{split}
  \chi(X,Y)-\chi(Y,X)=g(e_{4},[X,Y]) ,\\ \underline{\chi}(X,Y)-\underline{\chi}(Y,X)=g(e_{3},[X,Y]).
 \end{split}
 \end{equation*}
 Hence, $\chi,\underline{\chi}$ are symmetric if and only if $[X,Y]\perp e_{3}$ and $[X,Y]\perp e_{4}$ and thus if and only if $\left\langle e_{3},e_{4}\right\rangle^{\perp}\ni [X,Y]\in TS_{v}$. The symmetry of $\chi,\underline{\chi}$ is thus equivalent to the integrability of the orthogonal complement $\left\langle e_{3},e_{4}\right\rangle^{\perp}$.
 
  Furthermore, we can decompose $\chi$ and $\underline{\chi}$ into their trace and traceless parts by
 \begin{equation}
\chi=\hat{\chi}+\frac{1}{2}(tr\chi)\gi, \ \ \ \
\underline{\chi}=\hat{\underline{\chi}}+\frac{1}{2}(tr\underline{\chi})\gi.
\label{tracechi}
\end{equation}
The trace of the $S$-tensor fields $\chi,\underline{\chi}$ (and more general $S$-tensor fields) is taken with respect to the induced metric $\gi$. The trace $tr\chi$ is known as the  \textit{expansion} and the component $\hat{\chi}$ is called the \textit{shear} of $S_{v}$ with respect to $\hh$.

\medskip

\textbf{2.} The $S$ 1-form $\zeta$ is known as the \textit{torsion}. 
If $\di$ denotes the exterior derivative on $S_{v}$ then the $S$ 1-forms $\eta,\underline{\eta}$ are related to $\zeta$ via
\begin{equation*}
\begin{split}\eta=\zeta+\di(\log\Omega),\ \ \ \ \underline{\eta}=-\zeta+\di{\log\Omega},
\end{split}
\end{equation*}
and hence
\[\zeta= \frac{1}{2}(\eta-\underline{\eta}), \ \ \ \ \di\log\Omega=\frac{1}{2}(\eta+\underline{\eta}).\]
The 1-forms $\eta,\underline{\eta}$ can be thought of as the torsion of the null hypersurfaces with respect to the geodesic vector fields. Indeed, the previous relations imply
\[\eta_{A}=\Omega^{2}g(\nabla_{A}L_{geod}, \underline{L}_{geod}).\]
Furthermore, we have
\begin{equation}
\begin{split}
[L,\underline{L}]=-2\Omega^{2}\zeta^{\sharp}.
\end{split}
\label{torsionll}
\end{equation}
Hence the torsion $\zeta$ is the obstruction to the integrability of the timelike planes $\left\langle e_{3},e_{4}\right\rangle$ orthogonal to the spheres $S_{v}$.

\medskip

\textbf{3. } Let $\li_{L}$ denote the projection of the Lie derivative $\lie_{L}$ onto the spheres $S_{v}$. The \textbf{first variation formula} then reads
\begin{equation}
\li_{L}\gi=2\Omega\chi,\ \ \ \ \ \li_{L}(\gi^{-1})=-2\Omega\chi^{\sharp\sharp}.
\label{1stvarform}
\end{equation}Note that we use the induced metric $\gi$ to raise and lower indices.
 Hence, since $[L,\partial_{\theta^{i}}]=0$ on $\hh$,
\[ L\big(\gi_{ij}\big)=2\Omega\chi_{ij} \]
and hence
\begin{equation}
L\sqrt{\gi}=\Omega tr\chi\sqrt{\gi}
\label{derdet}
\end{equation}
on $\hh$. Similarly, since $\partial_{u}=\underline{L}$ we obtain
\[\partial_{u}\sqrt{\gi}=\Omega tr\underline{\chi}\sqrt{\gi}\]
everywhere. Therefore,
\begin{equation}
\partial_{u}\partial_{v}\sqrt[4]{\gi}=\left[\frac{1}{2}\partial_{v}(\Omega tr\underline{\chi})+\frac{1}{4}(\Omega tr\chi)(\Omega tr\underline{\chi})\right]\cdot \sqrt[4]{\gi}
\label{needforw}
\end{equation}
on $\hh$.

\medskip

\textbf{4. } We denote by $\li_{L},\nabb_{L}$ the projection of $\lie_{L},\nabla_{L}$ on the sections $S_{v}$ and by $\lapp,\nabb$ the induced Laplacian and gradient of $(S_{v},\gi)$.

\subsection{The wave equation}
\label{sec:TheWaveEquation}
Let $\mathcal{D}$ be a double null foliation of $\m$ such that $\hh=\left\{u=0\right\}$ as defined in Section \ref{sec:TheDoubleNullFoliation}. Let $\mathcal{S}$ be the associated foliation on $\hh$. Using the canonical coordinates $(u,v,\theta^{1},\theta^{2})$ and recalling that $b^{A}=0$ on $\h$ we obtain
\begin{equation*}
\begin{split}
\Box_{g}\psi&=\frac{1}{\sqrt{g}}\partial_{a}\Big(\sqrt{g}\cdot g^{ab}\cdot \partial_{b}\psi\Big)\\
&=\frac{1}{\sqrt{g}}\partial_{\theta^{i}}\Big(\sqrt{g}\cdot g^{ib}\cdot \partial_{b}\psi\Big)+\frac{1}{\sqrt{g}}\partial_{v}\Big(\sqrt{g}\cdot g^{vb}\cdot \partial_{b}\psi\Big)+\frac{1}{\sqrt{g}}\partial_{u}\Big(\sqrt{g}\cdot g^{ub}\cdot \partial_{b}\psi\Big),
\end{split}
\end{equation*} 
where for simplicity we denote $\sqrt{g}=\sqrt{|\text{det}(g)|}$. Since the vector fields $\partial_{\theta^{i}},\partial_{v}$ are tangential to $\h$, it suffices to compute $g^{ib},g^{vb}$ only for the case where $b^{i}=0$. In this case we have
\begin{equation*}
g^{vv}=g^{iv}=g^{iu}=0, \ \ \ \ \ g^{ij}=\gi^{ij}, \ \ \ \  \ g^{vu}=-\Omega^{-2}\ : \text{ on }\h.
\end{equation*} 
We also need to compute $g^{ub}$ everywhere (not just on $\h$). We obtain
\begin{equation*}
g^{uu}=O(b^{2}), \ \ \ \ \ g^{uv}=-\Omega^{-2},\ \ \ \ \ g^{ui}=-\frac{b^{i}}{\Omega^{2}}.
\end{equation*}
Therefore, using that $b^{i}=0$, we obtain on $\h$: 
\begin{equation}
\begin{split}
\Omega^{2}\cdot &(\Box_{g}\psi)= \\ &=-\frac{1}{\sqrt{\gi}}\partial_{v}\Big(\sqrt{\gi}\cdot \partial_{u}\psi\Big)-\frac{1}{\sqrt{\gi}}\partial_{u}\Big(\sqrt{\gi}\cdot \partial_{v}\psi\Big)+\frac{1}{\sqrt{\gi}}\partial_{\theta^{i}}\Big(\sqrt{\gi}\cdot\big(\Omega^{2}\cdot\gi^{ij}\cdot\partial_{j}\psi\big)\Big)-(\partial_{u}b^{i})\cdot (\partial_{\theta^{i}}\psi)\\
&=-\frac{1}{\sqrt{\gi}}\bigg(2\sqrt[4]{\gi}\partial_{v}\partial_{u}\left(\sqrt[4]{\gi}\psi\right)-2\sqrt[4]{\gi}\Big(\partial_{v}\partial_{u}\sqrt[4]{\gi}\Big)\cdot\psi\bigg)+\divv\Big(\Omega^{2}\,\nabb\psi\Big)-(\partial_{u}b^{i})\cdot (\partial_{\theta^{i}}\psi),\\
\end{split}
\label{important}
\end{equation}
where $\divv $ denotes the divergence on the 2-surfaces with respect to the metric $\gi$.

Recalling \eqref{torsionll} we can rewrite the restriction of the wave equation of $\hh$ with respect to a double null foliation $\mathcal{D}$ as follows
\begin{equation}
-2\partial_{v}\pu(\sqrt[4]{\gi}\psi) +\sqrt[4]{\gi}\cdot\mathcal{Q}^{\mathcal{S}}\psi=0,
\label{wenull}
\end{equation}
on $\hh$, where the operator $\mathcal{Q}^{\s}:C^{\infty}(\hh)\rightarrow\mathbb{R}$ is defined by
\begin{equation}
\mathcal{Q}^{\mathcal{S}}\psi=\Omega^{2}\cdot\lapp\psi+(\nabb\Omega^{2}-2\Omega^{2}\cdot\zeta^{\sharp})\cdot\nabb\psi+w\cdot\psi,
\label{ovoperator}
\end{equation}
where 
\begin{equation}
w=2\frac{\pv\pu\sqrt[4]{\gi}}{\sqrt[4]{\gi}}=\bigg[\partial_{v}(\Omega tr\underline{\chi})+\frac{1}{2}(\Omega tr\underline{\chi})(\Omega tr\chi)\bigg] \text{ on }\hh.
\label{w}
\end{equation}
Clearly, the operator $\mathcal{Q}^{\mathcal{S}}$ depends on the foliation $\mathcal{S}$ of $\hh$. 
If we introduce the conformal factor $\phi$ then, using that $\partial_{u}\sqrt{\gi_{_{\mathbb{S}^{2}}}}=\partial_{v}\sqrt{\gi_{_{\mathbb{S}^{2}}}}=0$, the wave equation can be rewritten as
\begin{equation}
-2\partial_{v}\pu(\phi\cdot\psi) +\phi\cdot\mathcal{Q}^{\mathcal{S}}\psi=0, 
\label{wenull}
\end{equation}
on $\hh$. If we define the operator $\mathcal{F}^{\mathcal{D}}:C^{\infty}(\m)\rightarrow\mathbb{R}$ such that
\begin{equation}
\mathcal{F}^{\mathcal{D}}\psi=\frac{2}{\phi}\cdot \pv\pu(\phi\cdot\psi),
\label{defF}
\end{equation} 
then the wave equation on $\hh$ reads
\begin{equation}
\mathcal{F}^{\mathcal{D}}\psi=\mathcal{Q}^{\mathcal{S}}\psi.
\label{wenull1}
\end{equation}

\section{Gluing constructions and conservation laws}
\label{sec:TheGeneralCase}

\subsection{Elliptic theory on $\hh$ and the operator $\mathcal{O}^{\s}$}
\label{sec:EllipticTheoryOnHh}

 In this subsection, we will introduce some basic facts about elliptic operators on the sections $S_{v}$ of a foliation $\s$ of $\hh$ that will be needed for the proof of the main theorem.

The restriction of the operator $\mathcal{Q}^{\mathcal{S}}$ on a section $S_{v}$ gives rise to an elliptic operator\footnote{For simplicity, we will often use the same notation for the operator $\mathcal{Q}^{\mathcal{S}}$ and its restriction  $\mathcal{Q}^{\mathcal{S}}_{v}$ on the section $S_{v}$ of $\mathcal{S}$ on $\hh$. }
\begin{equation}
\mathcal{Q}^{\mathcal{S}}_{v}:=\mathcal{Q}^{\mathcal{S}}\left.\right|_{S_{v}}:C^{\infty}\big(S_{v}\big)\rightarrow \mathbb{R}
\label{restofq}
\end{equation} whose spectrum consists only of discrete eigenvalues with the property that the only limit (accumulation) point is infinity. Indeed, if we consider the operator
\[\mathcal{Q}^{\mathcal{S}}_{temp}\psi=
\mathcal{Q}^{\mathcal{S}}\psi-\frac{1}{\epsilon}\psi= \Omega^{2}\cdot\lapp\psi+(\nabb\Omega^{2}+2\Omega^{2}\zeta^{\sharp})\cdot\nabb\psi+w\cdot\psi-\frac{1}{\epsilon}\psi   \]
then
\begin{equation*}
\begin{split}
\int_{S_{v}}\mathcal{Q}^{\mathcal{S}}_{temp}\psi\cdot\psi=&-\int_{S_{v}}\psi\nabb\Omega^{2}\cdot\nabb\psi+\Omega^{2}\cdot|\nabb\psi|^{2}+\Big(\frac{1}{2}\lapp\Omega^{2}+Div(\Omega^{2}\zeta^{\sharp})-w+\frac{1}{\epsilon}\Big)\psi^{2}\\
\leq &
\!-\!\!\!\int_{S_{v}}\!(\Omega^{2}-\epsilon_{1})\cdot|\nabb\psi|^{2}+\Big(\frac{1}{2}\lapp\Omega^{2}+Div(\Omega^{2}\zeta^{\sharp})-w-\frac{1}{\epsilon_{1}}|\nabb\Omega^{2}|^{2}+\frac{1}{\epsilon} \Big)\psi^{2}
 \end{split}
\end{equation*}

We first take $\epsilon_{1}$ sufficiently small such that  $\Omega^{2}>\epsilon_{1}$ on $S_{v}$ and then take $\epsilon$ sufficiently small so the coefficient of $\psi^{2}$ is strictly positive. Hence, for these choices, the operator $\mathcal{Q}^{\mathcal{S}}_{temp}$ is negative definite and hence it has trivial kernel. By the Atiyah--Singer theorem  we have  $ind\big(\mathcal{Q}^{\mathcal{S}}_{temp}\big)=0$ and hence by the Fredholm alternative the operator $\mathcal{Q}^{\mathcal{S}}_{temp}$ is invertible. By the Poincar\'{e} inequality and Rellich's theorem the inverse is a compact operator $\big(\mathcal{Q}^{\mathcal{S}}_{temp}\big)^{-1}:L^{2}(S_{v})\rightarrow L^{2}(S_{v})$. The spectrum of this operator contains zero and only discrete eigenvalues whose limit point is zero. Denote this spectrum by $\sigma_{temp}$. It follows that the spectrum of $\mathcal{Q}^{\mathcal{S}}_{temp}$ is the set $\frac{1}{\sigma_{temp}}$ which consists of discrete eigenvalues whose only limit point is infinity. Then, the spectrum of $\mathcal{Q}^{\mathcal{S}}$, restricted at the section $S_{v}$, is precisely the set $\sigma_{v}=\frac{1}{\sigma_{temp}}+\frac{1}{\epsilon}$.

Note also that if $f\in C^{\infty}(S_{v})$  then the equation 
\begin{equation}
\mathcal{Q}^{\mathcal{S}}\psi=f_{v},
\label{tosolveelliptic}
\end{equation} 
has a solution $\psi$ on $S_{v}$ if and only if $f_{v}$ lies in the orthogonal complement of the kernel of the adjoint of $\mathcal{Q}^{\mathcal{S}}$ with respect to the space $\big(S_{v},\gi\big)$. We have the following definition
\begin{definition}
Let $\s$ be a foliation of regular null hypersurface $\hh$ of a Lorentzian manifold $(\m,g)$, as defined in Section \ref{sec:TheDoubleNullFoliation}. We define the operator $\o^{\s}:C^{\infty}(\hh)\rightarrow\mathbb{R}$ given by 
\begin{equation}
\begin{split}
\o^{\mathcal{S}}\psi=& \Omega^{2}\cdot\lapp\psi+\left[\nabb\Omega^{2}+2\Omega^{2}\cdot\zeta^{\sharp}\right]\cdot\nabb\psi\\&+\left[2\divv\,\Big(\Omega^{2}\cdot\zeta^{\sharp}\Big)+\partial_{v}(\Omega tr\underline{\chi})+\frac{1}{2}(\Omega tr\underline{\chi})(\Omega tr\chi)\right]\cdot\psi.
\label{adjoint}
\end{split}
\end{equation}
We also denote by 
\begin{equation}
\mathcal{O}^{\mathcal{S}}_{v}:=\mathcal{O}^{\mathcal{S}}\left.\right|_{S_{v}}:C^{\infty}(S_{v})\rightarrow \mathbb{R}
\label{restofo}
\end{equation}
the restriction of $\mathcal{O}^{\mathcal{S}}$ on a section $S_{v}$. 
\label{definitiono}
\end{definition}
The operator $\mathcal{O}^{\mathcal{S}}_{v}$ is the adjoint of $\mathcal{Q}^{\mathcal{S}}_{v}$ with respect to the space $\big(S_{v},\gi\big)$. Hence we have 
\[ Im\big(\mathcal{Q}^{\s}_{v}\big)=\left(Ker(\o^{\s}_{v})\right)^{\perp}. \]
Therefore, the equation \eqref{tosolveelliptic} as a solution if and only if 
\begin{equation}
f_{v}\in \big(Ker(\mathcal{O}^{\mathcal{S}}_{v}) \big)^{\perp}.
\label{eq:integrability}
\end{equation}
where the operator $\mathcal{O}^{\mathcal{S}}$ is the adjoint of $\mathcal{Q}^{\mathcal{S}}$ with respect to the space $\big(S_{v},\gi\big)$ and is given by 

If $f_{v}$ depend smoothly on $v$ for all $v\in I=[v_{1},v_{2}]$  and \eqref{eq:integrability} is satisfied in $I$, then $\psi$ can be chosen to depend smoothly on $v$ too.

Finally, if $S_{v}$ are endowed with the standard unit metric $\gi_{\mathbb{S}^{2}}$ (see Section \ref{sec:TheDoubleNullFoliation}) and $f(v,\theta^{1},\theta^{2}):\hh\rightarrow \mathbb{R}$ is a smooth function then 
\[f(v,\theta^{1},\theta^{2})=\sum_{ml}f_{ml}(v)\cdot Y_{ml}(\theta^{1},\theta^{2}),\]
where $Y_{ml}$ are the standard spherical harmonics.

\subsection{The gluing construction}
\label{sec:GeneralizedConservationLaws}

Let $\mathcal{D}$ be a regular double null foliation such that $\hh=\left\{u=0\right\}$. Clearly $\mathcal{D}$ defines a foliation $\s=\big(S_{v}\big)_{v\in\mathbb{R}}$ of $\hh$. Let $S_{0}$ and $S_{1}$ be two sections of $\s$. We will first show that we can always glue data on $S_{0}$ to data on $S_{1}$ in the sense of Definition \ref{firstordergluingdefinition} if the operator $\mathcal{Q}^{\mathcal{S}}_{v}$ is surjective for some $v\in[0,1]$.

\begin{proposition}
Let $\s$ be a foliation of a regular null hypersurface $\hh$ of a four-dimensional Lorentzian manifold $(\m,g)$, as defined in Section \ref{sec:NullFoliationsandOpticalFunctions}. Let also $\o_{v}^{\s}$ be the elliptic operator given by \eqref{adjoint}. If there is $v_{0}\in[0,1]$ such that 
\[Ker(\mathcal{O}_{v_{0}}^{\s})= \left\{0\right\},\]
i.e.~if $0$ is \textbf{not} an eigenvalue of $\mathcal{O}^{\s}_{v_{0}}$, then we can glue arbitrary data at $S_{0}$ to arbitrary data at $S_{1}$ in the sense of Definition \ref{firstordergluingdefinition}. 
\label{pert1prop}
\end{proposition}
\begin{proof}

Suppose that the spectrum $\sigma(\o_{v_{0}}^{\s})$ does not contain zero for some $v_{0}\in[0,1]$. Then zero is in the resolvent set $\rho(\o_{v_{0}}^{\s})$. By Kato's upper semicontinuity of spectrum (see \cite{kato}), the spectrum $\sigma(\o_{v}^{\s})$ of the operator $\sigma(\o_{v}^{\s})$ also does not contain zero for all $v\in[v_{0}-\epsilon,v_{0}+\epsilon]$ with $\epsilon>0$ sufficiently small. In other words, since the operators $\mathcal{O}^{\s}_{v}$ depend smoothly on $v$, there is a sufficiently small $\epsilon$ such that the resolvent set of the operators $\o_{v}^{\s}$ contains zero for all $v\in[v_{0}-\epsilon,v_{0}+\epsilon]$. In view of previous comments, these operators are also surjective.

If we integrate the wave equation \eqref{wenull} along the null generators of $\hh$ we obtain
\begin{equation}
2\pu(\phi\cdot\psi)_{\big|_{S_{1}}}-2\pu(\phi\cdot\psi)_{\big|_{S_{0}}} =\int_{0}^{1}\phi\cdot\mathcal{Q}_{v}^{\s}\psi \ dv.
\label{malista}
\end{equation}
In the context of our gluing problem the first two terms are given. We smoothly extend $\psi$  in the cylinders
\[(v,\theta^{1},\theta^{2})\in[0,v_{0}-\epsilon]\times\mathbb{S}^{2}, \ \ \ (v,\theta^{1},\theta^{2})\in[v_{0}+\epsilon,1]\times\mathbb{S}^{2}.\]
such that $\psi$ vanishes at all orders at the spheres $S_{v_{0}-\epsilon}$ and $S_{v_{0}+\epsilon}$. 
Then equation \eqref{malista} is satisfied if 
\[\int_{v_{0}-\epsilon}^{v_{0}+\epsilon}\phi(v,\theta^{1},\theta^{2})\cdot\mathcal{Q}_{v}^{\s}\psi(v,\theta^{1},\theta^{2}) \ dv =\rho(\theta^{1},\theta^{2}),     \]
where $\rho$ is a given (prescribed) function of the sphere (which depends only on the initial data at $S_{0}$ and $S_{1}$ and  the extension of $\psi$ in the complement of the cylinder for which $v\in[v_{0}-\epsilon,v_{0}+\epsilon]$. )

We consider a smooth function $f(v,\theta^{1},\theta^{2})$,  such that  
\[\int_{v_{0}-\epsilon}^{v_{0}+\epsilon}\phi\big(v,\theta^{1},\theta^{2}\big)\cdot f\big(v,\theta^{1},\theta^{2}\big)\ dv=\rho\big(\theta^{1},\theta^{2}\big) \]
and $f$  vanishes to all orders at $v=v_{0}-\epsilon$ and $v=v_{0}+\epsilon$ (such a function clearly exists). Then, we simply have to solve the equations
\[\mathcal{Q}_{v}^{\s}\psi=f\big(v,\cdot,\cdot\big)\]
on $S_{v}$ for all $v\in[v_{0}-\epsilon,v_{0}+\epsilon]$. This is clearly possible in view of the fact that the operators $\o_{v}^{\s}$ (and hence $\mathcal{Q}_{v}^{\s}$) are all invertible and the comments in the Section \ref{sec:EllipticTheoryOnHh}. Moreover, $\psi$ must necessarily vanish to all orders at $v=v_{0}-\epsilon$ and $v=v_{0}+\epsilon$ and hence extends to a smooth function in the cylinder where $\big(v,\theta^{1},\theta^{2}\big)\in [0,1]\times\mathbb{S}^{2}$. 

\end{proof}

Hence, in order to have a conservation law along $\hh$ we must have that $Ker\big(\o_{v}^{\s}\big)\neq \left\{0\right\}$ for all $v$  (otherwise we can perform gluing). The above result, however, does not exclude the possibility of gluing \textit{general} characteristic data even if $Ker\big(\o_{v}^{\s}\big)\neq \left\{0\right\}$ for all $v$. For we have the following general result

\begin{theorem}
Let $\mathcal{S}$ be a foliation of a regular null hypersurface $\hh$ of a four-dimensional Lorentzian manifold $(\m,g)$, as defined in Section \ref{sec:TheDoubleNullFoliation}. Let also $S_{0}$ and $S_{1}$ be two sections of $\s$. Then

\begin{enumerate}
	\item We can glue (to first order) general data on $S_{0}$ to general data on $S_{1}$ in the sense of Definition \ref{firstordergluingdefinition} if and only if $\hh$ does not admit conservation laws with respect to $\s$ in the sense of Definition \ref{definitionconservationlaw}. If $\hh$ admits conservation laws with respect to $\s$, then we can glue characteristic data if and only if their associated charges are equal, i.e. if and only if the data at $S_{0}$ and $S_{1}$ are such that
	\[\int_{S_{0}}Y^{\s}\big(\phi\cdot\psi\big)\cdot \Theta^{\s}\,d\mu_{_{\mathbb{S}^{2}}}=\int_{S_{1}}Y^{\s}\big(\phi\cdot\psi\big)\cdot \Theta^{\s}\,d\mu_{_{\mathbb{S}^{2}}}, \]
	for all $\Theta^{\s}\in  \mathcal{W}^{\s}$, where $\mathcal{W}^{\s}$ is the kernel of the conservation laws as defined in Section \ref{sec:ConservationLawsForTheWaveEquations}.  Here $\phi$ denotes the conformal factor of the sections of an associated double null foliation $\mathcal{D}$.
	
	\item The null hypersurface $\hh$ admits (first-order) conservation laws with respect to $\s=\big(S_{v}\big)_{v\in\mathbb{R}}$  in the sense of Definition \ref{definitionconservationlaw} if and only if there is a non-trivial linear space  $\mathcal{U}^{\mathcal{S}}\subset  \mathcal{V}_{\hh}$, where $\mathcal{V}_{\hh}$ is the linear space defined in \eqref{linearspace}, such that 
		\[\o^{\mathcal{S}}\left(\frac{1}{\phi}\cdot\Theta^{\s}\right)=0 \text{ on }\hh,\text{ for all }\Theta^{\s}\in \mathcal{U}^{\mathcal{S}}.\] Furthermore, the kernel of the conservation laws satisfies $\mathcal{W}^{\mathcal{S}}=\mathcal{U}^{\mathcal{S}}$, and moreover, $dim\, \mathcal{W}^{\mathcal{S}}=dim\, \mathcal{U}^{\mathcal{S}}<\infty$.

	\end{enumerate}
\label{theorem}
\end{theorem}

\begin{remark}
The conserved charges are independent of the choice of the diffeomorphism $\Phi$ defined by \eqref{phidiffeo}.  
\label{rem1giatheorema}
\end{remark}

\begin{remark}
Gluing is \textbf{not} always possible even if we allow to freely choose the initial data in an angular neighborhood $\mathcal{X}$ on  $S_{0}$ or $S_{1}$. Indeed, if the function  $\Theta^{\s}$ vanishes in $\mathcal{X}$, then the charges cannot change value even if we change the data in that region $\mathcal{X}$. Hence if the charges at $S_{0}$ and $S_{1}$ do not initially coincide then they will not coincide even after changing the data at $\mathcal{X}$. 
\label{re2}
\end{remark}

Before we give the proof of the Theorem \ref{theorem} we present some lemmata. The first one concerns the kernels of the elliptic operators $\o_{v}^{\s}$.
\begin{lemma}
\textbf{(Variation analysis of the kernel of elliptic operators)} There is an upper bound for the dimension of the kernel $K(v)\subset L^{2}\big(\mathbb{S}^{2}\big)$ of the operator $\o_{v}^{\s}$, given by \eqref{adjoint}, for $v\in[0,1]$. Moreover, there is a dense set $\mathcal{Z}\subseteq [0,1]$ of point $x$ for which there is an open neighborhood $V_{x}$ containing $x$ such that $K(v)$ varies smoothly for $v\in V_{x}$ (and hence in particular the dimension $\text{dim}\big(K(v)\big)$ is constant for all $v\in V_{x}$). 
\label{lemmagiaelliptickernel}
\end{lemma}
\begin{proof}
Following the idea of Section \ref{sec:EllipticTheoryOnHh} we have that for sufficiently large  $\lambda >0$ the operator $\o_{v}^{\s}-\lambda\cdot I$ has a compact inverse $\mathcal{C}_{v}:L^{2}(\mathbb{S}^{2})\rightarrow H^{1}(\mathbb{S}^{2})\subset L^{2}(\mathbb{S}^{2})$. By the continuity of the resolvent theorem (see \cite{extremumproblemsbook}, Chapter 2) we obtain that the operators $\mathcal{C}_{v}$ vary continuously in $v$ with respect to the topology of the space $\mathcal{L}\big(L^{2},L^{2} \big)$. 
The kernel $K(v)$ of $\o_{v}^{\s}$ coincides with the kernel of the operator 
\begin{equation}
\mathcal{P}_{v}=\mathcal{C}_{v}-\frac{1}{\lambda}\cdot I:L^{2}(\mathbb{S}^{2})\rightarrow  L^{2}(\mathbb{S}^{2}),
\label{operatorp}
\end{equation}
which also varies continuously in $v$. 
Since $\mathcal{C}_{v}$ is compact it is easily seen that the operator \[\Big. \mathcal{P}_{v}\Big|_{\big(K(v)\big)^{\perp}}   :\big(K(v)\big)^{\perp}\rightarrow  L^{2}(\mathbb{S}^{2})\]
is bounded from below. This implies that the mapping
\[K:[0,1]\rightarrow K(v):=ker\big(\o_{v}\big)\subset L^{2}\big(\mathbb{S}^{2}\big)  \]
is upper semicontinuous, i.e.~if $B(1)$ is the unit ball in $L^{2}\big(\mathbb{S}^{2}\big)$ then for all $\epsilon>0$ there is $\delta>0$ such that if $|v-v_{0}|<\delta$ then 
\[ K(v)\cap B(1)\subset B_{\epsilon}\Big( K(v_{0})\cap B(1)  \Big), \] 
where $B_{\epsilon}(S)$ denotes the set of points who distance from $S$ is at most  $\epsilon$. It thus follows that 
\begin{equation}
\limsup_{n}\text{dim}\big(K(v_{n})\big)\leq \text{dim}\big(K(v_{0})\big).
\label{simantikodiastasi}
\end{equation}
Define now the sets
\begin{equation}
A_{n}=\Big\{v\in[0,1]\, :\, \text{dim}\big(K(v)\big)\geq n\Big\}.
\label{tasinolaa}
\end{equation}
In view of \eqref{simantikodiastasi} the sets $A_{v}$ are closed in [0,1]. Moreover, $A_{n+1}\subset A_{n}$. By the compactness of $[0,1]$ it follows that there is a $n_{0}\in \mathbb{N}$ such that $A_{n}=\emptyset$ for all $n\geq n_{0}$. Hence, there is an upper bound on the dimension of the kernel $K(v)$ for all $v\in [0,1]$. We consider next the following sets
\begin{equation}
B_{n}=A_{n}/A_{n-1}=\Big\{v\in[0,1]\, :\, \text{dim}\big(K(v)\big)=n\Big\}.
\label{tasinolab}
\end{equation}
Clearly, 
\begin{equation}
\bigcup_{n=0}^{n_{0}}B_{n}=[0,1].
\label{tasinolab1}
\end{equation}
The set $B_{0}=[0,1]/A_{1}$ is open in $[0,1]$ and hence $B_{0}=\text{int}\big(B_{0}\big)$. We will show that for all $1\leq n\leq n_{0}$ we have 
\begin{equation}
\text{int}\big(B_{n}\big)=B_{n}/\text{clos}\big(B_{0}\cup...\cup B_{n-1}\big).
\label{interiorofb}
\end{equation}
The inclusion $\text{int}\big(B_{n}\big)\subseteq B_{n}/\text{clos}\big(B_{0}\cup...\cup B_{n-1}\big)$ is trivial. If now there is $x\in B_{n}/\text{clos}\big(B_{0}\cup...\cup B_{n-1}\big)$ such that $x\notin \text{int}\big(B_{n}\big)$ then there is a sequence $y_{k}\rightarrow x$ with $y_{k}\notin B_{n}$. Clearly, $y_{k}$ cannot have an infinite subsequence in either $B_{0}\cup ...\cup B_{n-1}$ since otherwise $x\in \text{clos}\big(B_{0}\cup...\cup B_{n-1}\big)$.   Hence, $y_{k}\in A_{n+1}$ and since $A_{n+1}$ is closed we have $x\in A_{n+1}$ and hence $x\notin B_{n}$, contradiction. 

Define the set
\begin{equation}
\mathcal{Z}=\bigcup_{n=1}^{n_{0}}\text{int}\big(B_{n}\big).
\label{eq:thesetz}
\end{equation}
The set $\mathcal{Z}$ is open in $[0,1]$ and dense. Indeed, in view of \eqref{tasinolab1} and \eqref{interiorofb} we have
\[\text{clos}(\mathcal{Z})= \text{clos}\big(B_{0}\cup...\cup B_{n_{0}}\big)=[0,1] .  \]
In other words, there is a dense set of points $x\in [0,1]$ for which there is an open neighborhood $V_{x}$ such that $\text{dim}\big(K(v)\big)$ is constant for all $v\in V_{x}$. It remains to show that $K(v)$ varies smoothly in $v$ for $v\in V_{x}$, i.e.~the curve $V_{x}\ni x\mapsto K(v)$ is smooth in the Grassmannian $Gr(L^{2},n)$ where $n=\text{dim}\big(K(v)\big)$. Using an adaptation of the aforementioned result of \cite{extremumproblemsbook} and the fact that the coefficients of $\o_{v}^{\s}$ depend smoothly in $v$ one can show that $\mathcal{C}_{v}$, and hence $\mathcal{P}_{v}$, varies smoothly in $v$ in the space $\mathcal{L}\big(L^{2},L^{2}\big)$. In view of the fact that $K(v)$ is upper semicontinuous and has constant dimension we obtain that if $v_{0}\in V_{x}$ then  
\[\Big. \text{proj}\Big|_{K(v_{0})\cap B(1)}\Big(K(v)\cap B(1)\Big)= K(v_{0})\cap B(1)  \]
for all $v$ sufficiently close to $v_{0}$. Given $x_{v_{0}}\in K(v_{0})\cap B(1)$ there is a unique $x_{v}\in B_{\epsilon}\Big(K(v)\cap B(1) \Big)$ such that $x_{v}=x_{v_{0}}+a_{v}^{\perp}$. Since $\mathcal{P}_{v}=\mathcal{P}_{v_{0}}+(v-v_{0})\cdot Z+O\big((v-v_{0})^{2}\big)$ and  $\mathcal{P}_{v}(x_{v})=0$ we obtain $\mathcal{P}_{v_{0}}(a_{v}^{\perp})=(v_{0}-v)\cdot Zx_{v}+O(v^{2})  $ and since $\Big. \mathcal{P}_{v_{0}}\Big|_{\big(K(v)\big)^{\perp}}:\big(K(v)\big)^{\perp}\rightarrow Im(\mathcal{P}_{0})$ has a bounded inverse we obtain that $a_{v}^{\perp}\rightarrow 0$ in $L^{2}$ as $v\rightarrow v_{0}$ in a differentiable manner. This shows that $K(v)$ varies differentiably in $v$. Similary we can show that $K(v)$ is smooth in $v$.

\end{proof}

If there is $v_{0}\in[0,1]$ such that $Ker\big(\o^{\s}_{v}\big)= 0$ then the Theorem \ref{theorem} follows from Proposition \ref{pert1prop}. In particular, in this case we have $dim\,\mathcal{W}^{\s}=dim\,\mathcal{U}^{\s}=0$. We assume that $Ker\big(\o_{v}^{\s}\big)\neq 0$ for all $v\in[0,1]$, that is $B_{0}=\emptyset$, where $B_{0}$ is defined by \eqref{tasinolab}. Then according to the above lemma there is a dense set of points $x$ which have an open neighborhood $V_{x}$ in which the kernels $K(v)$ vary smoothly.   In each of these intervals, we can find a smoothly varying in $v$ basis 
\[\mathcal{B}_{v}^{\s}=\left\{ \big(E_{1}^{\s}\big)_{v}, \big(E_{2}^{\s}\big)_{v},\cdots, \big(E_{i}^{\s}\big)_{v}, \ i=dimKer\big(\o_{v}^{\s}\big) \right\}\]
of $K(v)=Ker\big(\o_{v}^{\s}\big)$. We  next localize in each of these intervals.

Let $v_{0}\in V_{x}$ for some $x$ as above. We can smoothly extend $\psi\left.\right|_{S_{0}}$ and $\psi\left.\right|_{S_{1}}$ in the cylinders $[0,v_{0}-\epsilon]\times\mathbb{S}^{2}$ and $[v_{0}+\epsilon,1]\times \mathbb{S}^{2}$, where $\epsilon>0$ is sufficiently small such that $v_{0}-\epsilon,v_{0}+\epsilon\in V_{x}$, such that $\psi$ vanishes to infinite order  on $\hh$  at $\left\{v_{0}-\epsilon\right\}\times\mathbb{S}^{2}$ and $\left\{v_{0}+\epsilon\right\}\times\mathbb{S}^{2}$.  Then, $\o^{\s}_{v}\psi$ is also known in the union of these two cylinders and hence, by \eqref{wenull1}, $\mathcal{F}^{\mathcal{D}}\psi$, given by \eqref{defF}, is  known there. In fact $\mathcal{F}^{\mathcal{D}}\psi$  vanishes to infinite order on $\hh$  at $\left\{v_{0}-\epsilon\right\}\times\mathbb{S}^{2}$ and $\left\{v_{0}+\epsilon\right\}\times\mathbb{S}^{2}$.  We wish to extend $\psi$ smoothly in $[0,1]\times\mathbb{S}^{2}$ such that $\psi$ solves \eqref{wenull1} and the transversal derivatives $\pu\psi\left.\right|_{S_{0}}$ and $\pu\psi\left.\right|_{S_{1}}$ agree with the given characteristic data.

 We will do so by first smoothly extending $\mathcal{F}^{\mathcal{D}}\psi$ everywhere in $[0,1]\times\mathbb{S}^{2}$ such that the above are satisfied and, using \eqref{wenull1}, we can solve with respect to $\psi$. In other words, gluing is possible if we can extend $\mathcal{F}^{\mathcal{D}}\psi$, or equivalently $\mathcal{F}^{\mathcal{D}}\psi\cdot\phi$, in $[v_{0}-\epsilon,v_{0}+\epsilon]\times\mathbb{S}^{2}$ such that the following conditions hold:

\medskip

\noindent{\underline{\textbf{The conditions 1--3:}}}

\medskip

\begin{enumerate}
	\item \textbf{\underline{Smoothness on $\hh$}:} The function $\mathcal{F}^{\mathcal{D}}\psi$, or equivalently $\mathcal{F}^{\mathcal{D}}\psi\cdot\phi$, vanishes to all orders on $\hh$  at $\left\{v_{0}-\epsilon\right\}\times\mathbb{S}^{2}$ and $\left\{v_{0}+\epsilon\right\}\times\mathbb{S}^{2}$. 
	\item \textbf{\underline{Gluing for the transerval derivative $\pu\psi$}:} In view of \eqref{defF} and \eqref{wenull1}, $\mathcal{F}^{\mathcal{D}}\psi$ satisfies on $\hh$:  \[\displaystyle\int_{v_{0}-\epsilon}^{v_{0}+\epsilon}\!\!\!\big(\mathcal{F}^{\mathcal{D}}\psi\big)(v,\theta^{1},\theta^{2})\cdot \phi(v,\theta^{1},\theta^{2})\, dv=\rho(\theta^{1},\theta^{2}),\] where $\rho$ is a given (prescribed) function of the sphere (which depends only on the initial data at $S_{0}$ and $S_{1}$ and also the extension of $\psi$ in the complement of the cylinder for which $v\in[v_{0}-\epsilon,v_{0}+\epsilon]$).
	
	\item \textbf{\underline{Integrability (orthogonality) condition}:} Given $\mathcal{F}^{\mathcal{D}}\psi$ we can solve with respect to $\psi$ (i.e. ``invert'' the operator $\mathcal{F}^{\mathcal{D}}$) if, using \eqref{wenull1} and the comments of Section \ref{sec:EllipticTheoryOnHh}, 
	\[\mathcal{F}^{\mathcal{D}}\psi \in  Im\big(\mathcal{Q}^{\s}_{v}\big)=\left(Ker(\o^{\s}_{v})\right)^{\perp},   \]
	or equivalently,
	\begin{equation}
	\displaystyle\int_{S_{v}}\mathcal{F}^{\mathcal{D}}\psi\cdot \big(E_{n}^{\s}\big)_{v}\ d\mu_{_{\gi}}=0
	\label{1inte}
	\end{equation} for all $v\in[v_{0}-\epsilon,v_{0}+\epsilon]$ and $n=1,..., dimKer\big(\o_{v}^{\s}\big)$.
	\end{enumerate}
	Note that since
	\begin{equation*}
\begin{split}
\int_{S_{v}}\mathcal{F}^{\mathcal{D}}\psi\cdot \big(E_{n}^{\s}\big)_{v}\ d\mu_{_{\gi}}=&\int_{S_{v}}2\frac{\sqrt[4]{\gi_{\mathbb{S}^{2}}}}{\sqrt[4]{\gi}}\cdot\pv\pu(\phi\cdot\psi)\cdot \big(E_{n}^{\s}\big)_{v}\sqrt{\gi}\, d\theta^{1}\, d\theta^{2}\\=&\int_{S_{v}}2\frac{\sqrt[4]{\gi}}{\sqrt[4]{\gi_{\mathbb{S}^{2}}}}\cdot\pv\pu(\phi\cdot\psi)\cdot \big(E_{n}^{\s}\big)_{v}\sqrt{\gi_{\mathbb{S}^{2}}}\, d\theta^{1}\, d\theta^{2}\\=&\int_{S_{v}}\Big(\mathcal{F}^{\mathcal{D}}\psi\cdot\phi\Big)\cdot\Big(\big(E_{n}^{\s}\big)_{v}\cdot \phi\Big)\, d\mu_{_{\mathbb{S}^{2}}},
\end{split}
\end{equation*}
the integrability condition \eqref{1inte} is equivalent to the following:
\begin{equation}
\int_{S_{v}}\Big(\mathcal{F}^{\mathcal{D}}\psi\cdot\phi\Big)\cdot\Big(\big(E_{n}^{\s}\big)_{v}\cdot \phi\Big)\, d\mu_{_{\mathbb{S}^{2}}}=0,
\label{integrability}
\end{equation}
for all $n=1,..., dimKer\big(\o_{v}^{\s}\big)$. As we shall see, splitting $\phi$ in both terms is very important as it also is the fact that the above integral is with respect to the standard unit metric on $S_{v}$. We define the functions
\begin{equation}
G_{n}=\big(E_{n}^{\s}\big)_{v}\cdot\phi, 
\label{g0}
\end{equation}
for $n=1,\cdots, dimKer\big(\o_{v}^{\s}\big)$ and thus \eqref{1inte} is equivalent to 
\begin{equation}
\int_{S_{v}}\Big(\mathcal{F}^{\mathcal{D}}\psi\cdot\phi\Big)\cdot\big(G_{n}\big)\, d\mu_{_{\mathbb{S}^{2}}}=0,
\label{integrability}
\end{equation}
for $n=1,\cdots, dimKer\big(\o_{v}^{\s}\big)$. 
\bigskip

\noindent\underline{ \textbf{A special case:}} $dim Ker\big(\o_{v}^{\s}\big)=1$ for all $v\in[0,1]$

\medskip

Before we consider the  general case let us first consider the special case for which $dim Ker\big(\o_{v}^{\s}\big)=1$ for all $v\in[0,1]$ (and hence $V_{x}=[0,1]$).  This will make our argument clear. We simplify the notation by denoting $G=G_{1}$.

We  decompose $\mathcal{F}^{\mathcal{D}}\psi\cdot\phi,\,G$ in (standard) angular frequencies. Let
\begin{equation}
\big(\mathcal{F}^{\mathcal{D}}\psi\cdot\phi\big)(v,\theta^{1},\theta^{2})=\sum_{ml}F_{ml}(v)\cdot Y^{ml}(\theta^{1},\theta^{2})
\label{f1}
\end{equation}
and 
\begin{equation}
G(v,\theta^{1},\theta^{2})=\sum_{ml}G_{ml}(v)\cdot Y^{ml}(\theta^{1},\theta^{2}),
\label{g1}
\end{equation}
where $Y^{ml}$ denote the standard spherical harmonics on $\mathbb{S}^{2}$. 

We can glue to first order characteristic data on $S_{0}$ and $S_{1}$ if there exist smooth functions  $F_{ml}(v)\!:\![v_{0}-\epsilon,v_{0}+\epsilon]\rightarrow\mathbb{R}$ such that the following conditions are satisfied
\begin{enumerate}
\item \underline{Smoothness on $\hh$}:
The functions $F_{ml}(v)$ vanish to all orders  at $v=v_{0}-\epsilon$ and $v=v_{0}+\epsilon$ for all $m,l$.
\item 
\underline{Gluing for the transversal derivative $\partial_{u}\psi$}:
The integrals
\[\int_{v_{0}-\epsilon}^{v_{0}+\epsilon}F_{ml}(v)\, dv\]
 are all given. 
 \item 
\underline{Orthogonality condition}: The integrability condition \eqref{integrability} is satisfied:
\begin{equation*}
\begin{split}
\sum_{ml} F_{ml}(v)\cdot G_{ml}(v)=0.
\end{split}
\end{equation*}
\end{enumerate}

\underline{\textbf{Case I}}:\medskip  \\ 
For all $m,l$ we have $G_{ml}(v)=\textbf{G}(v)\cdot c_{ml}$, for all $v\in[0,1]$, for some function $\textbf{G}:[0,1]\rightarrow\mathbb{R}$ and some constant non-zero $l^{2}$ sequence $c_{ml}$.  That is to say, for all $m,l,m',l'$ we have $\frac{G_{ml}(v)}{G_{m'l'}(v)}=c_{mlm'l'}$, where $c_{mlm'l'}$ is a constant. Then, the condition 3  becomes
\[\sum_{ml}F_{ml}(v)\cdot c_{ml}=0. \]
  The above equation implies that the functions $F_{ml}(v)$ are linearly dependent and hence condition  2 \textbf{cannot} be satisfied \textit{in general}. In fact, in this case we have that $G(v,\theta^{1},\theta^{2})=\textbf{G}(v)\cdot{\Theta}(\theta^{1},\theta^{2}), $ for some function ${{\Theta}}$ which is constant along the null generators, i.e.~$\Theta\in\vh$. Then, by \eqref{g0}, we have the splitting 
\[  \big(E^{\s}_{v}\cdot\phi\big)(v,\theta^{1},\theta^{2})= \mathbf{G}(v)\cdot\Theta(\theta^{1},\theta^{2}) \]
which shows that 
\[E_{v}^{\s}(\theta^{1},\theta^{2})= \mathbf{G}(v)\cdot \frac{1}{\phi}\cdot \Theta(\theta^{1},\theta^{2}), \]
and since the function $\mathbf{G}$ does not depends on the angular coordinates we have that 
\[\frac{1}{\phi}\cdot\Theta\in Ker(\o_{v}^{\s}),  \]
for all $v\in[0,1],$ as required. Furthermore,  by  \eqref{integrability} and \eqref{defF} we have 
\[0=\int_{S_{v}}\big(\mathcal{F}^{\mathcal{D}}\psi\cdot\phi\big)\cdot\Theta\, d\mu_{_{\mathbb{S}^{2}}}=\int_{S_{v}}\Big(\pv\pu(\phi\cdot\psi)\Big)\cdot\Theta\, d\mu_{_{\mathbb{S}^{2}}} \]
and, since the function $\Theta$ and the measure of integration do not depend on $v$, we obtain that the quantity
\begin{equation}
\int_{S_{v}}\big(\pu(\phi\cdot\psi)\big)\cdot\Theta\ d\mu_{_{\mathbb{S}^{2}}}
\label{eq:conserved}
\end{equation}
is conserved, i.e.~independent of $v$. Clearly the above conservation law is an obstruction to gluing of general initial data on $S_{0}$ to general initial data on $S_{1}$. 

We will next show that this conservation law is the \textbf{only} obstruction to gluing. Indeed suppose that  the initial data on $S_{0}$ and $S_{1}$ are such that 
\[\int_{S_{0}}\big(\pu(\phi\cdot\psi)\big)\cdot\Theta\ d\mu_{_{\mathbb{S}^{2}}}=\int_{S_{1}}\big(\pu(\phi\cdot\psi)\big)\cdot\Theta\ d\mu_{_{\mathbb{S}^{2}}}.\]
In this case we need to construct functions $F_{ml}:[v_{0}-\epsilon,v_{0}+\epsilon]\rightarrow\mathbb{R}$ such that 
\begin{enumerate}
\item $F_{ml}(v)$ vanishes to infinite order at $v=v_{0}-\epsilon$ and $v=v_{0}+\epsilon$ for all $m,l$.
\item \begin{equation}
\sum_{ml}F_{ml}(v)\cdot c_{ml}=0
\label{in1}
\end{equation} for $v\in[v_{0}-\epsilon,v_{0}+\epsilon]$, where $c_{ml}$ is a constant non-zero $l^{2}$ sequence. 

\item The integrals
\[ \mathbf{I}_{ml}=\int_{v_{0}-\epsilon}^{v_{0}+\epsilon}F_{ml}(v)\, dv \]are all prescribed such that    
\begin{equation}
\sum_{ml}\mathbf{I}_{ml}\cdot c_{ml} =0, 
\label{in2}
\end{equation} 
since
\[Char(S_{v_{0}-\epsilon})=Char(S_{0})=Char(S_{1})=Char(S_{v_{0}+\epsilon}),\]
and using \eqref{defF} and \eqref{malista}, where we denote 
\[char(S_{v})=\int_{S_{v}}\big(\pu(\phi\cdot\psi)\big)\cdot\Theta\ d\mu_{_{\mathbb{S}^{2}}}.\] \end{enumerate}

For simplicity we rename the sequences $F_{ml},G_{ml,}c_{ml}$ as $F_{i},G_{i},c_{i},i\geq 0$. We assume without loss of generality that $c_{0}\neq 0$. We construct the functions $F_{i}, i\geq 1,$ such that condition 1 is satisfied and such that the integrals $\int_{v_{0}-\epsilon}^{v_{0}+\epsilon}F_{i}(v)\,dv$ agree with the prescribed values (condition 3).  We next construct $F_{0}$ by solving the equation \eqref{in1} with respect to $F_{0}$. Clearly, condition 1 is satisfied for $F_{0}$. Condition 2 holds by construction. Moreover, the integral $\int_{v_{0}-\epsilon}^{v_{0}+\epsilon}F_{0}(v)\, dv$ agrees with its prescribed value in view of \eqref{in2} and the linearity of the integrals. This finishes the construction of $F_{ml}$ which in turn allows us to extend $\psi$ and hence to obtain the gluing of the characteristic data. 

We remark that as along as gluing is possible then it can be achieved in a highly non-unique way.  

\bigskip

\underline{\textbf{Case II}}:\medskip\\ Using the above simplified notation, we can assume without loss of generality that \begin{equation}
G_{1}(v)=a(v)\cdot G_{0}(v) \text{ and } \pv a\neq 0,\  G_{0}(v)\neq 0
\label{gmalista}
\end{equation} in a (sufficiently) small interval $I$ of $v$ (and hence $G_{1}$ is not linearly dependent on $G_{0}$). We choose $v_{0}$ to be in $I$ and we take $\epsilon>0$ small enough such that $[v_{0}-\epsilon,v_{0}+\epsilon]\subset I$. By relabeling $m,l$, as before, we can rewrite condition 3 as follows 
\[ F_{0}(v)\cdot G_{0}(v)+F_{1}(v)\cdot G_{1}(v)+\sum_{i\geq 2}F_{i}(v)\cdot G_{i}(v)=0 \]
and hence 
\begin{equation}
 F_{0}(v)=-\frac{G_{1}(v)}{G_{0}(v)}F_{1}(v)-\sum_{i\geq 2}\frac{G_{i}(v)}{G_{0}(v)}F_{i}(v)=-a(v)\cdot F_{1}(v)-\sum_{i\geq 2}\frac{G_{i}(v)}{G_{0}(v)}F_{i}(v). 
\label{eq:pro3}
\end{equation}
We can then prescribe $F_{i}(v)$, $i\geq 2$, such that condition 1 is satisfied and such that the integrals $\int_{v_{0}-\epsilon}^{v_{0}+\epsilon}F_{i}(v)\,dv$ agree with the prescribed values (condition 3). 

We can then prescribe $F_{1}$ such that condition 1 is satisfied, the integral $\int_{v_{0}-\epsilon}^{v_{0}-\epsilon}F_{1}(v)\, dv$ agrees with its prescribed value \textbf{and} the integral  $\int_{v_{0}-\epsilon}^{v_{0}-\epsilon}a(v)\cdot F_{1}(v)\, dv$ is such that the integral $\int_{v_{0}-\epsilon}^{v_{0}+\epsilon}F_{0}(v)\, dv$, computed via  \eqref{eq:pro3}, agrees with its prescribed value. Note that in view of \eqref{gmalista}, the integrals $\int_{v_{0}-\epsilon}^{v_{0}-\epsilon}a(v)\cdot F_{1}(v)\, dv$ and $\int_{v_{0}-\epsilon}^{v_{0}+\epsilon}F_{0}(v)\, dv$ are independent. In view of the fact that condition 1 is satisfied for all $F_{i},\, i\geq 1$, we have that it is automatically satisfied for $F_{0}$, again via \eqref{eq:pro3}. The remaining conditions hold by construction.  This finishes the construction of $F_{ml}$'s for the case where $dimKer\big(\o_{v}^{\s}\big)=1$ for all $v\in [0,1]$.

\bigskip

\noindent\underline{\textbf{The general case}}

\medskip

We now return to the general case. We first derive the following lemmata.

\begin{lemma}
Let $I$ be a closed interval of $\mathbb{R}$. Given $n\in\mathbb{N}$ linearly independent functions $f_{1},f_{2},...,f_{n}\in C^{\infty}\big(I\big)$ and $\lambda_{1},\lambda_{2},\cdots, \lambda_{n}\in \mathbb{R}$ there is a function $\a\in C^{\infty}\big(I\big)$ such that  
 \[\int_{I}\a(v)\cdot f_{i}(v)\, dv=\lambda_{i}, \text{ for all } i=1,2,...,n.\]
\label{lemma1stoproof}
\end{lemma}
\begin{proof}

Let $V_{n}=\Big\langle f_{1},f_{2},\cdots, f_{n}\Big\rangle\subset L^{2}(I)$
denote the $n$-dimensional span of the functions $f_{1},f_{2},\cdots, f_{n}$. Using the Gram--Schmidt process we produce an orthonormal basis $\left\{e_{1},e_{2},\cdots, e_{n}\right\}$ of $V_{n}$. We extend this basis to obtain an orthonormal basis $\left\{e_{1},\cdots, e_{n},e_{n+1},\cdots\right\}$ of $L^{2}(I)$. Clearly  for all $i=1,2,...,n$ we have that $e_{i}\in C^{\infty}\big(I\big)$ and  
\[f_{i}=\sum_{k=1}^{i}\left\langle e_{k},f_{i}\right\rangle\cdot e_{k} \]
with $\left\langle e_{i},f_{i}\right\rangle>0$, where $\left\langle\, \cdot\, ,\, \cdot\, \right\rangle$ denotes the inner product of $L^{2}(I)$.
We want to construct a function $\a\in C^{\infty}\big(I\big)$ such that $\left\langle a, f_{i}\right\rangle=\lambda_{i}$ for $i=1,2,...,n$. The system
\[\lambda_{i}= \sum_{k=1}^{i}\left\langle e_{k},f_{i}\right\rangle\cdot x_{i} \]
has a unique solution with respect to $x_{1},x_{2},...,x_{n}$. We then define the function $\a$ such that 
$\left\langle \a, e_{i}\right\rangle=x_{i}$ for $i=1,2,...,n$ and $\left\langle \a, e_{i}\right\rangle=0$ for $i\geq n+1$, which clearly satisfies the required relations.

\end{proof}

\begin{lemma}
Let $I$ be a compact interval of $\mathbb{R}$. Let $G_{1},G_{2},...,G_{n}\in C^{\infty}\big(I\times \mathbb{S}^{2}\big)$ be $n$ functions such that for each $v\in I$ the functions $G_{1}(v,\cdot), G_{2}(v,\cdot),...,G_{n}(v,\cdot)\in C^{\infty}\big(\mathbb{S}^{2}\big)$ are linearly independent and let  \[\Pi(v)=\Big\langle G_{1}(v,\cdot), G_{2}(v,\cdot),...,G_{n}(v,\cdot)\Big\rangle\subset L^{2}\big(\mathbb{S}^{2}\big), \] denote the ($v$-dependent) $n$-dimensional subspace of $L^{2}\big(\mathbb{S}^{2}\big)$ spanned by them.  Then, given $\rho\in C^{\infty}\big(\mathbb{S}^{2}\big)$ there is a function $F_{\rho}\in C^{\infty}\big(I\times\mathbb{S}^{2}\big)$  which vanishes to infinite order at $\partial I\times \mathbb{S}^{2}$ and is such that 
\begin{equation}
\int_{I}F_{\rho}(v,\cdot)\, dv= \rho(\cdot)
\label{eq:rel1}
\end{equation}
and 
\begin{equation}
\int_{\mathbb{S}^{2}}F_{\rho}(v,\cdot)\cdot G_{i}(v,\cdot)\, d\mu_{_{\mathbb{S}^{2}}}=0  \text{ for all }   i=1,2,...,n,
\label{eq:rel2}
\end{equation}
if and only if 
\[\rho\in \Bigg(\bigcap_{v\in I}\Pi(v)\Bigg)^{\perp}\subset L^{2}\big(\mathbb{S}^{2}\big). \]
\label{deuterolemma}
\end{lemma}
\begin{proof}
We denote 
\begin{equation}
V=\bigcap_{v\in I}\Pi(v)\subset L^{2}\big(\mathbb{S}^{2}\big).
\label{sxesi1}
\end{equation}
 Clearly, $V$ is finite dimensional. 

 If given a function $\rho$ the function $F_{\rho}$ exists then $F_{\rho}\in \Big(\Pi(v)\Big)^{\perp}$ for all $v$ and hence $F_{\rho}\in V^{\perp}$. Therefore,  an immediate application of Fubini's theorem yields that $\rho\in V^{\perp}$.

Let us assume now that $\rho \in V^{\perp}$. We will show that a function $F_{\rho}\in C^{\infty}\big(I\times \mathbb{S}^{2}\big)$  satisfying the above properties exists.  

Since $V$ is finite dimensional, we have the decomposition
\[L^{2}\big(\mathbb{S}^{2} \big)=V\oplus V^{\perp}.  \]
If we define the spaces
\[\Big.\Pi(v)\Big|_{V^{\perp}} :=\Big.\text{proj}\Big|_{V^{\perp}}\big(\Pi(v)\big).   \]
Since $V\subset \Pi(v)$ for all $v\in I$ the space $\Big.\Pi(v)\Big|_{V^{\perp}} $ varies smoothly in $v$. Indeed, since $\Pi(v)$ varies smoothly in $v\in I$ we can write $\Pi(v)=V\oplus T(v)$ where $T(v)$ varies smoothly in $v\in I$ and $\text{dim}\big(T(v)\big)=n-\text{dim}(V)$, where $\text{dim}\big(\Pi(v)\big)=n$ for all $v\in I$. Then $\Big.\Pi(v)\Big|_{V^{\perp}}=\Big.T(v)\Big|_{V^{\perp}}$. The projection
\[\Big.\text{proj}\Big|_{V^{\perp}}: T(v)\rightarrow V^{\perp} \]
has full rank since otherwise we would have $T(v)\cap V\neq \left\{0\right\}$, contradiction. Therefore, since $T(v)$ varies smoothly in $v$ the space $\Big.T(v)\Big|_{V^{\perp}}$ varies smoothly in $v$ and hence so does $\Big.\Pi(v)\Big|_{V^{\perp}}$. It also follows that 
\begin{equation}
\text{dim}\left(\Big.\Pi(v)\Big|_{V^{\perp}}\right)=n-\text{dim}(V).
\label{dimerelationw}
\end{equation}  
Furthermore, we  obtain
\begin{equation}
\bigcap_{v\in I}\left( \Big.\Pi(v)\Big|_{V^{\perp}} \right)=\left\{0\right\}.
\label{sxesi2}
\end{equation}
Indeed, if the line $\left\langle l\right\rangle$ lies in the above intersection then for every $v\in I$ there is $y_{v}\in \Pi(v)$ such that $y_{v}=x_{v}+l,$ where $x_{v}\in V$ (and $l\in V^{\perp}$). However, by \eqref{sxesi1}, $x_{v}\in \Pi(v)$ and hence by linearity $l\in \Pi(v)$. Since this holds of all $v$, it immediately contradicts \eqref{sxesi1}.

We next show that there is a finite dimensional subspace $W\subset V^{\perp}$ such that if 
\[\widetilde{\Pi}(v):= \Big.\text{proj}\Big|_{W} \Big( \Big.\Pi(v)\Big|_{V^{\perp}}\Big)\subset W,\]
then 
\begin{equation}
\bigcap_{v\in[0,1]}\widetilde{\Pi}(v)=\left\{0\right\}.
\label{xrisimipi}
\end{equation}
Indeed, in view of \eqref{dimerelationw},\eqref{sxesi2}, there are $v_{1},v_{2}\in I$ such that $K_{1}= \Big.\Pi(v_{1})\Big|_{V^{\perp}}\cap \Big.\Pi(v_{2})\Big|_{V^{\perp}}$ is at most $(n-1)$-dimensional (where  the dimension of $\Pi(v)$ is $n$). In view of \eqref{sxesi2}, there is $v_{3}\in I$ such that $K_{2}=K_{1}\cap \Big.\Pi(v_{3})\Big|_{V^{\perp}}$ is at most $(n-2)$-dimensional. Continuing inductively we deduce that there  are $v_{1},v_{2},...,v_{n}$ such that 
\begin{equation}
\bigcap_{v_{i}}\Big.\Pi(v_{i})\Big|_{V^{\perp}}=\left\{0\right\}.
\label{finiteintersectionmiden}
\end{equation}
 We then define $W$ by
\begin{equation}
W=\Big\langle \Big.\Pi(v_{1})\Big|_{V^{\perp}},\Big.\Pi(v_{2})\Big|_{V^{\perp}},...,\Big.\Pi(v_{n})\Big|_{V^{\perp}} \Big\rangle,
\label{eq:definitionofw}
\end{equation}
which is clearly finite dimensional and satisfies \eqref{xrisimipi}. Indeed, by the construction of $W$,  \eqref{xrisimipi} is satisfied even if we restrict the intersection for the values of $v$ in the set $ \left\{v_{1},...,v_{n}\right\}$.

We next consider the projection
\[\Big.\text{proj}\Big|_{W}(v): \Big.\Pi(v)\Big|_{V^{\perp}}\rightarrow W. \]
By virtue of \eqref{eq:definitionofw} we have that $\Big.\text{proj}\Big|_{W}(v)$ has full rank for $v=v_{1},...,v_{n}$. Since $\Big.\Pi(v)\Big|_{V^{\perp}}$ varies smoothly in $v$ we have that $\Big.\text{proj}\Big|_{W}(v)$ has full rank in the union of sufficiently small intervals $J_{i}$ containing $v_{i}$.

Let also $\widetilde{O}(v)$ denote the orthogonal complement of $\widetilde{\Pi}(v)$ in $W$. It is important to note that the space $\widetilde{O}(v)$ varies smoothly for $v\in \cup J_{i}$ since $\widetilde{\Pi}(v)$ varies smoothly for $v\in \cup J_{i}$. In view of \eqref{finiteintersectionmiden} we have that 
\begin{equation}
\Big\langle \widetilde{O}(v_{1}),...,\widetilde{O}(v_{n}) \Big\rangle= W.
\label{finitespanofo}
\end{equation}
Indeed, if there is a proper subspace $X\subset W$ such that $\widetilde{O}(v)\subset X$ for all $v\in \left\{v_{1},...v_{n}\right\}$ then $X^{\perp}\subset \big(\widetilde{O}(v)\big)^{\perp}\subset W$ for all $v\in \left\{v_{1},...v_{n}\right\}$. This is however contradiction since, by definition, $\big(\widetilde{O}(v)\big)^{\perp}=\widetilde{\Pi}(v)$ and by \eqref{finiteintersectionmiden} these spaces cannot have non-trivial common intersection. Therefore, there are $u_{1},...,u_{\text{dim}(W)}\in\left\{v_{1},...,v_{n}\right\}\subset \cup J_{i}$ and $x_{i}\in \widetilde{O}(u_{i})$ with $i=1,...,\text{dim}(W)$ such that the set $\left\{x_{i}, i=1,...,\text{dim}(W)\right\}$ is a basis of $W$. We can assume that $u_{i}$ are pairwise distinct since otherwise we can consider small perturbations $u_{i}^{\text{perp}}$ of them in the union $\cup J_{i}$. Then, since $\widetilde{O}(v)$ varies smoothly in $\cup J_{i}$, the perturbed vectors $x_{i}^{\text{perp}}=\widetilde{O}\big(u_{i}^{perp}\big)$ still form a basis of $W$. Moreover, we can assume that $u_{i}^{\text{perp}}$ lie in the interior of $I$ for all $i=1,...,\text{dim}(W)$. We can thus define the following closed intervals $I_{i},i=1,...,\text{dim}(W)$ as follows:
\begin{equation}
I_{i}\subset J_{i}\cap \text{int}I \text{ such that } u_{i}^{\text{perp}}\in I_{i},
\label{orismosdiastimatos}
\end{equation}
where $\text{int}I$ denotes the interior of $I$.  

Let now $W^{\perp}$ denote the orthogonal complement of $W$ in $V^{\perp}$ and hence 
\[L^{2}\big(\mathbb{S}^{2}\big)=V\oplus W \oplus W^{\perp}.  \] Since  $W=\widetilde{\Pi}(v)\oplus \widetilde{O}(v)$ we have

\begin{equation}
L^{2}\big(\mathbb{S}^{2}\big)= V\oplus \widetilde{\Pi}(v)\oplus \widetilde{O}(v)\oplus W^{\perp}. 
\label{sxesi3}
\end{equation}
Then, for any $F\in C^{\infty}\big(I\times \mathbb{S}^{2}\big)$ we have
\begin{equation}
F(v,\cdot)=\Big.F\Big|_{V}+\Big.F\Big|_{\widetilde{\Pi}(v)}+\Big.F\Big|_{\widetilde{O}(v)}+\Big.F\Big|_{W^{\perp}}
\label{sxesiF}
\end{equation}
where we denote $\Big.F\Big|_{Z}=\Big.\text{proj}\Big|_{Z}\big(F\big)$.

Given a function $\rho\in V^{\perp}\cap C^{\infty}\big(\mathbb{S}^{2}\big)$ we want to construct a function $F_{\rho}\in C^{\infty}\big(I\times \mathbb{S}^{2}\big)$ such that 
\begin{equation}
\int_{I}F_{\rho}(v,\cdot)dv=\rho\in V^{\perp} 
\label{edw1}
\end{equation}and
\begin{equation}
F_{\rho}(v,\cdot)\in \Big(\Pi(v)\Big)^{\perp}\subset V^{\perp}\ \text{ for all }\ v\in I.
\label{edw2}
\end{equation}
The first condition we impose on $F_{\rho}$ is:
\begin{equation}
\Big.F_{\rho}(v,\cdot)\Big|_{V}=0 \text{ for all } v\in I.
\label{conditionF1}
\end{equation}
Then condition \eqref{edw1} is equivalent to  
\begin{equation}
\int_{I}\Big.F_{\rho}(v,\cdot)\Big|_{W^{\perp}}dv=\Big.\rho\Big|_{W^{\perp}} 
\label{zitoumeno1}
\end{equation}
and 
\begin{equation}
\int_{I}\Big.F_{\rho}(v,\cdot)\Big|_{\widetilde{\Pi}(v)}dv+\int_{I}\Big.F_{\rho}(v,\cdot)\Big|_{\widetilde{O}(v)}dv=\Big.\rho\Big|_{W}.
\label{zitoumeno2}
\end{equation}
Moreover, condition \eqref{edw2}, using \eqref{conditionF1}, is equivalent to the following: 
\begin{equation}
\Big\langle \Big.F_{\rho}(v,\cdot)\Big|_{\widetilde{\Pi}(v)}, \big. w\big|_{\widetilde{\Pi}(v)} \Big\rangle =- \Big\langle \Big. F_{\rho}(v,\cdot)\Big|_{W^{\perp}},  \big. w\big|_{W^{\perp}}   \Big\rangle\text{ for any }w\in \Pi(v). 
\label{conditionF2}
\end{equation}
Indeed, for any $w\in \Pi(v)$ we have
\begin{equation*}
\begin{split}
0=&\Big\langle F_{\rho}(v,\cdot), w\Big\rangle=\Big\langle  \Big.F_{\rho}\Big|_{\widetilde{\Pi}(v)}+\Big.F_{\rho}\Big|_{\widetilde{O}(v)}+\Big.F_{\rho}\Big|_{W^{\perp}},  \big.w\big|_{V^{\perp}} \Big\rangle \\
& =\Big\langle \Big.F_{\rho}\Big|_{\widetilde{\Pi}(v)}+\Big.F_{\rho}\Big|_{\widetilde{O}(v)}+\Big.F_{\rho}\Big|_{W^{\perp}}, \big. w\big|_{\widetilde{\Pi}(v)} + \big. w\big|_{W^{\perp}} \Big\rangle\\
&=\Big\langle \Big.F_{\rho}\Big|_{\widetilde{\Pi}(v)},\big. w\big|_{\widetilde{\Pi}(v)} \Big\rangle+\Big\langle \Big.F_{\rho}\Big|_{W^{\perp}},  \big. w\big|_{W^{\perp}}  \Big\rangle,
\end{split}
\end{equation*}
since $\Big.F_{\rho}(v,\cdot)\Big|_{V}=\big.w\big|_{\widetilde{O}(v)}=0$. Note thus that  condition \eqref{conditionF2} (and hence \eqref{edw2}) is independent of  $\Big. F_{\rho}(v,\cdot)\Big|_{\widetilde{O}(v)}$. 

\bigskip

Let $\mathcal{B}_{1}=\left\{ e_{1},...,e_{\text{dim}(V)} \right\}$, $\mathcal{B}_{2}=\left\{ e_{\text{dim}(V)+1},...,e_{\text{dim}(V)+\text{dim}(W)} \right\}$, $\mathcal{B}_{3}=\left\{ e_{\text{dim}(V)+\text{dim}(W)+1},... \right\}$ be orthonormal bases of the spaces $V,W, W^{\perp}$, respectively. Clearly, $\mathcal{B}_{1}\cup\mathcal{B}_{2}\cup\mathcal{B}_{3}$ is an orthonormal basis of $L^{2}\big(\mathbb{S}^{2}\big)$. Since $\Pi(v)$ are spanned by smooth functions, it is easy to see that we can take $e_{i}$ to be smooth functions on $\mathbb{S}^{2}$.  Any function $F\in C^{\infty}\big(I\times \mathbb{S}^{2}\big)$ can be written as
\[F(v,\cdot)=\sum_{i\geq 1} F_{i}(v)\cdot e_{i}. \]
We impose $\big(F_{\rho}\big)_{i}(v)$ such that  for  $1\leq i\leq \text{dim}(V)$ as follows:
\begin{itemize}
	\item  $\big(F_{\rho}\big)_{i}(v)=0$ for all $v\in I$,  in accordance with \eqref{conditionF1}. 
\end{itemize}
 Moreover, we  prescribe  $\big(F_{\rho}\big)_{i}(v)$ for $i\geq \text{dim}(V)+\text{dim}(W)+1$ such that: 
\begin{itemize}
	\item $\left\{\big(F_{\rho}\big)_{i}(v) \right\}_{i}\in \ell^{2}$ for all $v\in I,$
\item 	$\big(F_{\rho}\big)_{i}(v)=0  $ for all $v\in I/ \big(I_{1}\cup ....\cup I_{\text{dim}(W)}\big),$ where $I_{k}$ as defined in \eqref{orismosdiastimatos},
\item $\big(F_{\rho}\big)_{i}(v)$ for $v\in I_{1}\cup ....\cup I_{\text{dim}(W)}$ are such that   $\int_{I}\big(F_{\rho}\big)_{i}(v)dv$ are imposed by condition \eqref{zitoumeno1}.
\end{itemize}
Clearly such functions exist and hence the functions $\Big.F_{\rho}\Big|_{W^{\perp}}(v)\in C^{\infty}\big(\mathbb{S}^{2}\big)$ are determined and, in particular are such that they vanish for $v\in I/\big(I_{1}\cup...\cup I_{\text{dim}(W)}\big)$.  

We next construct $\Big.F_{\rho}\Big|_{W}(v)=\Big.F_{\rho}\Big|_{\widetilde{\Pi}(v)}(v)+\Big.F_{\rho}\Big|_{\widetilde{O}(v)}(v)$ or, equivalently, to construct $\big(F_{\rho}\big)_{i}(v)$ for $\text{dim}(V)+1\leq i\leq \text{dim}(V)+\text{dim}(W)+1$. We first impose 
\begin{itemize}
	\item $\Big.F_{\rho}\Big|_{W}(v)=0$ for all $v\in I/\big(I_{1}\cup...\cup I_{\text{dim}(W)}\big)$ and hence $\big(F_{\rho}\big)_{i}(v)=0$ for $\text{dim}(V)+1\leq i\leq \text{dim}(V)+\text{dim}(W)+1$ and $v\in I/\big(I_{1}\cup...\cup I_{\text{dim}(W)}\big)$. 
\end{itemize}
Clearly, condition \eqref{conditionF2} is then satisfied for all $v\in I/\big(I_{1}\cup...\cup I_{\text{dim}(W)}\big)$.
It remains to construct the functions  $\Big.F_{\rho}\Big|_{\widetilde{\Pi}(v)}(v),\, \Big.F_{\rho}\Big|_{\widetilde{O}(v)}(v)$ for $v\in \big(I_{1}\cup...\cup I_{\text{dim}(W)}\big)$.

We note that condition \eqref{conditionF2} uniquely determines a well-defined $\Big.F_{\rho}\Big|_{\widetilde{\Pi}(v)}(v)$ for all $v\in\big(I_{1}\cup...\cup I_{\text{dim}(W)}\big)$.  Indeed, if $\widetilde{\Pi}(v)=\left\{0\right\}$ for some $v\in \big(I_{1}\cup...\cup I_{\text{dim}(W)}\big)$ then we have $\widetilde{\Pi}(v)=\left\{0\right\}$ for all $v\in\big(I_{1}\cup...\cup I_{\text{dim}(W)}\big)$ and hence $\Pi(v)=V$ for all $v\in I$.  In this case we take $W=\left\{0\right\}$ and \eqref{conditionF2} holds trivially. If, on the other hand, the spaces $\widetilde{\Pi}(v)$ are non-trivial and vary smoothly for all $v\in \big(I_{1}\cup...\cup I_{\text{dim}(W)}\big)$ then condition \eqref{conditionF2} determines $\Big.F_{\rho}\Big|_{\widetilde{\Pi}(v)}(v)$ for $v\in\big(I_{1}\cup...\cup I_{\text{dim}(W)}\big)$. Since $\Big.F_{\rho}\Big|_{W^{\perp}}(v)$  a smooth function in $I$ which vanishes in $I/\big(I_{1}\cup...\cup I_{\text{dim}(W)}\big)$ we obtain that $\Big.F_{\rho}\Big|_{\widetilde{\Pi}(v)}(v)$, as defined above,  is also a smooth function in $I$ which vanishes in $I/\big(I_{1}\cup...\cup I_{\text{dim}(W)}\big)$. 

We finally construct the function  $\Big.F_{\rho}\Big|_{\widetilde{O}(v)}(v)$ for $v\in \big(I_{1}\cup...\cup I_{\text{dim}(W)}\big)$ such that \eqref{zitoumeno2} holds (recall that this function vanishes in the complement of this union in $I$).  We will indeed show that this is possible.

Recall that  $\widetilde{O}(v)$ varies smoothly in $v\in I_{1}\cup...\cup I_{\text{dim}(W)}$. Moreover,  by the definition of the intervals $I_{1},...,I_{\text{dim}(W)}$, there are vectors $x_{i}\in \widetilde{O}(v_{i})$ with $v_{i}\in I_{i}, i=1,...,\text{dim}(W)$ such that $\left\{x_{i},i=1,...,\text{dim}(W)\right\}$ is a basis of $W$. We can construct a smooth curve 
\[\gamma:I\rightarrow W  \]
such that 
\begin{equation}
\gamma(v_{i})=x_{i} 
\label{gam1}
\end{equation}
and 
\begin{equation}
\gamma(v)=0 \text{ for all } v\in I/\big(I_{1}\cup...\cup I_{\text{dim}(W)}\big).
\label{gam2}
\end{equation}
The smoothness of the curve $\gamma$ can be guaranteed by the smoothness of $\widetilde{O}(v)$ in $I_{1}\cup...\cup I_{\text{dim}(W)}$. If $\big(f_{1}(v),f_{2}(v),...,f_{\text{dim}(W)}(v) \big)$ are the coordinates of $\gamma(v)$ with respect to the basis $\mathcal{B}_{2}$ of $W$, then, in view of \eqref{gam1}, the smooth functions $f_{i}:I\rightarrow\mathbb{R},i=1,...,\text{dim}(W),$ are linearly independent.  We will find an appropriate smooth function $\a:I\rightarrow\mathbb{R}$ such that 
\begin{equation}
\Big. F_{\rho}\Big|_{\widetilde{O}(v)}=\a(v)\cdot \gamma(v)\in \widetilde{O}(v). 
\label{ansatz1}
\end{equation}
Let us assume that $\Big(\big(F_{\rho}\big)_{1}(v),...,\big(F_{\rho}\big)_{\text{dim}(W)}(v) \Big)$ are the coordinates of $\Big.F_{\rho}(v,\cdot)\Big|_{\widetilde{O}(v)}$ with respect to the basis $\mathcal{B}_{2}$ of $W$. Condition \eqref{zitoumeno2} is satisfied if we choose these functions such that 
\[ \int_{I}\big(F_{\rho}\big)_{i}(v)\, dv=\lambda_{i},\]for all $i=1,2,...,\text{dim}(W)$, 
where $\lambda_{i}$ is completely determined by \eqref{zitoumeno2} and our previous constructions. Equivalently, in view of our ansatz \eqref{ansatz1}, it suffices to show the existence of a smooth function $\a:I\rightarrow\mathbb{R}$ such that\begin{equation}
\int_{I}\a(v)\cdot f_{i}(v)\, dv=\lambda_{i},
\label{eq:nai}
\end{equation}
for all $i=1,2,...,\text{dim}(W)$. Since the functions $f_{i}$ are linearly independent, the existence of  $\a$ follows from Lemma \ref{lemma1stoproof}. This completes the construction of the function $\Big.F_{\rho}(v,\cdot)\Big|_{\widetilde{O}(v)}$. Note that in view of \eqref{gam2} the function $\Big.F_{\rho}(v,\cdot)\Big|_{\widetilde{O}(v)}$ vanishes in $I/\big(I_{1}\cup...\cup I_{\text{dim}(W)}\big)$. 

This completes the construction of the function $F_{\rho}$ with all the required properties.
\end{proof}
Lemma \ref{lemma1stoproof} holds not only for intervals but also for union of intervals. Moreover, a modification of the above proof yields the following result
\begin{lemma}
Let $I=\bigcup_{k=1}^{n}I_{k}$, where $I_{k}$, with $k=1,2,...,n$, are compact intervals of $\mathbb{R}$. For each $k\in\left\{1,2,...,n\right\}$ consider  $\Pi_{k}(v), v\in I_{k}$, to be a smoothly varying $n_{k}$-dimensional subspace of $L^{2}\big(\mathbb{S}^{2} \big)$ spanned by $k$ smooth functions on $\mathbb{S}^{2}$. Define the subspaces $V_{k}\subset L^{2}\big(\mathbb{S}^{2} \big)$ as follows
\[ V_{k}=\bigcap_{v\in I_{k}}\Pi_{k}(v).  \]
Given a function $\rho\in C^{\infty}\big(\mathbb{S}^{2}\big)$ there is a function $F_{\rho}\in C^{\infty}\big(I\times\mathbb{S}^{2}\big)$  which vanishes to infinite order at $\partial I\times \mathbb{S}^{2}$ and is such that 
\begin{equation}
\int_{I}F_{\rho}(v,\cdot)\, dv= \rho(\cdot)
\label{eq:rel1}
\end{equation}
and 
\begin{equation}
F_{\rho}(v,\cdot)\in \Big( \Pi_{k}(v) \Big)^{\perp} \text{ for all }v\in I_{k} \text{ and  }k=1,2...,n
\label{eq:arbitragelemma}
\end{equation}
if and only if 
\[\rho\in \Big(V_{1}\cap V_{2}\cap...\cap V_{n}\Big)^{\perp}\subset L^{2}\big(\mathbb{S}^{2}\big). \]
\label{argitragelemma}
\end{lemma}
\begin{proof}
The above lemma is proved using the same arguments as in the proof of Lemma \ref{deuterolemma} if we replace the space $V$ by $V_{1}\cap V_{2}\cap...\cap V_{n}$.
\end{proof}

\bigskip

We now have all the tools needed for the proof of Theorem \ref{theorem}.

 Let $x\in\mathcal{Z}$ (where $\mathcal{Z}$ is given by Lemma \ref{lemmagiaelliptickernel}) and $V_{x}$ is a closed interval containing $x$ and such that  $dimKer(\o_{v}^{*})=n\geq 1$. If we apply Lemma \ref{deuterolemma} for $I=V_{x}$ and $\Pi(v)=\left\langle G_{1}(v),...,G_{n}(v)\right\rangle$, where  the functions $G_{n}$ given by \eqref{g0}, then the only obstruction to satisfying conditions 1--3  (and thus to gluing) is the existence of the following intersection
\begin{equation}
\mathcal{W}(V_{x})=\bigcap_{v\in V_{x}}\Pi(v). 
\label{subkernel}
\end{equation}
In other words, if there is $x\in \mathcal{Z}$ such that $\mathcal{W}(V_{x})=\left\{0\right\}$ then we can perform gluing. Suppose now that for all $x\in\mathcal{Z}$ we have $\mathcal{W}(V_{x})\neq \left\{0\right\}$. Let us fix an $x\in\mathcal{Z}$ and let $\Theta\in \mathcal{W}(V_{x})$ such that $\Theta\neq 0$.  By  \eqref{defF} and \eqref{integrability}, and since $\partial_{v}\Theta=0$ for $v\in V_{x}$,   we have 
\[0=\int_{S_{v}}(\mathcal{F}^{\mathcal{D}}\psi\cdot\phi)\cdot\Theta\, d\mu_{_{\mathbb{S}^{2}}}=\int_{S_{v}}\Big(\pv\pu(\phi\cdot\psi)\Big)\cdot\Theta\, d\mu_{_{\mathbb{S}^{2}}}=\pv\left(\int_{S_{v}}\Big(\pu(\phi\cdot\psi)\Big)\cdot\Theta\, d\mu_{_{\mathbb{S}^{2}}}\right), \]
where we have also used that the measure of integration does not depend on $v$. Hence, the quantity
\begin{equation}
\int_{S_{v}}\Big(\pu(\phi\cdot\psi)\Big)\cdot\Theta\ d\mu_{_{\mathbb{S}^{2}}}
\label{eq:conserved}
\end{equation}
is conserved in $V_{x}$, i.e.~independent of $v$ for all $v\in V_{x}$. Therefore, we have conservation laws in each interval $V_{x}$ and the kernel of the conservation laws is precisely the  space
 $\mathcal{W}(V_{x})$.  

We now consider two cases.

\medskip

\underline{ \textbf{Case I:}} 

\medskip

Suppose that 
\begin{equation}
\bigcap_{x\in\mathcal{Z}}\mathcal{W}(V_{x})=\left\{0\right\}.
\label{suppose1}
\end{equation}
By virtue of Lemma \ref{lemmagiaelliptickernel}, we have that there is $n\in\mathbb{N}$ such that $\text{dim}\big(\mathcal{W}(V_{x})\big)\leq n$ for all $x\in\mathcal{Z}$. In view of \eqref{suppose1}, there are $x_{1},x_{2}\in\mathcal{Z}$ such that $\mathcal{W}(V_{x_{1}})\cap\mathcal{W}(V_{x_{2}})$ is at most $(n-1)$-dimensional.  In view again of \eqref{suppose1}, there is $x_{3}\in\mathcal{Z}$ such that $\mathcal{W}(V_{x_{1}})\cap\mathcal{W}(V_{x_{2}})\cap\mathcal{W}(V_{x_{3}})$ is at most $(n-2)$-dimensional.  Continuing inductively we obtain $x_{1},...,x_{n+1}\in\mathcal{Z}$ such that $\mathcal{W}(V_{x_{1}})\cap\mathcal{W}(V_{x_{2}})\cap...\cap\mathcal{W}(V_{x_{n+1}})=\left\{0\right\}$. Thus, applying Lemma \ref{argitragelemma} for $V_{x_{1}},...,V_{x_{n+1}}$ we can satisfy conditions 1--3 and thus perform gluing. The smoothness of the extension of $\psi$ on $\hh$ follows from the results of Lemma \ref{lemmagiaelliptickernel} and their analogue in higher Sobolev spaces.

\medskip

\underline{ \textbf{Case II:}} 

\medskip

Suppose that 
\begin{equation}
\bigcap_{x\in\mathcal{Z}}\mathcal{W}(V_{x})\neq\left\{0\right\}.
\label{suppose1}
\end{equation}
Let $\Theta \in \bigcap_{x\in\mathcal{Z}}\mathcal{W}(V_{x})$ with $\Theta\neq 0$. Then the integrals \eqref{eq:conserved} are conserved for $v\in V_{x}$ for all $x\in \mathcal{Z}$. For any  $x,y\in\mathcal{Z}$ with $V_{x}\cap V_{y}\neq \emptyset$ we have that the integrals \eqref{eq:conserved} are conserved for $v\in V_{x}\cup V_{y}$. We will show that the integrals \eqref{eq:conserved} are in fact conserved for all $v\in[0,1]$. Let us denote the integral \eqref{eq:conserved} over $S_{v}$ by $f(v)$. Then  $f$ is a smooth function on $[0,1]$ which is constant on the connected components of the open and dense subset $\mathcal{Z}$ of $[0,1]$. If $f$ is not constant then there is $v_{0}\in [0,1]$ for which $\big.\partial_{v}f\big|_{v_{0}}\neq 0$ and hence there is an interval containing $v_{0}$ where $f$ is strictly monotonic. This is however contradiction since, by assumption, for every interval there is a subinterval where $f$ is constant.

Therefore, the integrals \eqref{eq:conserved} are conserved for $v\in[0,1]$ for all $\Theta \in \bigcap_{x\in\mathcal{Z}}\mathcal{W}(V_{x})$. Clearly, these conservation laws are legitimate obstructions to gluing since we cannot do gluing unless the corresponding integrals of the initial data at $S_{0}$ and $S_{1}$ are equal. Suppose, therefore, that the initial data at $S_{0}$ and $S_{1}$ are such that the corresponding integrals \eqref{eq:conserved} are equal. Then we will show that we can perform gluing (and hence there are no additional obstructions to gluing apart from the above conservation laws). 

In view of \eqref{wenull1} and \eqref{malista} and the second condition for the gluing of transversal derivative we need to construct a function $C^{\infty}\big([0,1]\times\mathbb{S}^{2}\big)$ such that 
\begin{equation}
\int_{0}^{1}F(v,\cdot)\, dv= 2\pu(\phi\cdot\psi)_{\big|_{S_{1}}}-2\pu(\phi\cdot\psi)_{\big|_{S_{0}}}.
\label{apopano}
\end{equation}
In view of our assumption on the initial data at $S_{0},S_{1}$ we have that the right hand side of \eqref{apopano} is orthogonal to the space $\bigcap_{x\in\mathcal{Z}}\mathcal{W}(V_{x})$ with respect to the inner product of $L^{2}\big(\mathbb{S}^{2}\big)$. Using a similar argument as above, we can choose finitely many $x_{1},...,x_{m}\in \mathcal{Z}$ such that $V_{x_{i}}\cap V_{x_{j}}=\emptyset$ for $i\neq j$ and  
\[\mathcal{W}(V_{x_{1}})\cap...\cap\mathcal{W}(V_{x_{m}})=\bigcap_{x\in\mathcal{Z}}\mathcal{W}(V_{x}). \]
It now suffices to extend $\psi$ in the complement of $V_{x_{1}}\cup...\cup V_{x_{m}}$ in $[0,1]$ and apply Lemma \ref{argitragelemma} for the  intervals $V_{x_{1}}, ....,  V_{x_{m}}$ in order to complete the gluing construction.  This finishes the proof of Theorem \ref{theorem}.

\subsection{Change of foliation and conservation laws}
\label{sec:EffectOfGaugeOnConservationLaws}

So far we have consider characteristic gluing constructions and conservation laws on null hypersurfaces with respect to a \textbf{fixed} foliation. In particular, it is clear the the conserved charges depend, at least a priori, on the choice of foliation\footnote{After all, the conserved charges are appropriate integrals over the leaves of the foliation.}. In this subsection we address the issue of change of foliation and its effect on the conservation laws. 

Let $\s=\Big\langle S_{0}, L_{geod}, \Omega\Big\rangle$ and $\s'=\left\langle S_{0}', L_{geod}', \Omega'\right\rangle$ be two foliations of a regular null hypersurface $\hh$, as defined in Section \ref{sec:NullFoliationsandOpticalFunctions}. Let $\vh$ be the linear space defined by \eqref{linearspace}. Consider the kernels $\mathcal{W}^{\s},\mathcal{W}^{\s'}\subset \vh$ of the conservation laws with respect to the foliations $\s,\s'$, respectively, as defined in Section \ref{sec:ConservationLawsForTheWaveEquations}. Given also the operators $\o^{\s},\o^{\s'}$ associated to the foliations $\s,\s'$, respectively, and defined by \eqref{adjoint} we consider the linear spaces $\mathcal{U}^{\s},\mathcal{U}^{\s'}\subset\vh$ which are (appropriately rescaled) subspaces of  $Ker\big(\o^{\s}\big),Ker\big(\o^{\s'}\big)$ defined as in Theorem \ref{theorem}. Then, according to Theorem \ref{theorem} we have
\begin{equation}
\mathcal{W}^{\s}=\mathcal{U}^{\s}\ \text{ and }\  \mathcal{W}^{\s'}=\mathcal{U}^{\s'}.
\label{eq:change1}
\end{equation}
The following proposition derives the relation between the spaces $\mathcal{U}^{\s}$ and $\mathcal{U}^{\s'}$. 
\begin{proposition}
Let $\s=\Big\langle S_{0}, L_{geod}, \Omega\Big\rangle$ and $\s'=\left\langle S_{0}', L_{geod}', \Omega'\right\rangle$ be two foliations of a (regular) null hypersurface $\hh$ of a four-dimensional Lorentzian manifold $(\m,g)$, as defined in Section \ref{sec:NullFoliationsandOpticalFunctions}. Let $f\in\vh$ be such that   
\begin{equation}
L_{geod}'=f^{2}\cdot L_{geod}\ : \text{  on }\, \hh.
\label{geodesicf}
\end{equation}
Then, we have
\begin{equation}
\mathcal{U}^{\s'}=f^{2}\cdot\mathcal{U}^{\s}=\Big\{f^{2}\cdot\Theta\ :\ \Theta\in\mathcal{U}^{\s} \Big\}.
\label{changef12}
\end{equation}
\label{corm1}
\end{proposition}
\begin{proof}
According to the main result of \cite{aretakiselliptic}, the operators $\mathcal{O}^{\s}$ and $\o^{\s'}$ satisfy the following relation
\begin{equation}
\o^{\s}\left(\frac{1}{\phi}\cdot\Theta\right)=\left(\frac{\Omega}{\Omega'}\right)^{2}\cdot\frac{1}{f^{2}}\cdot \o^{\s'}\left(f^{2}\cdot\frac{1}{\phi}\cdot\Theta \right),
\label{fromotherpaper}
\end{equation}
on $\hh$, for all functions $\Theta\in\vh$. Let now $\Theta\in\mathcal{U}^{\s}$ and $S'_{v'}$ be a section of the foliation $\s'$. Then, $S'_{v'}$ can be sweeped by the sections $S_{v}$ of the foliation $\s$ as depicted below. 
\begin{figure}[H]
   \centering
		\includegraphics[scale=0.13]{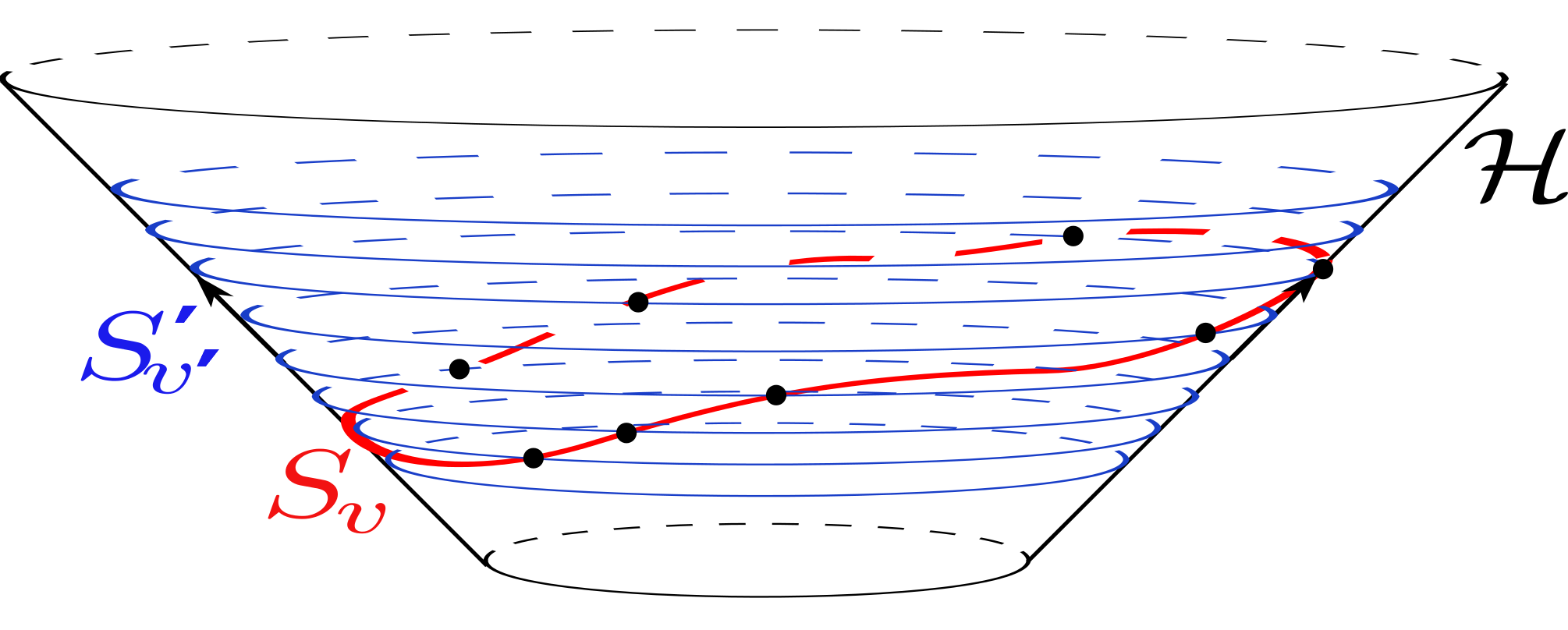}
	\label{fig:nullsectionxax}
\end{figure}
For each intersection point of $S'_{v'}$ with the sections of $\s$ we apply the relation \eqref{fromotherpaper} to deduce that 
\[\text{ If }\ \o_{v}^{\s}\left(\frac{1}{\phi}\cdot \Theta \right)=0 \ \text{for all }v\in\mathbb{R}\ \text{ then }\ \o^{\s'}_{v'}\left(\frac{1}{\phi}\cdot \big(f^{2}\cdot\Theta \big)\right)=0. \]
Hence, $ f^{2}\cdot\mathcal{U}^{\s}  \subseteq \mathcal{U}^{\s'}$. The inverse inclusion can be similarly shown by sweeping the section of the foliation $\s$ with the sections of $\s'$, yielding the required result. 
\end{proof}
We remark that the equation \eqref{changef12} holds only if we can sweep the sections of the foliation $\s$ with those of $\s'$ and vice versa. This shows that the existence of conservation laws (and hence the obstructions to gluing) are due to global properties of the sections of the null hypersurface (and not just pointwise or local; cf. Remark 1 in Section \ref{sec:Remarks}). Specifically, they are related with the properties of the kernels of the elliptic operators $\o^{\s}_{v}, v\in\mathbb{R}$. The calculation of \cite{aretakiselliptic} showed that these elliptic operators are covariant (in the sense of \eqref{fromotherpaper}) under change of foliation. This covariance property holds also for the kernels of these operators as long as the sweeping property holds. Hence, although capturing the elliptic structure associated to the sections of $\s$ is of fundamental importance in the present paper, it turns out this elliptic structure does not depend on the choice of the foliation of $\hh$.

In view of the results of Theorem \ref{theorem}, an immediate corollary of Proposition \ref{corm1} is the following
\begin{corollary}
A null hypersurface $\hh$ of a four-dimensional Lorentzian manifold $(\m,g)$ admits conservation laws with respect to a foliation $\s=\Big\langle S_{0}, L_{geod}, \Omega\Big\rangle$ in the sense of Definition \ref{definitionconservationlaw} if and only it admits conservation laws with respect to any other foliation $\s'=\Big\langle S_{0}', L_{geod}', \Omega'\Big\rangle$. Specifically, the kernels $\mathcal{W}^{S},\mathcal{W}^{\s'}$ of the conservation laws satisfy
\begin{equation}
\mathcal{W}^{\s'}=f^{2}\cdot\mathcal{W}^{\s}=\Big\{f^{2}\cdot\Theta\ :\ \Theta\in\mathcal{W}^{\s} \Big\},
\label{corequchangeoffoliation}
\end{equation}
where $f$ is given by \eqref{geodesicf}. In other words, the integrals 
\begin{equation}
char\big(S_{v}\big)[\psi]=\int_{S_{v}}Y^{\s}\big(\phi\cdot\psi\big)\cdot\Theta \, d\mu_{_{\mathbb{S}^{2}}}
\label{corollaryequation1}
\end{equation}
with $\Theta\in\vh$ are conserved (i.e.~independent of $v$), for all solutions $\psi$ to the wave equation, if and only if the integrals 
\begin{equation}
char\big(S'_{v'}\big)[\psi;\Theta]=\int_{S'_{v'}}f^{2}\cdot Y^{\s'}\big(\phi\cdot\psi\big)\cdot\Theta \, d\mu_{_{\mathbb{S}^{2}}} 
\label{corollaryequation2}
\end{equation}
are conserved, i.e.~independent of $v$, for all solutions $\psi$ to the wave equation. The vector fields  $Y^{\s},\, f^{2}\cdot Y^{\s'}$ are null and normal to the sections of $\s,\s'$, respectively, conjugate to $\hh$ and normalized such that 
\[g\big(Y^{\s},L_{geod}\big)=g\big(f^{2}\cdot Y^{\s'},L_{geod}\big)=-1.  \]
Moreover, if $\hh$ admits conservation laws then we in fact have
\begin{equation}
char\big(S_{v}\big)[\psi;\Theta]=char\big(S'_{v'}\big)[\psi;\Theta]
\label{corollaryequation3}
\end{equation}
and hence the value of the conserved charges is independent of the choice of foliation. 
\label{corollarychangeofoliation}
\end{corollary}

\subsection{Perturbation analysis}
\label{sec:PerturbativeAnalysis}

We next investigate the stability/genericity properties of the conservation laws on a null hypersurface $\hh$ of a Lorentzian manifold $(\m,g)$ under perturbations of the ambient metric $g$ (i.e.~under perturbations of the geometry of $\hh$). We have the following

\begin{proposition}
Let $\hh$ be a null hypersurface as in Theorem \ref{theoremmainintro} which admits conservation laws for the wave equation on a Lorentzian manifold $(\m,g)$. There are  arbitrarily small perturbations of the ambient metric $g$ for which $\hh$ is a null hypersurface without, however, admitting conservation laws (and hence gluing is possible on these perturbed null hypersurfaces).  In fact, the set of all ambient metrics for which $\hh$ admits consevation laws is of positive codimension. 

Furthermore, if the underlying Lorentzian manifold $(\m,g)$ satisfies the Einstein--vacuum equations then we can perturb the metric $g$ in the class of spacetimes satisfying the Einstein--vacuum equations such that the above conclusions still hold.

\label{prop2pert}
\end{proposition} 
\begin{proof}
In view of the results of Section \ref{sec:EffectOfGaugeOnConservationLaws} it suffices to consider a geodesic foliation $\s=\left\langle S_{0}, \,L_{geod}, \, \Omega=1\right\rangle$ of $\hh$. By assumption we have that $dim\mathcal{W}^{\s}=\dim\mathcal{U}^{\s}\geq 1$, which, in particular, implies that $Ker\big(\o_{v}^{\s}\big)\neq \left\{0\right\}$ for all $v\in\mathbb{R}$. Recall that (for the unperturbed metric):
\begin{equation}
\o_{v}^{\s}\psi=\lapp\psi+2\zeta^{\sharp}\cdot\nabb\psi+\Big[2\divv\zeta+\pv( tr\underline{\chi})+\frac{1}{2}( tr\underline{\chi})\cdot( tr\chi) \Big]\cdot\psi 
\label{toxreiazomai}
\end{equation}
We will show that for any fixed $v$ there is an $\epsilon_{1}>0$ such that for all $\epsilon\in(0,\epsilon_{1})$ the operator
\[\o_{v}^{\s,\epsilon}\psi=\o_{v}^{\s}\psi-\epsilon\cdot \psi\]
has \textit{trivial} kernel. Indeed, if we consider  the operator
\[\o_{temp}\psi=\o_{v}^{\s}\psi-\frac{1}{\epsilon_{0}}\cdot \psi \]
then we can show as before that if we take $\epsilon_{0}$ sufficiently small (depending on $v$) then the operator $\o_{temp}$ is invertible and hence has a discrete spectrum consisting only of eigenvalues. Therefore, the operator $\o_{v}^{\s}$ has a discrete set of eigenvalues in the spectrum (whose limit point is infinity). By assumption, one of these eigenvalues is zero. By the discreteness of the spectrum, if we take $\epsilon_{1}$ sufficiently small, then the operator $\o_{v}^{\s,\epsilon}$ has trivial kernel. 

Therefore, if we consider a perturbed metric $g_{\epsilon}$ for which $\gi_{\epsilon}\!\left.\right|_{S_{0}}=\gi\!\left.\right|_{S_{0}}$, $\zeta_{\epsilon}\!\left.\right|_{\hh}=\zeta\!\left.\right|_{\hh},$  $\tr\chi_{\epsilon}\!\left.\right|_{\hh}=\tr\chi\!\left.\right|_{\hh},$ $tr\underline{\chi}_{\epsilon}\!\left.\right|_{S_{0}}=tr\underline{\chi}\!\left.\right|_{S_{0}}$, $\partial_{v}tr\underline{\chi}_{\epsilon}\!\left.\right|_{S_{0}}=\partial_{v}tr\underline{\chi}\!\left.\right|_{S_{0}}+\epsilon$ for some $v$. Then, 
\[\widetilde{\o}_{v}^{\s}=\o^{\s,\epsilon}_{v},\]
where $\widetilde{\o}_{v}^{\s}$ denotes the elliptic operator $\o_{v}^{\s}$ with respect to the metric $g_{\epsilon}$. By the above discussion and Proposition \ref{pert1prop} we deduce that $\hh$ does not admit conservation laws as a null hypersurface embedded in the Lorentzian manifold  $\big(\m,g_{\epsilon}\big)$.

We next consider perturbations in the class of spacetimes satisfying the Einstein--vacuum equations. The freely prescribable initial data of the geometry of $\hh$ in the context of the characteristic problem of the Einstein--vacuum equations is the following:
\begin{enumerate}
	\item The conformal geometry of $\hh$,
\item The metric of $S_{0}$,
	\item The expansions $tr\chi$ and $tr\underline{\chi}$ at $S_{0}$ and
		\item The torsion $\zeta$ at $S_{0}$.
\end{enumerate}
In our perturbation argument, we will fix the conformal geometry of $\hh$, the metric of $S_{0}$ and the torsion $\zeta$ at $S_{0}$ and we will perturb the initial expansions $tr\chi,tr\underline{\chi}$ at $S_{0}$.

The propagation equation for the transversal null second fundamental form yields (see \cite{DC09})
\[\partial_{v}tr\underline{\chi}=\divv\zeta+|\zeta|^{2}+\rho-(\chi,\underline{\chi}). \]
Furthermore, the Gauss equation gives us
\[\rho-(\chi,\underline{\chi})=-K-tr\chi tr\underline{\chi}. \]
Therefore, by eliminating the left hand side of the above equation we obtain a linear propagation equation for $tr\underline{\chi}$ along $\hh$:
\begin{equation}
\partial_{v}tr\underline{\chi}=\divv\zeta+|\zeta|^{2}-K-tr\chi tr\underline{\chi}.
\label{propaforzeta}
\end{equation}
Hence the coefficient of the zeroth order term in \eqref{toxreiazomai} becomes
\begin{equation}
w=3\divv\zeta+|\zeta|^{2}-K-\frac{1}{2}( tr\underline{\chi})\cdot( tr\chi). 
\label{eq:w}
\end{equation}
By assumption the first three terms on the right hand side are fixed under our perturbations. We can then freely perturb $tr\chi$ and $tr\underline{\chi}$ such that 
\[w_{\epsilon}=w+\epsilon\]
at $S_{0}$. The result then follows from the above discussion that Proposition \ref{pert1prop}.

\end{proof}

\section{Black hole spacetimes}
\label{sec:KillingHorizons}

Let $(\m,g)$ be a stationary spacetime with a black hole region (for the relevant definitions see \cite{haw}). For such a  spacetime the event horizon $\hh$ is  a Killing horizon, that is there is a Killing vector field $\xi$ normal to $\hh$. In this case $\xi$ satisfies
\begin{equation}
\nabla_{\xi}\xi=\kappa\cdot \xi\, : \text{ on  }\hh,
\label{surface}
\end{equation}
where $\kappa$ is constant along the null generators of $\hh$. In consistency with the zeroth law of black hole mechanics, we will assume that $\kappa$ is globally constant on $\hh$ in which case $\kappa$ is the so-called \textit{surface gravity} of $\hh$. If $\kappa=0$, then $\hh$ is called an \textit{extremal horizon}. We will next investigate the existence of conservation laws on the event horizon $\hh$. 

The following properties of Killing horizons were shown in \cite{aretakiselliptic}:
\begin{lemma}
Let $\hh$ be a Killing horizon and $\mathcal{S}=\left\langle S_{0},L_{geod}, \Omega=1\right\rangle$ be a geodesic foliation of $\hh$, as defined in Section \ref{sec:NullFoliationsandOpticalFunctions}. Assume that $\xi$ is a Killing vector field normal to $\hh$ and such that \eqref{surface} is satisfied. Then, the following  relations hold on $\hh$: 
\begin{enumerate}
	\item $\chi=0$, 
\item $\li_{L}\gi=0$,
	\item $\di\kappa=g(\xi, \underline{L})\cdot \beta$, where the curvature component $\beta$ is given by \eqref{curvcompdeflist},
	\item $\li_{L}\zeta=\nabb_{L}\zeta=-\beta$,
	\item If, in addition, we take $\left.L_{geod}\right|_{S_{0}}=\left.\xi\right|_{S_{0}}$ and $\kappa$ is constant on $\hh$, then
	
	$\li_{L}\underline{\chi}=\nabb_{L}\underline{\chi}=\displaystyle\frac{\kappa}{f}\cdot\underline{\chi}$, where $f$ is such that $\xi=f\cdot {L}_{geod}$ on $\hh$.

\end{enumerate}
\label{lemma}
\end{lemma}

\noindent Recall that since $\Omega=1$ we have $L_{geod}=L=\pv$. If we trace the last identity of the above lemma we obtain
\[L tr\underline{\chi}=\frac{\kappa}{f}\cdot tr\underline{\chi}.  \]
Since $Lf=\kappa$ we obtain
\begin{equation}
tr\underline{\chi}=\left.tr\underline{\chi}\right|_{S_{0}}\cdot f
\label{trchibarlambda}
\end{equation} and so 
\begin{equation}
\pv tr\underline{\chi}=\left.tr\underline{\chi}\right|_{S_{0}}\cdot\kappa.
\label{tracechibar}
\end{equation}
Therefore, we obtain
\begin{equation}
	\begin{split}
\o^{\s}\psi=\lapp\psi+2\zeta^{\sharp}\cdot\nabb\psi+\left[2\divv\,\zeta^{\sharp}+\left.tr\underline{\chi}\right|_{S_{0}}\cdot\kappa\right]\cdot\psi,
\end{split}
\label{killingoperator}
\end{equation}
with respect to the foliation \begin{equation}
\mathcal{S}=\left\langle S_{0}, \left.L_{geod}\right|_{S_{0}}=\left.\xi\right|_{S_{0}},\,  \Omega=1\right\rangle.
\label{foliationkilling}
\end{equation}  

Since $\li_{L}\gi=0$ all the sections of $\hh$ are isometric. Moreover, since $\kappa$ is constant on $\hh$, and hence $\di k=0$, by Lemma \ref{lemma}, we obtain that  $\beta=0$ on $\hh$ and hence $\zeta$ is conserved on $\hi$, i.e.~$\li_{L}\zeta=0$. 	From \eqref{tracechibar}, $\pv tr\underline{\chi}$ does not depend on $v$. By virtue of the equation $L\phi=\frac{1}{2}tr\chi\cdot\phi$, the conformal factor $\phi$ also does not depend on $v$. Therefore, the operators $\mathcal{O}_{v}^{\s}$ do \textbf{not} depend on $v$ (modulo identifying the sections $S_{v}$ with $S_{0}$ via the diffeomorphisms $\Phi_{v}$). 

If we now consider a general foliation ${\mathcal{S}'}=\left\langle S_{0},\left.L_{geod}\right|_{S_{0}}=\left.\xi\right|_{S_{0}},\Omega\right\rangle$ where $\Omega$ is a smooth function on $\hh$, then by virtue of \eqref{fromotherpaper} and using the fact that the geodesic vector field $L_{geod}$ is the same for both foliations $\s,\s'$, we obtain
\[\frac{1}{\Omega^{2}}\cdot \o^{{\s'}}\Psi=\o^{{\s}}\Psi\ :\ \text{ for all }\, \Psi\in\vh. \]
and hence the operators $\frac{1}{\Omega^{2}}\cdot \o^{{\s}'}_{v}$ do not depend on $v$ (again modulo identifying the sections $S_{v}'$ of $\s'$ with $S_{0}$ via the diffeomorphisms $\Phi_{v}$). 

\begin{remark}
Given a foliation $\mathcal{S}=\left\langle S_{0}, \left.L_{geod}\right|_{S_{0}}=\left.\xi\right|_{S_{0}},\,  \Omega=1\right\rangle$, as defined in Section \ref{sec:NullFoliationsandOpticalFunctions},  of a Killing horizon of a four-dimensional Lorentzian manifold $(\m,g)$, we can rewrite the operator $\mathcal{O}^{\s}_{v}$ given by \eqref{killingoperator} as follows
\begin{equation}
\o_{v}^{\s}\psi=\underbrace{\lapp\psi+2\zeta^{\sharp}\cdot\nabb\psi+2\divv\,\zeta^{\sharp}\cdot \psi}_{{\mathcal{K}_{v}^{\s}\psi}}+\underbrace{\left.tr\underline{\chi}\right|_{S_{0}}\cdot\kappa\cdot\psi}_{\mathcal{T}_{v}^{\s}\psi},
\label{killingoperator1}
\end{equation}
The section $S_{0}$ can be freely chosen for the foliation $\mathcal{S}$. In view of the above discussion the operator $\mathcal{K}_{v}^{\s}$ does not depend on the choice of the section $S_{0}$ (again, modulo identifying all sections of $\hh$ via the flow of the null generators). However, clearly the operator $\mathcal{T}_{v}^{\s}$ depends on $S_{0}$. Specifically, if we consider another foliation $\mathcal{S}'=\left\langle S_{0}', \left.L_{geod}\right|_{S_{0}'}=\left.\xi\right|_{S_{0}'},\,  \Omega=1\right\rangle$ then in view of \eqref{fromotherpaper}, and recalling that $L\phi=0$, we have that
\begin{equation}
\o_{v}^{\s}(\Psi)=\frac{1}{f^{2}}\cdot\o_{v}^{\s'}\big(f^{2}\cdot\Psi\big)\ :\ \text{ for all }\, \Psi\in\vh,
\label{changeofsection}
\end{equation}
where $f$ is such that $L_{geod}'=f^{2}\cdot L_{geod}$, where $L_{geod}',L_{geod}$ denote the geodesic vector fields of $\s',\s$, respectively.
\label{remarkgiazeta}
\end{remark}
In view of the above proposition, the main Theorem \ref{theorem} and the fact that $L\phi=0$, i.e.~$\phi\in\vh$, the existence of conservation laws along Killing horizons is equivalent to the non-triviality of the kernel $Ker\big(\o_{v}^{\s}\big)$ for some fixed $v$. Specifically, we have shown that 
\begin{equation}
dim\, \mathcal{W}^{\s}= dim\, \mathcal{U}^{\s}= \phi\cdot Ker\big(\o_{v}^{\s}\big).
\label{killingapotelesma}
\end{equation}

Summarizing we have shown the following
\begin{proposition}
Let $\hh$ be a Killing horizon with constant surface gravity $\kappa$ of a four-dimensional Lorentzian manifold $(\m,g)$. Let also $S=\left\langle S_{0},\left.L_{geod}\right|_{S_{0}}=\left.\xi\right|_{S_{0}},\Omega\right\rangle$ be a foliation of $\hh$, as defined in Section \ref{sec:NullFoliationsandOpticalFunctions}. Then the operators $\frac{1}{\Omega^{2}}\cdot\o^{\mathcal{S}}_{v}$ do not depend on $v$ modulo identifying $S_{v}$ with $S_{0}$ via the diffeomorphism $\Phi_{v}$, i.e.
\[ \big(\Phi_{v}\big)^{*}\left(\frac{1}{\Omega^{2}}\cdot\mathcal{O}_{v}^{\s}\right)=\frac{1}{\Omega^{2}}\cdot\mathcal{O}_{0}^{\s}   .\] 
Under the same identification we have 
\begin{equation}
Ker\big(\o_{v}^{\s}\big)=Ker\big(\o_{0}^{\s}\big),
\label{kernelrelation}
\end{equation}
for all $v\in\mathbb{R}$.  Moreover, the kernel of the conservation laws along $\hh$ satisfies
\[dim\, \mathcal{W}^{\s}= dim\, \mathcal{U}^{\s}= \phi\cdot Ker\big(\o_{v}^{\s}\big)=\left\{\phi\cdot f\, :\, f\in Ker\big(\o_{v}^{\s}\big)\right\}.\]
\label{lastprop}
\end{proposition}

\subsection{Conservation laws on extremal black holes}
\label{sec:ConservationLawsOnExtremalHorizons}

By definition,  extremal black holes satisfy $\kappa=0$. In this case the operator $\mathcal{O}_{v}^{\s}$ takes the form 
\begin{equation}
\o_{v}^{\s}\psi=\lapp\psi+\divv\big(2\psi\cdot\zeta\big)
\label{operatorextremal}
\end{equation}
with respect to the foliation $\s$ given by \eqref{foliationkilling}. 
Following the argument of Lucietti and Reall \cite{hj2012} one obtains that $dim\mathcal{U}^{\s}=dimKer(\o_{v}^{\s})=1$ for all $v$. Indeed, since $\int_{S_{v}}\o^{\s}_{v}\psi=0$ for all $\psi$, the unique positive principal eigenfunction $\Psi$ of $\o_{v}^{\s}$ (see \cite{andersson, evans}) must lie in the kernel of $\o_{v}^{\s}$. This shows the following
\begin{proposition}
Let $\hh$ be an extremal horizon (i.e.~$\kappa=0$) of a four-dimensional Lorentzian manifold $(\m,g)$. If $\s$ is the foliation of $\hh$ given by \eqref{foliationkilling} then 
\[dim\,\mathcal{W}^{\s}=dim\,\mathcal{U}^{\s}=dimKer(\o_{v}^{\s})=1.\]
Therefore, $\hh$ admits a unique conservation law with respect to the foliation $\s$ (and thus with respect to any foliation). This law coincides with the 
conservation law found in \cite{aretakis4, hj2012, murata2012}.
\label{propositionextremal}
\end{proposition}

The above conservation law coupled with dispersive estimates away from the event horizon forces higher order derivatives of generic solutions to the wave equation to blow up asymptotically along the event horizon (see \cite{aretakis1,aretakis3}).  We remark that the previous result is in stark contrast with the subextremal case for which Dafermos and Rodnianski \cite{lecturesMD,enadio,tria} have derived quantitative decay estimates for all higher order derivatives in the exterior region up to and including  the event horizon.

\subsection{Gluing constructions for sub-extremal black holes}
\label{sec:GluingConstructionsForNonExtremalHorizons}

We next consider sub-extremal Killing horizons and specifically such that $\kappa>0$. We also assume that there exists a section $S_{0}$ such that $\left.tr\underline{\chi}\right|_{S_{0}}<0$ on $S_{0}$.\footnote{Note that under these assumptions we can use the calculations in \cite{aretakiselliptic} and the method of the present subsection to deduce that there must exist a section $S$ such that $\left.tr\underline{\chi}\right|=c,$ where $c<0$ is constant on $S$. } Let $\s$ be the foliation given by \eqref{foliationkilling} such that its  `initial' section is the above one. In view of \eqref{killingoperator} the operator $\o_{v}^{\s}$ can then be written as
\[\o_{v}^{\s}\psi=\lapp\psi+\divv\big(2\psi\cdot\zeta\big)+\left[\left.tr\underline{\chi}\right|_{S_{0}}\cdot\kappa\right]\cdot\psi.\]
Let $\Psi>0$ be the unique (up to rescaling) positive principal eigenfunction of $\o_{v}^{\s}$ and let $\lambda$ be its principal (maximum) eigenvalue. Then, we immediately obtain
\[ \int_{S_{v}}\left(\left.tr\underline{\chi}\right|_{S_{0}}\cdot\kappa\right)\cdot \Psi \, d\mu_{_{\gi}}=\lambda \cdot\int_{S_{v}}\Psi\, d\mu_{_{\gi}}. \]
The left hand side is manifestly negative which forces the maximum eigenvalue $\lambda$ to be strictly negative and hence $ Ker\big(\o_{v}^{\s}\big)=\left\{0\right\}$. We have thus shown the following
\begin{proposition}
Let $\hh$ be a Killing horizon with positive surface gravity $\kappa>0$. We additionally assume that there is a  spherical section $S_{0}$ of $\hh$ with negative transversal null expansion, i.e.~$\big.tr\underline{\chi}\big|_{S_{0}}<0$. Then $\hh$ does not admit any conservation laws (i.e.~$dim\,\mathcal{W}^{\s}=0$) and hence gluing in the sense of Definition \ref{firstordergluingdefinition} of characteristic data is always possible on $\hh$. 
\label{gluingnonextremal}
\end{proposition}
We note that the non-existence of  conservation laws on the event horizon $\hh$ of a subextremal Kerr black hole $(|a|<M)$ does not follow from the previously mentioned decay results of Dafermos and Rodnianski. Indeed, according to Definition \ref{definitionconservationlaw}, if $\hh$ admitted conservation laws then those would involve the $Y^{\s}$ derivative of the scalar field $\psi$. However, Dafermos and Rodnianski have shown decay results for the translation-invariant derivative $Y^{dec}$, where 
\[Y^{\s}=e^{\kappa v}\cdot Y^{dec}.   \]

We also remark that the assumption on the negativity of the transversal null expansion is necessary. For a counterexample see the second example in Section \ref{sec:AnExampleOfAMetricForWhichC1Neq0AndC20}.

\section{The Newman--Penrose constants}
\label{sec:TheNullInfinityMathcalI}

In this section we will show that the Newman--Penrose constants (see \cite{np1}) can be recovered by Theorem \ref{theorem}. We first start with a discussion about the geometry of the null infinity $\mathcal{I}$ of general asymptotically flat four-dimensional Lorentzian manifolds. 

\subsection{The null infinity $\mathcal{I}$}
\label{sec:TheAsympoticGauge}

Let $(\m,g)$ be a general asymptotically flat four-dimensional Lorentzian manifold. 
Consider an outgoing null hypersurface $\hh_{0}=\left\{u=0\right\}$ of $\m$ with future complete null generators. Let $\s=\left\langle S_{0},L_{geod},\Omega=1\right\rangle$ be a foliation of $\hh_{0}$ such that  $S_{0}$ is a surface embedded in a Cauchy hypersurface $\Sigma$ terminating at the spacelike infinity $i^{0}$ (see figure below). Let $\tau$ be the affine parameter of $L_{geod}$ on $\hh_{0}$, i.e.~$\left.\tau\right|_{S_{0}}=0$ and $L_{geod}(\tau)=1$, and let $S_{\tau}$ denote the corresponding sections on $\hh_{0}$. Let $\underline{\hh}_{\tau}$ denote the ingoing null hypersurface of $\m$ generated by incoming null geodesics normal to  $S_{\tau}$. We consider the  collection $\left\{\mathcal{D}_{\tau},\, \tau\geq 0\right\}$ of the double null foliations generated by
$\mathcal{D}_{\tau}=\left\langle S_{\tau}, \left.L_{geod}\right|_{S_{\tau}}, \left.\Omega\right|_{\hh_{0}}=1, \left.\Omega\right|_{\underline{\hh}_{\tau}}=1 \right\rangle$ 
where $\left.L_{geod}\right|_{S_{\tau}}$ is normalized such that 
\begin{equation}
tr\chi+tr\underline{\chi}=0\ : \ \text{ on }\, S_{\tau}.
\label{normalizationnp}
\end{equation}
\begin{figure}[H]
   \centering
		\includegraphics[scale=0.115]{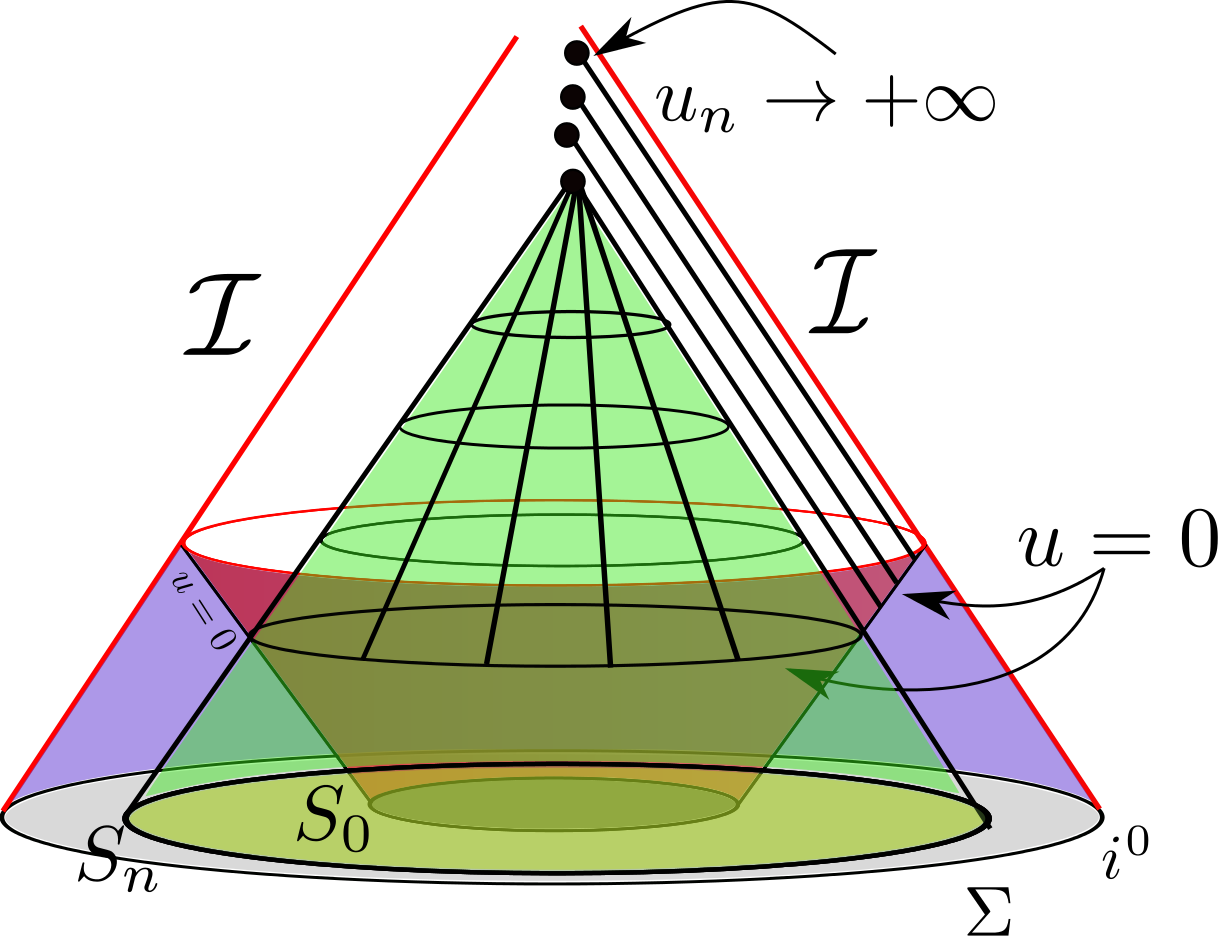}
		\label{fig:nullinf2}
\end{figure}
Consider now the induced foliation $\underline{\s}_{\tau}=\left\langle S_{\tau}, \left.\underline{L}_{geod}\right|_{S_{\tau}},  \left.\Omega\right|_{\underline{\hh}_{\tau}}=1 \right\rangle$ on $\underline{\hh}_{\tau}$. Let $\gi$ be the induced metric, $\nabb$ the induced covariant derivative and $\lapp$ the induced Laplacian. Let also $A$ be the area and $r=\sqrt{A/4\pi}$ the radius function. 

Suppose that  as $\tau\rightarrow +\infty$ we have $r(S_{\tau})\rightarrow +\infty$.
Then, the null infinity $\mathcal{I}$ is defined to be the limit of the hypersurfaces $\underline{\hh}_{\tau}$ as $\tau\rightarrow +\infty$.  In view of its limiting character, quantites associated to $\mathcal{I}$ can only be understood in a limiting, appropriately rescaled, sense. The limit of the foliations $\underline{\mathcal{S}}_{\tau}$, as $\tau\rightarrow+\infty$, gives rise to a foliation of $\mathcal{I}$. We have the following asymptotic behavior (see, for example, \cite{memorychistodoulou}):
\begin{equation}
\begin{split}
&\frac{1}{r^{2}}\gi\rightarrow \gi_{\mathbb{S}^{2}}, \ \ \ \frac{1}{r^{2}}\sqrt{\gi}\rightarrow \sin\theta,\ \ \  {r^{2}}\nabb\rightarrow \nabb_{\mathbb{S}^{2}}, \ \ \ r^{2}\lapp\rightarrow \lapp_{\mathbb{S}^{2}}, \ \ \ \frac{1}{r}\phi\rightarrow 1,\\
&\partial_{v}r\rightarrow \frac{1}{\sqrt{2}}, \ \ \ \partial_{u}r\rightarrow- \frac{1}{\sqrt{2}}, \ \ \ r\zeta\rightarrow Z,
\\&  rtr\chi\rightarrow \sqrt{2}, \ \ \ rtr\underline{\chi}\rightarrow -\sqrt{2}, \ \ \ r^{2}\partial_{v}tr\underline{\chi}\rightarrow 1, \ \ \ r^{2}\partial_{u}tr{\chi}\rightarrow 1,\\
\end{split}
\label{asym}
\end{equation}
\begin{figure}[H]
   \centering
		\includegraphics[scale=0.125]{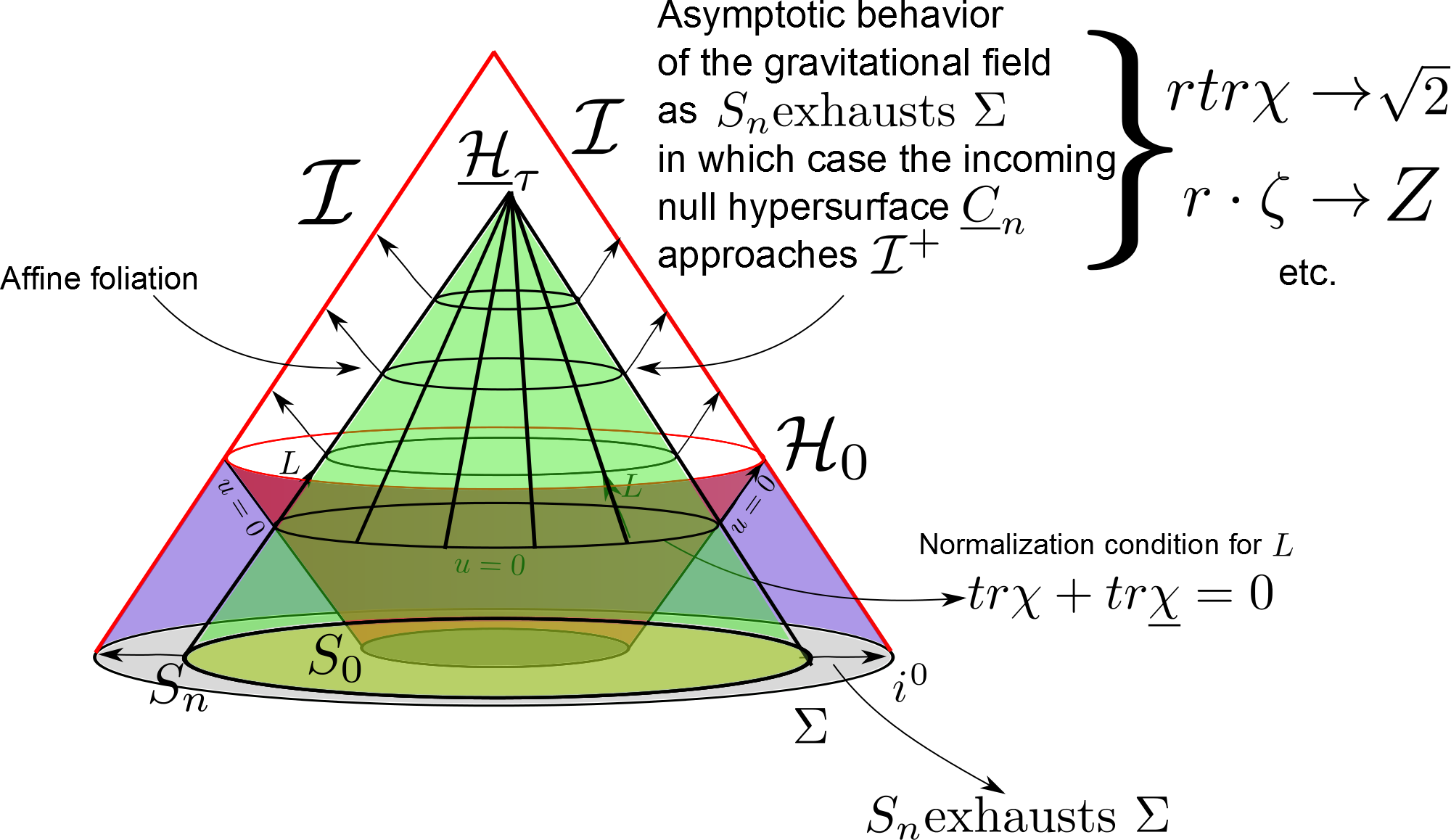}
		\label{fig:nullinfpeel}
\end{figure}
\noindent where $Z$ is a 1-form on $\mathbb{S}^{2}$. The above limits should be understood in terms of  the pullback of the induced tensor fields to the standard sphere via tha diffeomorphism $\Phi_{u,v}$ (see Section \ref{sec:TheDoubleNullFoliation}). 



\subsection{The Newman--Penrose constants}
\label{sec:TheNewmanPenroseConstants}

Let us briefly recall the (first-order) Newman--Penrose constant. If we  use the coordinate system $(u,r,\theta^{1},\theta^{2})$ to make the dependence on the powers of $r$ more explicit then we obtain
\[\psi(u,r,\theta)=\frac{\a_{1}(u,\theta^{1},\theta^{2})}{r}+\frac{\a_{2}(u,\theta^{1},\theta^{2})}{r^{2}}+O\left(\frac{1}{r^{3}}\right)\]
close to null infinity $\mathcal{I}$. Note that $\partial_{u}$ is tangential to $\mathcal{I}$ and $\partial_{v}$ is transversal. 
Here, \[ \a_{1}(u,\theta^{1},\theta^{2})=\lim_{r\rightarrow+\infty}(r\psi)(u,r,\theta^{1},\theta^{2}) \]
is the radiation field of $\psi$ on $\mathcal{I}$. 
 Furthermore,
\[ \partial_{v}(r\psi)= -\sqrt{2}\cdot \a_{2}(u,\theta^{1},\theta^{2})\cdot \frac{1}{r^{2}}+O\left(\frac{1}{r^{3}}\right).\]
Therefore, as $r\rightarrow +\infty$ we have $\partial_{v}(r\psi)\rightarrow 0$ and $r\cdot \partial_{v}(r\psi)\rightarrow 0$. The higher order non-trivial quantity
\[r^{2}\cdot \partial_{v}(r\psi)=-\sqrt{2}\cdot \a_{2}(u,\theta^{1},\theta^{2})\]  
gives rise to the Newman--Penrose costant
\begin{equation}
\lim_{r\rightarrow+\infty}\int_{S_{u}}r^{2}\partial_{v}(r\psi)\, d\mu_{_{\mathbb{S}^{2}}}=-\sqrt{2}\cdot \int_{S_{u}}\a_{2}(u,\theta) \, d\mu_{_{\mathbb{S}^{2}}}
\label{npconstant}
\end{equation}
which turns out to be  conserved (i.e. independent of $u$) along $\mathcal{I}$.

We will show that the (first-order) Newman--Penrose constant is a limiting case of Theorem \ref{theorem}. The  wave equation with respect to the double null foliation $\mathcal{D}_{\tau}$ can be written on $\underline{\hh}_{\tau}$ as follows
\begin{equation}
-2\partial_{u}\pv(\phi\cdot\psi) +\phi\cdot \mathcal{Q}^{\underline{\s}_{\tau}}\psi=0,
\label{wenullinf}
\end{equation}
where 
\begin{equation}
\mathcal{Q}^{\underline{\s}_{\tau}}\psi=  \lapp\psi-2\zeta^{\sharp}\cdot \nabb\psi+ \left[\pu( tr{{\chi}})+\frac{1}{2}(tr\underline{\chi})\cdot(tr\chi)\right]\cdot\psi.
\label{ovoperatorinf}
\end{equation}
Similarly, the adjoint operator is given by
 	\begin{equation*}
	\begin{split}
\o^{\underline{\s}_{\tau}}\psi=\lapp\psi+2\zeta^{\sharp}\cdot \nabb\psi+\left[2\divv\,\zeta^{\sharp}+\pu( tr{\chi})+\frac{1}{2}(tr\underline{\chi})\cdot(tr\chi)\right]\cdot\psi
\end{split}
\end{equation*}
on $\underline{\hh}_{\tau}$. Let $\underline{\s}$ denote the limiting foliation  $\lim_{\tau\rightarrow+\infty}\underline{\s}_{\tau}$. Clearly, in this limit $\underline{\hh}_{\tau}\rightarrow\mathcal{I}$. We define the following rescaled limiting operator
\begin{equation}
\o^{\underline{\s}}\psi=\lim_{r\rightarrow+\infty}\o^{\underline{\s}_{\tau}}\big(r^{2}\cdot\psi\big) \ : \ \text{ on }\, \mathcal{I}.
\label{limitingoperator}
\end{equation}
In view of the  asymptotics \eqref{asym}, we obtain
	\begin{equation}
	\begin{split}
\o^{\underline{S}}\psi=\lapp_{{\mathbb{S}^{2}}}\psi. 
\end{split}
\label{operatorfori}
\end{equation}
\textbf{The key observation here is that, in view of the asymptotics \eqref{asym}, $\zeta$, and more importantly, $  r^{2}\cdot\bigg(\partial_{u}tr{\chi}+\frac{1}{2}tr\chi tr\underline{\chi}\bigg)$  decay like $1/r$} and hence all the terms but the Laplacian  vanish at the limit $r\rightarrow +\infty$.

In view of the presence of limiting quantities we need to be particularly careful in order to apply Theorem \ref{theorem}. Indeed, the integrability condition \eqref{1inte} does not hold on $\mathcal{I}$ since the measure of integration has infinite area element. Using however that on $\mathcal{I}$ we have $d\mu_{_{\gi}}=r^{2}\cdot d\mu_{_{\mathbb{S}^{2}}}$ and hence \eqref{1inte}, using \eqref{wenull1}, takes the following limiting form
\begin{equation}
\int_{S_{u}}\Big(\mathcal{Q}_{u}^{\underline{\s}}(\psi)\Big) \cdot\Big( \big(E^{\underline{S}}\big)_{u}\Big)\, d\mu_{_{\gi}}=\int_{S_{u}}\psi\cdot \o^{\underline{\s}}_{u}\Big(\big(E^{\s}\big)_{u}\Big)\, d\mu_{_{\gi}}=\int_{S_{u}}(r\psi)\cdot \o^{\underline{\s}}_{u}\Big(r\cdot \big(E^{\underline{\s}}\big)_{u}\Big) \,d\mu_{_{\mathbb{S}^{2}}} =0
\label{limitingform}
\end{equation}
where, for the last equation, we take $\big(E^{\underline{\s}}\big)_{u}$ such that \begin{equation}
r\cdot\big(E^{\underline{\s}}\big)_{u}\in Ker\big(\o_{u}^{\underline{\s}}\big).
\label{eq:notethat}
\end{equation} Note that $r\psi$, being the radiation field of $\psi$ on $\mathcal{I}$, is finite. 

Recall that the kernel $\mathcal{W}^{\underline{\s}}_{\mathcal{I}}$ of the conservation laws on $\mathcal{I}$ consists of all functions $\Theta^{\underline{\s}}$ which are ``constant'' along the null generators of $\mathcal{I}$ such that 
\begin{equation}
\Theta^{\underline{\s}}=\phi\cdot \big(E^{\underline{\s}}\big)_{u}=r\cdot \big(E^{\underline{\s}}\big)_{u}.
\label{rescaledu}
\end{equation}
In view of \eqref{limitingoperator}, \eqref{operatorfori} and \eqref{eq:notethat} we can (only) take $\big(E^{\underline{\s}}\big)_{u}=r$ and hence 
\begin{equation}
\Theta^{\underline{\s}}=r^{2}.
\label{eq:thetanp}
\end{equation}
Modulo a trivial rescaling we have, $Y^{\underline{\s}}=\pu$, where $Y^{\underline{\s}}$ is the vector field defined in Section \ref{sec:ConservationLawsForTheWaveEquations},  and in view of the asymptotic behavior of the function $r$ we can in fact take
\begin{equation}
Y^{\underline{\s}}=\partial_{r} \ : \ \text{ on }\, \mathcal{I}{+}.
\label{ynp}
\end{equation}
Hence, the conserved charge \eqref{eq:integrals} on $\mathcal{I}$ is equal to
\begin{equation}
\lim_{r\rightarrow+\infty}\int_{S_{u}}r^{2}\cdot \partial_{r}(r\psi)\,d\mu_{_{\mathbb{S}^{2}}},
\label{conpcha}
\end{equation}
which coincides with the Newman--Penrose constant \eqref{npconstant}.
\begin{remark}
Clearly, the function $r$ is not constant on the incoming null hypersurfaces $\underline{\hh}_{\tau}$. However, it can be considered constant on $\mathcal{I}$ in a limiting sense . As far as the conservation laws are concerned we have the following: Since 
  $\phi-r\in O\left(\frac{1}{r^{a}}\right),\ a>0$, we obtain
\begin{equation*}
\begin{split}  r^{2}\partial_{u}\partial_{v}(\phi\cdot\psi)&=\partial_{u}\big(r^{2}\partial_{v}(\phi\cdot\psi)\big)-2r\partial_{u}r\partial_{v}(\phi\cdot\psi)\\
&=\left[\partial_{u}\big(r^{2}\partial_{v}(r\psi)\big)-2r\partial_{u}r\partial_{v}(r\psi)+\partial_{u}\big(r^{2}\partial_{v}(r^{-a}\psi)\big)-2r\partial_{u}r\partial_{v}(r^{-a}\psi)\right]\\
&\rightarrow \partial_{
u}\big(r^{2}\partial_{v}(r\psi)\big).
\end{split}
\end{equation*}
\label{remarknkconse}
\end{remark}
\noindent The restriction $\o_{u}^{\underline{\s}}$ on $S_{u}$ of the operator $\o^{\underline{\s}}$ on $\mathcal{I}$ is given by
\begin{equation}
\o_{u}^{\underline{\s}}=\lapp_{\mathbb{S}^{2}}
\label{eq:oni}
\end{equation}
and hence is independent of $u$. This ``$u$-invariance''  of the operator $\mathcal{O}^{\underline{\s}}_{u}$ is due to the BMS symmetry group of $\mathcal{I}$ (see \cite{wald}). Recall from Section \ref{sec:KillingHorizons} that a similar result holds for Killing horizon and in fact the kernel of $\o_{u}^{\underline{\s}}$ is isomorphic to the kernel of the associated operator on extremal horizons. This reveals yet another common property of $\mathcal{I}$ and extremal horizons. 

If $\mathcal{W}^{\underline{\s}}$ denotes the kernel of the (limiting) conservation laws on $\mathcal{I}$ and $\mathcal{U}^{\underline{\s}}$ denotes the appropriately rescaled kernel of the operator $\o^{\underline{\s}}$ (see \eqref{rescaledu}) then, in view of \eqref{operatorfori}, we have
\[dim\,\mathcal{W}^{\underline{\s}}=dim\,\mathcal{U}^{\underline{\s}}=1\]
and, therefore, there exists \textbf{only} one (non-trivial limiting) conservation law on $\mathcal{I}$, namely that given by the Newman--Penrose constants.  Moreover, the result of Section \ref{sec:EffectOfGaugeOnConservationLaws} applies for this conservation law. Specifically, if ${\underline{\s}}'$ is a foliation of $\mathcal{I}$ and $\widetilde{Y}^{\underline{\s}'}$ denotes the unique null conjugate to ${\underline{\s}}'$ vector field normalized  such that $\widetilde{Y}^{\underline{\s}'}r=1$ then the integrals
\[\lim_{r\rightarrow+\infty}\int_{S'_{u'}}r^{2}\cdot \widetilde{Y}^{\underline{\s}'}(r\cdot\psi) \, d\mu_{_{\mathbb{S}^{2}}} \]
are conserved, i.e.~independent of $u$. 

Summarizing  we have shown the following
\begin{proposition}
Let $\underline{\s}=\big(S_{u}\big)_{u\in\mathbb{R}}$ be a foliation of the null infinity $\mathcal{I}$ of an asymptotically flat spacetime $(\m,g)$, as defined in Section \ref{sec:TheAsympoticGauge}. Then, the appropriately rescaled operator $\mathcal{O}^{\underline{\s}}$ given by \eqref{operatorfori} is $u$-invariant, i.e.~the operators $\mathcal{O}^{\underline{\s}}_{u}$ do not depend on $u$. Moreover, $dim\,\mathcal{W}^{\underline{\s}}=dim\,\mathcal{U}^{\underline{\s}}=1$ and the unique associated conservation law on $\mathcal{I}$ gives rise to the (first-order) Newman--Penrose constant. 
\label{nullinfiprop}
\end{proposition}

 Note that all the higher order Newman--Penrose constants can be obtained by commuting the wave equation with $\pv^{k}$ (see also Section \ref{sec:Genericity}).

\section{Spherical symmetry}
\label{sec:Genericity}

In this section we investigate the existence of higher order conservation laws. Although our method applies for general spacetimes for the sake of simplicity we focus on spherically symmetric backgrounds.

Let $\hi$ be a spherically symmetric null hypersurface and  $\s=\big(S_{v}\big)_{v\in\mathbb{R}}$ be a spherically symmetric  foliation on $\hh$ in a spherically symmetric four-dimensional Lorentzian manifold $(\m,g)$. Then the wave equation restricted on $\hh$ can be written as
\begin{equation*}
\begin{split}
\Box_{g}\psi
=&-2\partial_{u}\partial_{v}(r\psi)+ \o^{\s}(r\psi)=0,  
\end{split}
\end{equation*}
where \[\o^{\s}\psi=\Omega^{2}\frac{1}{r^{2}}\lapp_{\mathbb{S}^{2}}\psi +2\Omega^{2}\cdot \frac{(\partial_{u}\partial_{v}r)}{r}\cdot \psi.\]
We next assume that  $\Omega=1$ on $\hh$. Since all the expressions are spherically symmetric, in view of Theorem \ref{theorem}, we have a (first order) conservation law if and only if 
\[\frac{2\partial_{u}\pv r}{r}=\frac{l(l+1)}{r^{2}}, \]
for some $l\in\mathbb{N}$.

\subsection{Higher order conservation laws}
\label{sec:HigherOrderConservationLaws}

For a given $n\in\mathbb{N}$, we want to find necessary and sufficient conditions under which we can glue general data 
\begin{equation}
\big.\psi\big|_{S_{0}}, \ \ \ \left.\partial_{u}^{k}\psi\right|_{S_{0}},\ \ \  1\leq k\leq n
\label{data1}
\end{equation}
on $S_{0}$ to general data  
\begin{equation}
\big.\psi\big|_{S_{1}}, \ \ \ \left.\partial_{u}^{k}\psi\right|_{S_{1}},\  \ 1\leq k\leq n, 
\label{data2}
\end{equation}
on $S_{1}$ as explained in Section \ref{sec:TheCharacteristicInitialValueProblem}.

As we shall see the only obstruction to such gluings is higher order conservation laws. By decomposing $\psi$ in angular frequencies we can assume that it is supported on the $l$ angular frequency. Then, the wave equation on $\hh$ reads
\begin{equation*}
\begin{split}
\Box_{g}\psi
=&-2\partial_{u}\partial_{v}(r\psi)+ \o^{\s}(r\psi)=0,  
\end{split}
\end{equation*}
where \[\o^{\s}\psi=\Omega^{2}\cdot\left[2\frac{\pu\pv r}{r}-\frac{l(l+1)}{r^{2}} \right]\cdot\psi.\]
Set 
\begin{equation}
\Psi=r\psi, \ \ \ c_{1}=\Omega^{2}\cdot\left[\frac{\pu\pv r}{r}-\frac{l(l+1)}{2r^{2}} \right],\ \ \ Y=\pu.
\label{phi}
\end{equation}
The data for $\Psi $ must satisfy the constraint equations
\begin{equation}
\partial_{v}\big(Y\Psi \big)=c_{1}\cdot\Psi , \ \ \ \partial_{v}\big(Y^{k+1}_{u}\Psi \big)=Y^{k}\big(c_{1}\cdot\Psi \big), \ k\geq 1,
\label{constraint}
\end{equation}
 By integrating along the null generators we obtain
\begin{equation}
(Y\Psi )(\tau)=(Y\Psi )(0)+\int_{0}^{1}c_{1}\cdot\Psi ,
\label{con1}
\end{equation}
and more generally
\begin{equation}
(Y^{k+1}\Psi )(\tau)=(Y^{k+1}\Psi )(0)+ \int_{0}^{1}Y^{k}(c_{1}\cdot\Psi )
\label{conk}
\end{equation}
We consider first the case $k=1$. Clearly, we want to construct $\Psi $ such that 
\begin{equation*}
\begin{split}
\int_{0}^{1}c_{1}\cdot\Psi =\a_{1}, \ \ \ 
\int_{0}^{1}\Big((Yc_{1})\cdot\Psi + c_{1}\cdot (Y\Psi )\Big)d\tilde{\tau}=\a_{2}, 
 \end{split}
\end{equation*}
where $\a_{1},\a_{2}$ are arbitrary real numbers. Using \eqref{con1}, with $\tau\mapsto\tilde{\tau}$, the above is equivalent to 
\begin{equation*}
\begin{split}
\int_{0}^{1}\bigg[(Y c_{1} )\cdot\Psi + c_{1}\cdot (Y\Psi )(0)+ c_{1}\cdot \int_{0}^{\tilde{\tau}} c_{1}\cdot \Psi   \bigg]d\tilde{\tau}=\a_{2}.
\end{split}
\end{equation*}
If $c_{1}=0$ but and $Yc_{1}=0$ then we have two conservation laws. If $c_{1}=0$ but $Yc_{1}\neq 0$ then we have a first order conservation law and gluing for second order derivatives.  If $c_{1}\neq 0$  for all  $v\in[a,b]$ then we compute
\begin{equation*}
\begin{split}
\a_{2}=&\int_{0}^{a}\bigg[(Y c_{1} )\cdot\Psi + c_{1}\cdot (Y\Psi )(0) + c_{1}\cdot \int_{0}^{\tilde{\tau}} c_{1}\cdot \Psi \bigg]d\tilde{\tau}+\int_{b}^{1}\bigg[(Y c_{1} )\cdot\Psi + c_{1}\cdot (Y\Psi )(0) \bigg]d\tilde{\tau}\\
&+\int_{b}^{1}\bigg[ c_{1}\cdot \int_{0}^{a} c_{1}\cdot \Psi + c_{1}\cdot \int_{b}^{\tilde{\tau}} c_{1}\cdot \Psi \bigg]d\tilde{\tau}\\
&+\int_{a}^{b}\bigg[ c_{1}\cdot (Y\Psi )(0)+ c_{1}\cdot \int_{0}^{a} c_{1}\cdot \Psi \bigg]d\tilde{\tau} \\
&+\int_{a}^{b}\bigg[(Y c_{1} )\cdot\Psi + c_{1}\cdot \int_{a}^{\tilde{\tau}} c_{1}\cdot \Psi  \bigg]d\tilde{\tau}+
\Bigg(\int_{a}^{b} c_{1}\cdot \Psi \Bigg)\cdot\Bigg(\int_{b}^{1}c_{1}\Bigg). 
 \end{split}
\end{equation*}
Note that the quantities in the first three lines depend only of the values of $\Psi $ in the region $\mathcal{C}=[0,a]\cup[b,1]$ and on $(Y\Psi )(0)$. We want to extend $\Psi$ everywhere in $[0,1]$ so we can do gluing. It suffices to construct $\Psi $  in $[a,b]$ such that 
\begin{equation*}
\begin{split}
\int_{a}^{b} c_{1}\cdot \Psi =\b_{1},\\
\int_{a}^{b}\bigg[(Y c_{1} )\cdot\Psi  &+ c_{1}\cdot \int_{a}^{\tilde{\tau}} c_{1}\cdot \Psi  \bigg]d\tilde{\tau}=\b_{2},
\end{split}
\end{equation*}
where $\b_{1},\b_{2}$ are given.
Define the function
\begin{equation}
\Phi_{1} :[a,b]\rightarrow\mathbb{R}: \ \ \ \Phi_{1} (t)=\int_{a}^{t} c_{1}\cdot \Psi .
\label{phidef}
\end{equation}
Note $\Phi_{1} '(t)=c_{1}(t)\cdot\Psi (t)$ and hence, since $c_{1}\neq 0$ in $[a,b]$, the function  $\Phi_{1} $ determines $\Psi $ in $[a,b]$. Therefore, it suffices to construct $\Phi_{1}$ such that $\Phi_{1} (b)=\beta_{1}$ and all derivatives of $\Phi$ at $b$ are prescribed and such that 
\begin{equation}
\begin{split}
\int_{a}^{b}\bigg[\frac{Y c_{1} }{c_{1}}\cdot \Phi_{1} '+ c_{1}\cdot \Phi_{1}  \bigg]d\tilde{\tau}=\beta_{2}
\end{split}
\label{conditions1}
\end{equation}
Now, by integration by parts we obtain
\[\int_{a}^{b}\frac{Y c_{1} }{c_{1}}\cdot \Phi_{1} '=\frac{(Y c_{1} )(b)}{c_{1}(b)}\cdot\Phi_{1} (b)-\int_{a}^{b}\bigg(\frac{Y c_{1} }{c_{1}}\Bigg)'\cdot\Phi_{1} .   \]
Therefore, if $f=-\bigg(\frac{Y c_{1} }{c_{1}}\Bigg)'$ then \eqref{conditions1} can be rewritten
\begin{equation}
\begin{split}
\int_{a}^{b}\bigg[(f+ c_{1}) \cdot\Phi_{1}  \bigg]d\tilde{\tau}=\beta_{3}.
\end{split}
\label{conditions2}
\end{equation}
Clearly, this has a solution if and only if $(f+ c_{1}) \neq 0$ at a point $\tau_{0}\in[a,b]$. Since
\begin{equation*}
\begin{split}
f+c_{1}=-\frac{(Y c_{1} )'\cdot c_{1}-c_{1}'\cdot(Y c_{1} )-c_{1}^{3}}{c_{1}^{2}}
\end{split}
\end{equation*}
we obtain that we have glue data to second order if  $c_{1}\neq 0$ and 
\begin{equation}c_{2}=
c_{1}^{3}+(\partial_{v} c_{1}) (\partial_{u} c_{1}) - c_{1}\cdot (\partial_{v}\partial_{u} c_{1}) \neq 0. 
\label{h1condition}
\end{equation}
We will next show that if $c_{2}=0$ (and $c_{1}\neq 0$) on $\hh$ then we have a second order conservation law. 

 Recall the definition \eqref{phidef} for the function $\Phi_{1} $. Then, in view of \eqref{con1} and \eqref{conk} we obtain
\[\Phi_{1} '= c_{1}\cdot \Psi ,\]
\begin{equation*}
\begin{split}
 (Y\Psi )(\tau)&=(Y\Psi )(0)+\Phi_{1} (\tau)
 \end{split}
\end{equation*}
and
\begin{equation*}
\begin{split} (YY\Psi)(\tau)&=(YY\Psi)(0)+\int_{0}^{\tau}\bigg[(Yc_{1})\cdot\Psi+c_{1}\cdot (Y\Psi)(0)+c_{1}\cdot \int_{0}^{\tilde{\tau}}c_{1}\cdot\Psi \bigg]d\tilde{\tau}\\
 &=(YY\Psi)(0)+\int_{0}^{\tau}\bigg[\frac{Yc_{1}}{c_{1}}\cdot\Phi_{1}' +c_{1}\cdot(Y\Psi)(0)+c_{1}\cdot\Phi_{1} \bigg]d\tilde{\tau}\\
  &=(YY\Psi)(0)+(Y\Psi)(0)\cdot\bigg(\int_{0}^{\tau}c_{1}\bigg) +\frac{(Yc_{1})(\tau)}{c_{1}(\tau)}\cdot\Phi_{1}(\tau)+\int_{0}^{\tau}\bigg[\bigg(-\left(\frac{Yc_{1}}{c_{1}}\right)'+c_{1}\bigg)\cdot\Phi_{1}    \bigg]d\tilde{\tau}\\
   &=(YY\Psi)(0)+(Y\Psi)(0)\cdot\bigg(\int_{0}^{\tau}c_{1}\bigg) +\frac{(Yc_{1})(\tau)}{c_{1}(\tau)}\cdot\Phi_{1}(\tau)\\
      &=(YY\Psi)(0)+(Y\Psi)(0)\cdot\bigg(\int_{0}^{\tau}c_{1}\bigg) +\frac{(Yc_{1})(\tau)}{c_{1}(\tau)}\cdot\Big((Y\Psi)(\tau)-(Y\Psi)(0)\Big)\\
&=(YY\Psi)(0)+\frac{(Yc_{1})(\tau)}{c_{1}(\tau)}\cdot\big((Y\Psi)(\tau)\big)+(Y\Psi)(0)\cdot\bigg[\bigg(\int_{0}^{\tau}c_{1}\bigg)-\frac{(Yc_{1})(\tau)}{c_{1}(\tau)}\bigg]\\
&=(YY\Psi)(0)+\frac{(Yc_{1})(\tau)}{c_{1}(\tau)}\cdot\big((Y\Psi)(\tau)\big)-(Y\Psi)(0)\cdot\frac{(Yc_{1})(0)}{c_{1}(0)},\\
  \end{split}
\end{equation*}
where we repeatedly used that $c_{1}=\left(\frac{Y c_{1} }{c_{1}}\right)'$. Therefore,  if $c_{1}\neq 0$ and $c_{1}=\left(\frac{Y c_{1} }{c_{1}}\right)'$ then the quantity
\begin{equation}
Y^{2}\Psi -\frac{Y c_{1} }{c_{1}}\cdot Y\Psi 
\label{newcon1}
\end{equation}
is conserved on $\hi$. 

The above provides a scheme in order to find necessary and sufficient conditions for the existence of higher order conservations laws. Clearly, these conservation laws are the only obstruction to gluing.

We next consider third order gluing constructions.   We have
\begin{equation*}
\begin{split}
(Y^{3}\Psi )(\tau)-(Y^{3}\Psi )(0)=&\int_{0}^{\tau}\bigg[c_{1}\cdot(Y^{2}\Psi )(0)+2\cdot(Y c_{1} )\cdot(Y\Psi )(0)+(Y^{2} c_{1}) \cdot\Psi      \bigg]d\tau_{1}\\
&+ \int_{0}^{\tau}2(Y c_{1} )\Bigg[\int_{0}^{\tau_{1}} c_{1}\cdot \Psi     \Bigg]d\tau_{1}\\
&+\int_{0}^{\tau}\Bigg[c_{1}\cdot\int_{0}^{\tau_{1}}\bigg[\bigg(
(Y c_{1} )\cdot\Psi +c_{1}\cdot(Y\Psi )(0)+c_{1}\cdot\int_{0}^{\tau_{2}}c_{1}\cdot\Psi  \Bigg)\bigg]d\tau_{2}  \Bigg]d\tau_{1}.
\end{split}
\end{equation*}
Then,
\begin{equation*}
\begin{split}
&(Y^{3}\Psi )(\tau)-(Y^{3}\Psi )(0)
\\=& \int_{0}^{\tau}\bigg[ c_{1}\cdot (Y^{2}\Psi )(0)+2\cdot(Y c_{1} )\cdot(Y\Psi )(0)+\frac{Y^{2}c_{1}}{c_{1}}\cdot (\Phi _{1})'+2(Y c_{1} )\cdot\Phi _{1} +(Y\Psi )(0)\cdot c_{1}\cdot\int_{0}^{\tau_{1}}c_{1}\bigg]d\tau_{1}
\\&+\int_{0}^{\tau}\Bigg[(Y c_{1} )\cdot\Phi _{1}+c_{1}\cdot\int_{0}^{\tau_{1}}\bigg[ \bigg(-\left(\frac{Y c_{1} }{c_{1}}\right)'+c_{1}\Bigg)\cdot\Phi _{1} \bigg]d\tau_{2}\Bigg]d\tau_{1}.\\
\end{split}
\end{equation*}
We now define the function
\[\Phi _{2}(t)=\int_{0}^{t}\bigg(-\left(\frac{Y c_{1} }{c_{1}}\right)'+c_{1}\Bigg)\cdot\Phi _{1}\]
and the function
\[c_{2}=\bigg(-\left(\frac{Y c_{1} }{c_{1}}\right)'+c_{1}\Bigg)\neq 0.\]
We thus obtain
\begin{equation*}
\begin{split}
(Y^{3}\Psi )(\tau)-(Y^{3}\Psi )(0)=&\int_{0}^{\tau}\bigg[ c_{1}\cdot (Y^{2}\Psi )(0)+2\cdot(Y c_{1} )\cdot(Y\Psi )(0)+(Y\Psi )(0)\cdot c_{1}\cdot\int_{0}^{\tau_{1}}c_{1}\bigg]d\tau_{1}
\\&+\frac{(Y^{2} c_{1}) (\tau)}{c_{1}(\tau)\cdot c_{2}(\tau)}\cdot\Phi '_{2}(\tau)
+\frac{1}{c_{2}(\tau)}\Bigg(-\left(\frac{Y^{2}c_{1}}{c_{1}}\right)'+3(Y c_{1} )\Bigg)(\tau)\cdot\Phi _{2}(\tau)\\&
+\int_{0}^{\tau}\Bigg[\Bigg(\Bigg(\frac{1}{c_{2}}\Bigg(\left(\frac{Y^{2}c_{1}}{c_{1}}\right)'-3(Y c_{1} )\Bigg)\Bigg)'+c_{1}\Bigg)\cdot\Phi _{2}\Bigg]d\tau_{1}.\\
\end{split}
\end{equation*}
The first line is completely determined by the initial data and the geometry of $\hi$. The second line is determined (by the geometry of $\hi$) and by $\Phi _{2}(\tau), \Phi '_{2}(\tau)$.

 Hence, we can arbitrarily prescribe $\Psi (\tau), Y\Psi (\tau),Y^{2}\Psi (\tau),Y^{3}\Psi (\tau)$ if and only if the function
 \begin{equation}
c_{3}=\Bigg(\Bigg(\frac{1}{c_{2}}\Bigg(\left(\frac{Y^{2}g}{g}\right)'-3(Y c_{1} )\Bigg) \Bigg) '+g\Bigg)
\label{eq:3}
\end{equation}
is non-zero at at least a point on $\hi$. If, on the other hand, we have $c_{3}=0$ everywhere on $\hi$ then the quantity
\begin{equation*}
 Y^{3}\Psi +\frac{1}{c_{2}}\cdot\left(\left(\frac{Y^{2}c_{1}}{c_{1}}\right)'-3(Y c_{1} )\right)\cdot Y^{2}\Psi +\Bigg[\frac{Y c_{1} }{c_{1}}\cdot\left[\frac{1}{c_{2}}\left(\left(\frac{Y^{2}c_{1}}{c_{1}}\right)'-3(Y c_{1} )\right)\right]+\frac{Y^{2}c_{1}}{c_{1}}\Bigg]  \cdot Y\Psi 
\label{phi3conserved}
\end{equation*}
is conserved, i.e.~independent of $\tau$. Using the above scheme, Theorem \ref{theo4} can be proved inductively.

\subsection{Some examples}
\label{sec:AnExampleOfAMetricForWhichC1Neq0AndC20}

\paragraph{1. An example of a metric for which $c_{1}\neq 0$ and $c_{2}=0$ for $l=0$\medskip\\}
\label{sec:AnExampleOfAMetricForWhichC1Neq0AndC20}

Consider a spherically symmetric metric such that 
\begin{equation}
\Omega(u,v)=1, \ \ \ r(u,v)=\frac{1}{2}+\frac{1}{2}(1+uv)^{2}>0. 
\label{examplec1c2}
\end{equation}
Let now $\hi=\left\{u=0\right\}$ along which $r=1, \partial_{u}r=v, \partial_{u}\partial_{u}r=v^{2}$. Then along $\hi$ we obtain   $c_{1}=1$ and  $c_{2}=0$. Therefore, $\hh$ admits a second-order conservation law for all spherically symmetric solutions to the wave equation. This example shows that the order of the conservation law and the angular frequency associated to its kernel are independent.

\paragraph{2. Conservation law on the Cauchy horizon of Reissner--Nordstr\"{o}m\medskip\\}
\label{sec:AnExampleOfdsaf}

The  Reissner--Norstr\"{o}m metric satisfies:
\[\partial_{u}\partial_{v}r=-\frac{\Omega^{2}}{4r}-\frac{1}{r}\partial_{v}r\partial_{u}r+\frac{1}{4}\Omega^{2}r^{-3}e^{2}. \]
On the (inner) horizon we have $\partial_{v}r=0$ and in fact 
\[r=M-\sqrt{M^{2}-e^{2}}.\]
Then, we have a first-order conservation law if:
\begin{equation}
\partial_{v}\partial_{u}r=\Omega^{2}\cdot l(l+1)\cdot\frac{1}{2r}. 
\label{malistatwra}
\end{equation}
Equation \eqref{malistatwra} is satisfied for a discrete set of values for the charge $e$. Indeed, we need $e$ to satisfy:
\[\left(\frac{e}{r}\right)^{2}=\frac{2l(l+1)+1}{2}.\]
In this case the kernel of the conservation law consists of all the eigenfunctions of the standard spherical Laplacian $\lapp_{\mathbb{S}^{2}}$ which correspond to the eigenvalue $-l(l+1)$.

\paragraph{3. A hierarchy of conservation laws for some spacetimes\medskip\\}
\label{sec:AnExampleOfdsaf}

The spherically symmetric null hypersurfaces of Minkowski spacetimes, the null infinity of asymptotically flat spacetimes and the event horizon of extremal Reissner--Norstr\"{o}m and extremal Kerr black holes admit the following hierarchy of conservation laws. Specifically, we have $R_{l+1,l}=0$ for all $l\in\mathbb{N}$. Hence, for all $l\in\mathbb{N}$ there is an $(l+1)$-order conservation law and its kernel consists of all the eigenfunctions of the standard spherical Laplacian $\lapp_{\mathbb{S}^{2}}$ which correspond to the eigenvalue $-l(l+1)$.

\section{Acknowledgements}
\label{sec:Acknowledgements}
	I would like to thank Mihalis Dafermos, Georgios Moschidis, Willie Wong and Shiwu Yang for their help and insights. I would also like to thank Harvey Reall, Sergiu Klainerman, Jan Sbierski and Jeremy Szeftel for several very stimulating discussions and comments. I acknowledge support through NSF grant  DMS-1265538.

\appendix

\section{Conservation laws and high frequency solutions}
\label{sec:ConservationLawsAndHighFrequencySolutions}

The main result of the present paper concerned characterizing the nature of the information which can be  propagated by \textbf{all} solutions to the wave equation along null hypersurfaces. We will here study the charges associated to high frequency solutions which themselves convey information along null geodesics (as the frequency tends to infinity).

\medskip

\noindent\textbf{ Relation with geometric optics approximation}

\medskip

Given  a neighborhood $\mathcal{N}$ of a null geodesic $\gamma$ one can construct (approximate) solutions to the wave equation which are of the form
\begin{equation}
\psi_{\lambda}=\frac{1}{\lambda}\cdot a\cdot e^{i\lambda k},
\label{hfs}
\end{equation}
where $a,k$ are real-valued smooth functions supported on $\mathcal{N}$ (and independent of $\lambda$). The normalization $1/\lambda$ is such that $\psi_{\lambda}$  has finite energy on a (fixed) Cauchy hypersurface. One is interested in highly oscillatory solutions that arise in the high frequency limit $\lambda\rightarrow +\infty$.  In order for $\psi_{\lambda}$ to be solution to the wave equation (in the limit $\lambda\rightarrow +\infty$), the functions $a,k$ must satisfy:
\begin{equation}
\begin{split}
dk\cdot dk=0,\\
2(\nabla k)(a)+\Box_{g}k\cdot a=0
\end{split}
\label{goapprox1}
\end{equation}
The first equation is the eikonal equation and hence $k$ must be an optical function. Therefore, the level sets of $k$ are null hypersurfaces. These hypersurfaces are generated by null geodesics which are integral curves of $\nabla k$.  Since we want to localize around the null geodesic $\gamma$ we need to take one of the integral curves of $\nabla k$ to coincide with $\gamma$. Hence, $k$ is constant along $\gamma$. The function $a$ then solves a transport equation. 
 Assume that construction for $k$ is possible for arbitrarily long time along $\gamma$. Let now $X$ be a vector field along $\gamma$. Then,
\begin{equation}
X \psi_{\lambda}= \frac{1}{\lambda}\cdot (Xa)\cdot e^{i\lambda k}+i \cdot a \cdot (Xk)\cdot e^{i\lambda k},
\label{psiapprox1}
\end{equation}
and therefore, as $\lambda\rightarrow +\infty$, the term that dominates is the one involving $Xk$. However, $\overset{\cdot}{\gamma}k=0$ along $\gamma$ and hence $\psi_{\lambda}$ oscillates in a direction transversal to $\gamma$.

Let now $\hh$ be a null hypersurface admitting conservation laws. We will investigate the conserved charges associated to $\psi_{\lambda}$ (as $\lambda\rightarrow+\infty$). Let $\s=\left\langle S_{0},L_{geod}, \Omega\right\rangle$ be a foliation of $\hh$ and $v$ be the associated optical function such that $Lv=1$, where $L=\Omega^{2}\cdot L_{geod}$. Let $u$ be the conjugate null coordinate and $Y^{\s}$ be the conjugate null vector field as defined in Section \ref{sec:ConservationLawsForTheWaveEquations} (for more details about the double null foliation see Section \ref{sec:TheDoubleNullFoliation}).  Note that 
\begin{equation}
Y^{\s}(\phi\cdot\psi_{\lambda})=\frac{1}{\lambda}\cdot Y^{\s}{\phi}\cdot a\cdot e^{i\lambda\cdot k}+\frac{1}{\lambda}\cdot \phi\cdot (Y^{\s}a)\cdot e^{i\lambda k}+i\cdot \phi\cdot a \cdot (Y^{\s}k)\cdot e^{i\lambda k}.
\label{gopequation}
\end{equation}
We distinguish the following two cases:

\smallskip

\noindent\textbf{Case I: $\gamma$ coincides with one of the null generators of $\hh$}

\smallskip

Without loss of generality, we can assume in this case that $k=u$, since $u=0$ on $\hh$. Then, 
\[Y^{\s}k=Y^{\s}u=1\]
on $\hh$, and hence the high frequency (approximate) solution $\psi_{\lambda}$ oscillates in the direction $Y^{\s}$ which coincides with the derivative appearing in the conservation law $\eqref{eq:integrals}$.  \begin{figure}[H]
   \centering
		\includegraphics[scale=0.06]{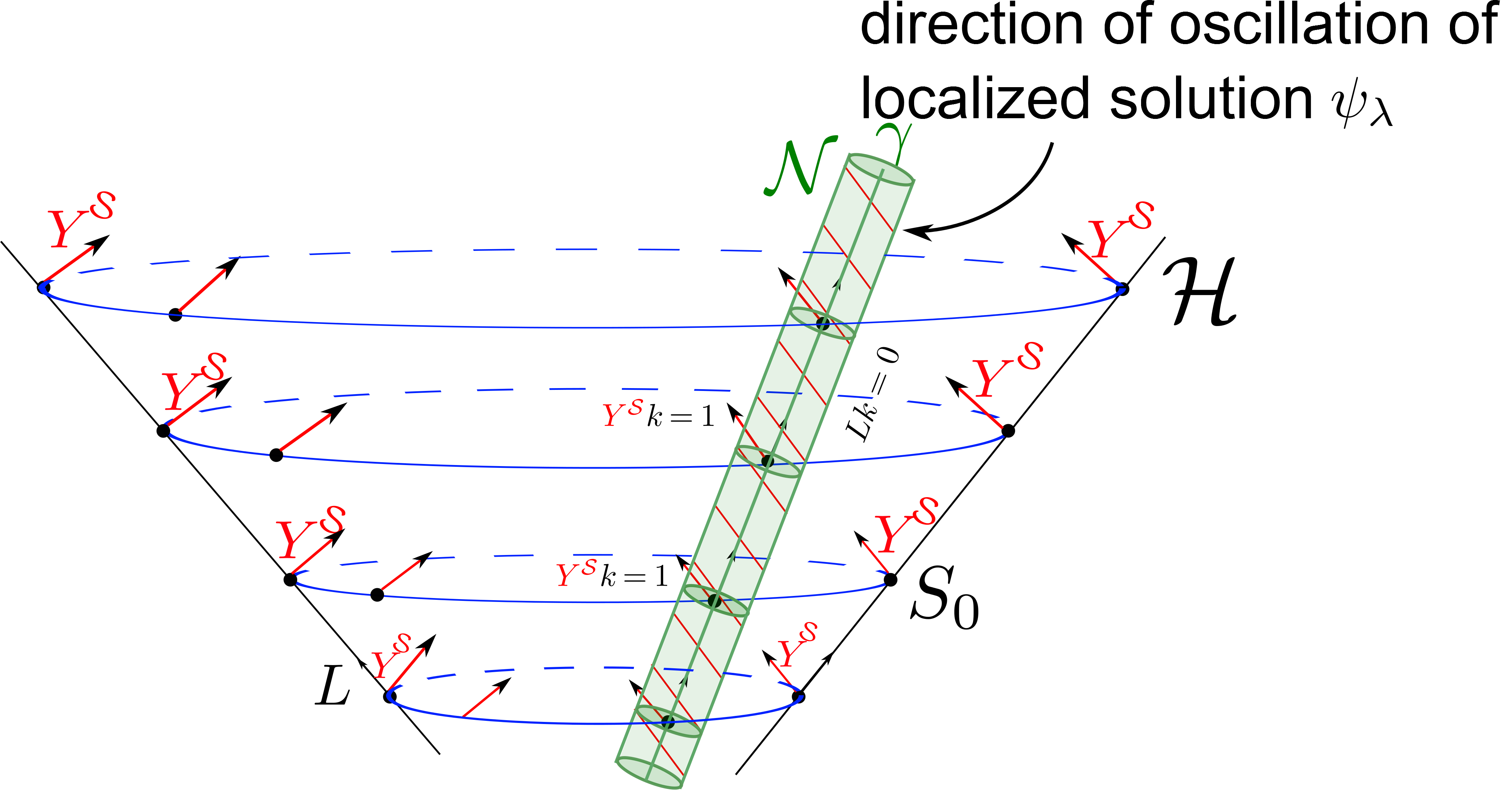}
	\label{fig:geomo1}
\end{figure}
Moreover, since $u$ is constant on $\hh$ we obtain in the limit $\lambda\rightarrow+\infty$
\[\left|char\big(S_{v}\big)[\psi_{\lambda}]\right|= \left|\int_{S_{v}} Y^{\s}\big(\phi\cdot\psi_{\lambda}\big)\cdot \Theta^{\s}\, d\mu_{_{\mathbb{S}^{2}}}\right| \sim\left|\int_{S_{v}}\phi\cdot a\cdot \Theta^{\s}\, d\mu_{_{\mathbb{S}^{2}}}\right|,  \]
where we have assumed that $\Theta^{\s}$ does not vanish along $\gamma$. Therefore, high frequency solutions localized in a neighborhood of a null generator of $\hh$ can ``carry'' arbitrarily large charges along $\hh$.  

\medskip

\noindent\textbf{Case II: $\gamma$ intersects $\hh$ transversally}

\smallskip

In this case we can take $k=v$ where $v$ is constant on conjugate null hypersurfaces $\underline{\hh}_{v}$ such that the null geodesic $\gamma$ is a null generator of one of them, say $\underline{\hh}_{0}$. Note that the optical function $v$ gives rise to a foliation  $\s=\big(S_{v}\big)_{v\in\mathbb{R}}$ of $\hh$. Since $\hh$ admits a conservation law, according to the main result, $\hh$ must admit a conservation law with respect to the foliation $\s$.  Let $S_{0}=\hh\cap \underline{\hh}_{0}$. We have
\[Lk=Lv=1\]
on $S_{0}$ and so $\psi_{\lambda}$ oscillates in the direction of the null generators of $\hh$. Moreover,
\[Y^{\s}k=0\]on $S_{0}$. 
\begin{figure}[H]
   \centering
		\includegraphics[scale=0.085]{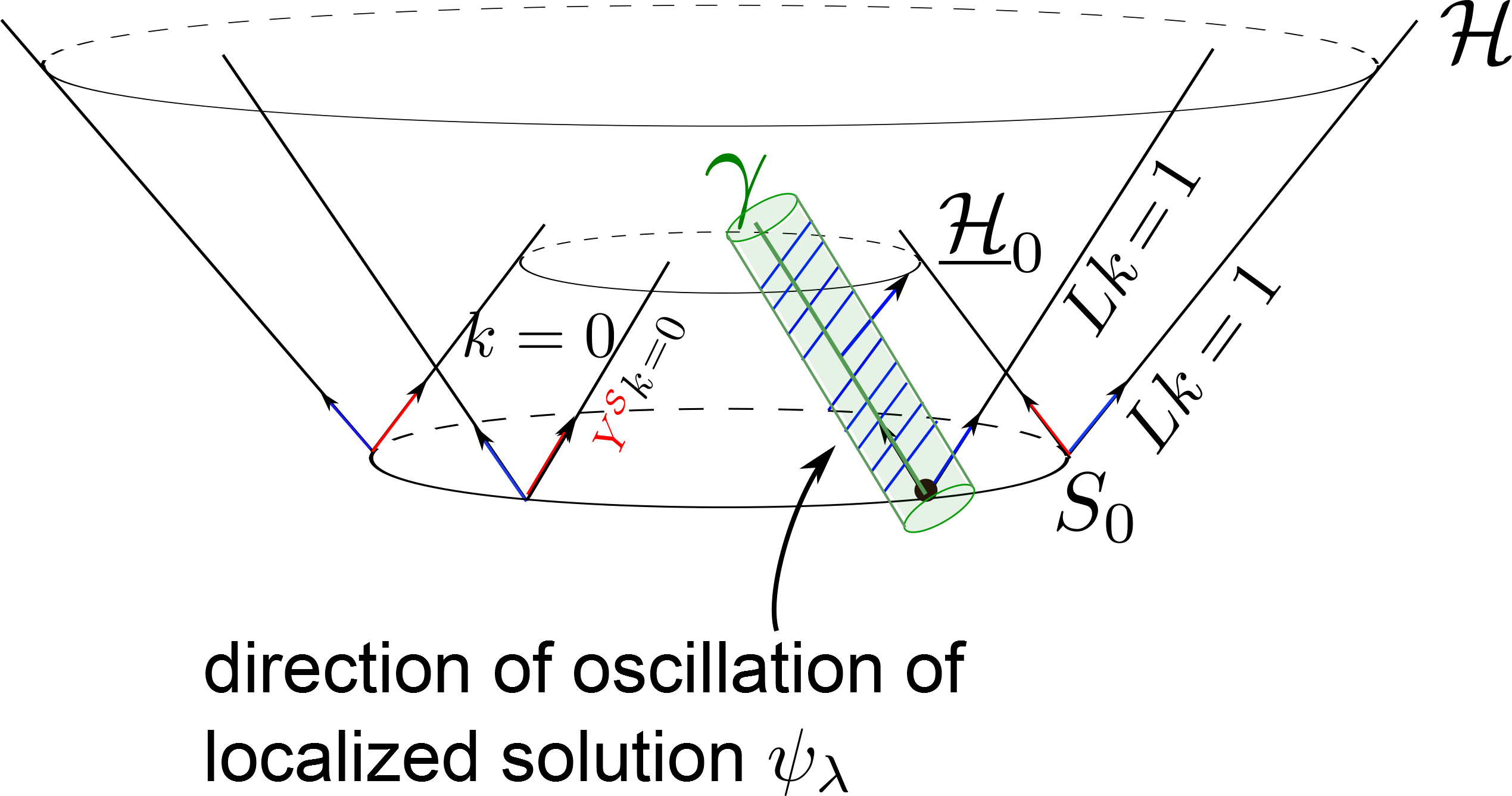}
	\label{fig:geomo2}
\end{figure}
Therefore, in view of \eqref{gopequation}, high frequency solutions localized in a neighborhood of a null geodesic $\gamma$ which intersects $\hh$ transversally have vanishing charges.

\begin{remark}
\textbf{Non-existence of conservation laws on timelike hypersurfaces. }The wave equation clearly does not admit conservation laws on spacelike hypersurfaces. 
 Using high frequency solutions we can also rule out the existence of conservation laws on timelike hypersurfaces. Indeed, let $\mathcal{T}=\big(S_{t}\big)_{t\in\mathbb{R}}$ be a timelike hypersurface which admits conserved charges with respect to a foliation with 2-spheres $S_{t}$.  Let $\mathcal{H}_{t}$ denote the null hypersurface  generated by null geodesics normal to $S_{t}$ and $\underline{\mathcal{H}}_{t}$ denote the conjugate null hypersurface generated by conjugate null geodesics  normal to $S_{t}$. Let also $L,Y$ be tangential to the null generators of $\hh_{t},\underline{\hh}_{t}$, respectively.  We define the optical functions $u,v$ such that their level sets are given by
\[ \left\{u=t\right\}=\hh_{t},\  \ \ \ \left\{v=t\right\}=\underline{\hh}_{t}. \]  
Let us first assume that there is a point $p\in S_{t}$ such that the derivative $N$ involved at the charge over $S_{t}$ at the point $p$ is distinct from $L$. We consider the high frequency solution $\psi_{\lambda}$ for which $k=u$ and $a$ is localized in a sufficiently small neighborhood of the point $p$. Then $\psi_{\lambda}$ is localized around a neighborhood of the null generator of $\hh_{t}$ emanating from the point $p$.\begin{figure}[H]
   \centering
		\includegraphics[scale=0.065]{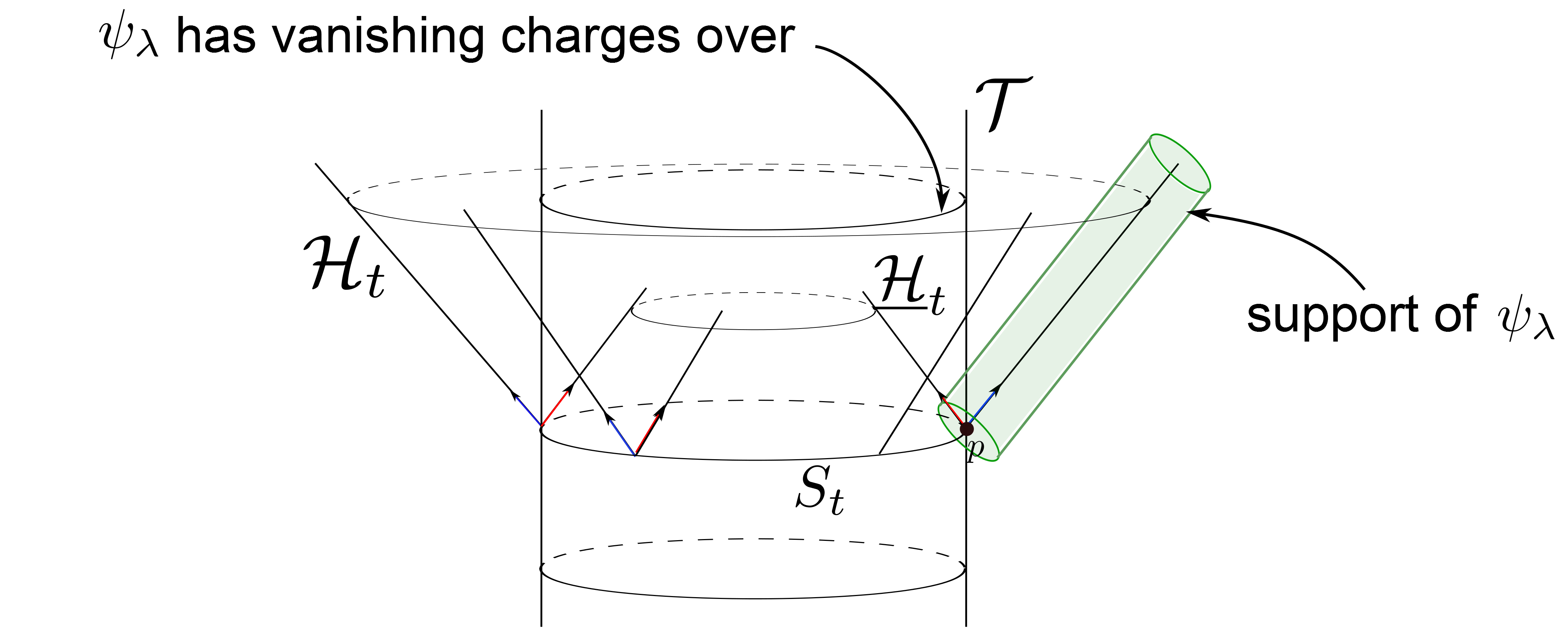}
	\label{fig:geomo3pci}
\end{figure} 
In view of \eqref{psiapprox1}, as $\lambda\rightarrow+\infty$ we obtain 
\begin{equation}
N\psi_{\lambda}\sim i\cdot a\cdot (Nu)\cdot e^{i\lambda u}.
\label{nonexistenceremark}
\end{equation}
at the point $p$. By our assumption we have $Nu\neq 0$ at $p$, and hence there is a sufficiently small neighborhood $\mathcal{N}_{p}$ of $p$ such that $Nu>\epsilon>0$ at $\mathcal{N}_{p}$.  Note also that since $u$ is constant on $S_{t}$ the term $e^{i\lambda u}$ does not interact with the integral over $S_{t}$. Therefore, we can choose $a$ so $\left.a\right|_{S_{t}}$ is supported in an sufficiently small neighborhood of $p$ and such that the charge over $S_{t}$ of  $\psi_{\lambda}$ is arbitrarily large.
This contradicts the fact that the charges of $\psi_{\lambda}$ over sections of $\mathcal{T}$ in the future of $S_{t}$ are necessarily zero. Therefore, the conserved charges must necessarily involve only the $L$ derivative (since it is only this derivative for which $L\psi_{\lambda}=0$). However,  if we consider solutions localized in a neighborhood of the null generators of $\underline{\hh}_{t}$ then we are again led to contradiction. 

\label{remarktimelike}
\end{remark}

\medskip

\bigskip

\noindent\textbf{ Relation of conservation laws and Gaussian beams} 

\medskip

The eikonal equation (and hence the geometric optics approximation) breaks down if caustics form. In particular, the above constructions are not valid if there are conjugate points along $\gamma$. However, one can still construct high frequency solutions $\psi_{\lambda}$ to the wave equation which at the limit $\lambda\rightarrow+\infty$ are supported only on given null geodesic $\gamma$. These solutions are of the form
\begin{equation}
\psi_{\lambda}=\frac{1}{\lambda^{\frac{1}{4}}}\cdot a\cdot e^{i\lambda \tau}
\label{gaubeam1}
\end{equation}
where $a$, and $\tau$ are complex valued functions. In order for $\Box_{g}\psi_{\lambda}$ to be small (in the $L^{2}$ sense) one still requires $d\tau\cdot d\tau$ to vanish to (at least) third order on $\gamma$ (and $a$ to satisfy a transport equation along $\gamma$). The main idea however is to additionally assume that $Im(\tau)\geq C\cdot (x^{1}_{2}+x_{2}^{2}+x_{3}^{2})$ where $(x_{0},x_{1},x_{2},x_{3})$ is a coordinate system covering a neighborhood of $\gamma$ such that $\gamma=\left\{x_{1}=0,x_{2}=0, x_{3}=0\right\}$. This construction is known as the \textit{Gaussian beams} method (see \cite{ralston2}). The normalization factor $\lambda^{1/4}$ is introduced so that $\psi_{\lambda}$ has finite energy on a fixed Cauchy hypersurface. 

A slightly more involved calculation than above  shows that the charges associated to Gaussian beams are zero regardless of whether $\gamma$ is a null generator of $\hh$ or not. On the other hand,  Sbierski \cite{janpaper} was able to show  that the \textit{energy} of Gaussian beams localized on a null generator $\gamma$ of a degenerate horizon\footnote{which by Theorem \ref{the2} admits a unique conservation law.} does \textit{not} decay (even though these solutions carry vanishing charges). This in particular implies that one cannot show a uniform local integrated energy decay estimate without degeneracy or loss of derivatives in a region containing the Killing horizon. Optimal such (degenerate) estimates were derived in \cite{aretakis1,aretakis3}.

\bibliographystyle{acm}
\bibliography{../../../../bibliography}

\begin{thebibliography}{10}

\bibitem{blukerr}
{\sc Andersson, L., and Blue, P.}
\newblock Hidden symmetries and decay for the wave equation on the {K}err
  spacetime.
\newblock {\em arXiv:0908.2265\/}.

\bibitem{andersson}
{\sc Andersson, L., Mars, M., and Simon, W.}
\newblock Stability of marginally outer trapped surfaces and existence of
  marginally outer trapped tubes.
\newblock {\em Advances in Theoretical and Mathematical Physics 12\/} (2008),
  853--888.

\bibitem{aretakis1}
{\sc Aretakis, S.}
\newblock Stability and instability of extreme {R}eissner--{N}ordstr\"om black
  hole spacetimes for linear scalar perturbations {I}.
\newblock {\em Commun. Math. Phys. 307\/} (2011), 17--63.

\bibitem{aretakis2}
{\sc Aretakis, S.}
\newblock Stability and instability of extreme {R}eissner--{N}ordstr\"om black
  hole spacetimes for linear scalar perturbations {II}.
\newblock {\em Ann. Henri Poincar\'{e} 12\/} (2011), 1491--1538.

\bibitem{aretakis3}
{\sc Aretakis, S.}
\newblock Decay of axisymmetric solutions of the wave equation on extreme
  {K}err backgrounds.
\newblock {\em J. Funct. Analysis 263\/} (2012), 2770--2831.

\bibitem{aretakis4}
{\sc Aretakis, S.}
\newblock Horizon instability of extremal black holes.
\newblock {\em ar{X}iv:1206.6598\/} (2012).

\bibitem{aretakiselliptic}
{\sc Aretakis, S.}
\newblock On a foliation-covariant elliptic operator on null hypersurfaces.
\newblock {\em preprint\/} (2013).

\bibitem{aretakis2013}
{\sc Aretakis, S.}
\newblock On a non-linear instability of extremal black holes.
\newblock {\em Phys. Rev. D 87\/} (2013), 084052.

\bibitem{bizon2012}
{\sc Bizon, P., and Friedrich, H.}
\newblock A remark about the wave equations on the extreme
  {R}eissner--{N}ordstr\"om black hole exterior.
\newblock {\em Class. Quantum Grav. 30\/} (2013), 065001.

\bibitem{memorychistodoulou}
{\sc Christodoulou, D.}
\newblock Nonlinear nature of gravitation and gravitational-wave experiments.
\newblock {\em Phys. Rev. Lett. 67\/} (1991), 1486--1489.

\bibitem{DC09}
{\sc Christodoulou, D.}
\newblock {\em The formation of black holes in general relativity}.
\newblock European Mathematical Society Publishing House, 2009.

\bibitem{christab}
{\sc Christodoulou, D., and Klainerman, S.}
\newblock {\em The Global Nonlinear Stability of the {M}inkowski Space}.
\newblock Princeton University Press, 1994.

\bibitem{chrugrav}
{\sc Chru\'{s}ciel, P.~T., Mac{C}allum, M. A.~H., and Singleton, D.~B.}
\newblock Gravitational waves in general relativity {XIV}. {B}ondi expansions
  and the ``polyhomogeneity'' of {S}cri.
\newblock {\em Phil. Trans. R. Soc. Lond. A. 350\/} (1995), 113.

\bibitem{enadio}
{\sc Dafermos, M., and Rodnianski, I.}
\newblock Decay for solutions of the wave equation on {K}err exterior
  spacetimes {I-II}: The cases $|a|\ll m$ or axisymmetry.
\newblock {\em ar{X}iv:1010.5132\/} (2010).

\bibitem{tria}
{\sc Dafermos, M., and Rodnianski, I.}
\newblock The black hole stability problem for linear scalar perturbations.
\newblock {\em Proceedings of the 12 Marcel Grossmann Meeting, edited by T.
  Damour et al (ed.), World Scientific, Singapore\/} (2011), 132--189,
  ar{X}iv:1010.5137.

\bibitem{lecturesMD}
{\sc Dafermos, M., and Rodnianski, I.}
\newblock Lectures on black holes and linear waves.
\newblock {\em in Evolution equations, {C}lay {M}athematics {P}roceedings,
  {V}ol. 17, Amer. Math. Soc., Providence, RI,\/} (2013), 97--205,
  ar{X}iv:0811.0354.

\bibitem{dd2012}
{\sc Dain, S., and Dotti, G.}
\newblock The wave equation on the extreme {R}eissner--{N}ordstr\"om black
  hole.
\newblock {\em ar{X}iv:1209.0213\/} (2012).

\bibitem{evans}
{\sc Evans, L.~C.}
\newblock {\em Partial Differential Equations}.
\newblock Graduate Studies in Mathematics, 1998.

\bibitem{npexton}
{\sc Exton, A.~R., Newman, E.~T., and Penrose, R.}
\newblock Conserved quantities in the {E}instein-{M}axwell theory.
\newblock {\em J. Math. Phys. 10\/} (1969), 1566--1570.

\bibitem{goldberg1}
{\sc Goldberg, J.~N.}
\newblock Invariant transformations and {N}ewman-{P}enrose constants.
\newblock {\em J. Math. Phys. 8\/} (1967), 2161--2166.

\bibitem{goldberg2}
{\sc Goldberg, J.~N.}
\newblock Green's theorem and invariant tranformations.
\newblock {\em J. Math. Phys. 9\/} (1968), 674--679.

\bibitem{goldberg3}
{\sc Goldberg, J.~N.}
\newblock Conservation of the {N}ewman--{P}enrose conserved quantities.
\newblock {\em Phys. Rev. Lett. 28\/} (1972), 1400.

\bibitem{haw}
{\sc Hawking, S., and { G.F.R. Ellis}}.
\newblock {\em The large scale structure of spacetime}.
\newblock Cambridge University Press, 1973.

\bibitem{extremumproblemsbook}
{\sc Henrot, A.}
\newblock {\em Extremum problems for eigenvalues of elliptic operators}.
\newblock Birkh\"auser, 2006.

\bibitem{kato}
{\sc Kato, T.}
\newblock {\em Perturbation theory for linear operators}.
\newblock Springer, 1995.

\bibitem{SK86}
{\sc Klainerman, S.}
\newblock The null condition and global existence to nonlinear wave equations.
\newblock {\em Lect. Appl. Math. 23\/} (1986), 293--326.

\bibitem{valientenp1}
{\sc Kroon, J. A.~V.}
\newblock Conserved quantities for polyhomogeneous spacetimes.
\newblock {\em Class. Quantum Grav. 15\/} (1998), 2479.

\bibitem{valientenp2}
{\sc Kroon, J. A.~V.}
\newblock Logarithmic {N}ewman--{P}enrose constants for arbitrary
  polyhomogeneous spacetimes.
\newblock {\em Class. Quantum Grav. 16\/} (1999), 1653.

\bibitem{valientenp3}
{\sc Kroon, J. A.~V.}
\newblock On {K}illing vector fields and {N}ewman--{P}enrose constants.
\newblock {\em J. Math. Phys. 41\/} (2000), 898.

\bibitem{hm2012}
{\sc Lucietti, J., Murata, K., Reall, H.~S., and Tanahashi, N.}
\newblock On the horizon instability of an extreme {R}eissner--{N}ordstr\"om
  black hole.
\newblock {\em JHEP 1303\/} (2013), 035, arXiv:1212.2557.

\bibitem{hj2012}
{\sc Lucietti, J., and Reall, H.}
\newblock Gravitational instability of an extreme {K}err black hole.
\newblock {\em Phys. Rev. D86:104030\/} (2012).

\bibitem{murata2012}
{\sc Murata, K.}
\newblock Instability of higher dimensional extreme black holes.
\newblock {\em Class. Quantum Grav. 30\/} (2013), 075002.

\bibitem{harvey2013}
{\sc Murata, K., Reall, H.~S., and Tanahashi, N.}
\newblock What happens at the horizon(s) of an extreme black hole?
\newblock {\em arXiv:1307.6800\/} (2013).

\bibitem{np1}
{\sc Newman, E.~T., and Penrose, R.}
\newblock 10 exact gravitationally conserved quantities.
\newblock {\em Phys. Rev. Lett. 15\/} (1965), 231.

\bibitem{np2}
{\sc Newman, E.~T., and Penrose, R.}
\newblock New conservation laws for zero rest mass fields in asympotically flat
  space-time.
\newblock {\em Proc. R. Soc. A 305\/} (1968), 175204.

\bibitem{ori2013}
{\sc Ori, A.}
\newblock Late-time tails in extremal {R}eissner-{N}ordstr\"{o}m spacetime.
\newblock {\em arXiv:1305.1564\/} (2013).

\bibitem{pressnp}
{\sc Press, W.~H., and Bardeen, J.~M.}
\newblock Non-conservation of the {N}ewman--{P}enrose conserved quantities.
\newblock {\em Phys. Rev. Lett. 27\/} (1971), 1303.

\bibitem{ralston2}
{\sc Ralston, J.}
\newblock Gaussian beams and the propagation of singularities.
\newblock {\em Studies in Partial Differential Equations, MAA Studies in
  Mathematics 23\/} (1983), 206--248.

\bibitem{robinson}
{\sc Robinson, D.~C.}
\newblock Conserved quantities of {N}ewman and {P}enrose.
\newblock {\em J. Math. Phys. 9\/} (1969), 1745--1753.

\bibitem{janpaper}
{\sc Sbierski, J.}
\newblock Characterisation of the energy of {G}aussian beams on {L}orentzian
  manifolds with applications to black hole spacetimes.
\newblock {\em preprint\/} (2013).

\bibitem{tataru2}
{\sc Tataru, D., and Tohaneanu, M.}
\newblock A local energy estimate on {K}err black hole backgrounds.
\newblock {\em Int. Math. Res. Not. 2011\/} (2008), 248--292.

\bibitem{wald}
{\sc Wald, R.~M.}
\newblock {\em General Relativity}.
\newblock The University of Chicago Press, 1984.

\end{thebibliography}

\end{document}